\documentclass[12pt]{article}
\usepackage[utf8]{inputenc}
\usepackage[sectionbib]{natbib}

\title{\vspace{-19mm}High-Dimensional Low-Rank Tensor Autoregressive Time Series Modeling}
\author{Di Wang${}^{a}$, Yao Zheng${}^{b,}$\footnote{Correspondence to: Department of Statistics, University of Connecticut,  215 Glenbrook Road, Storrs, CT 06269, United States of America. Email address: yao.zheng@uconn.edu (Y. Zheng).} 
	and Guodong Li${}^{c}$}

\usepackage[margin=1in]{geometry}
\usepackage{graphicx}
\usepackage{mathrsfs}
\usepackage{multirow}
\usepackage{amsfonts}
\usepackage{amsmath}
\usepackage{comment}
\usepackage{amsthm}
\usepackage{color}
\usepackage{mathtools}
\usepackage{algorithm}
\usepackage{sectsty}
\usepackage[colorlinks=true,linkcolor=blue,urlcolor=black,citecolor=blue,bookmarksopen=true]{hyperref}
\usepackage{bookmark}
\usepackage{amssymb}
\usepackage[toc,page]{appendix}
\usepackage[mathscr]{euscript}
\usepackage[toc]{appendix}

\usepackage{chngcntr}
\usepackage{apptools}
\usepackage{booktabs}
\usepackage{lscape}
\usepackage{soul}
\usepackage{epstopdf}
\usepackage{array}
\usepackage{diagbox}

\newcolumntype{L}[1]{>{\raggedright\let\newline\\\arraybackslash\hspace{0pt}}m{#1}}
\newcolumntype{C}[1]{>{\centering\let\newline\\\arraybackslash\hspace{0pt}}m{#1}}
\newcolumntype{R}[1]{>{\raggedleft\let\newline\\\arraybackslash\hspace{0pt}}m{#1}}

\makeatletter
\newcommand*{\addFileDependency}[1]{
  \typeout{(#1)}
  \@addtofilelist{#1}
  \IfFileExists{#1}{}{\typeout{No file #1.}}
}
\makeatother

\renewcommand{\arraystretch}{1.3}

\newtheorem{assumption}{Assumption}
\newtheorem{definition}{Definition}
\newtheorem{example}{Example}
\newtheorem{lemma}{Lemma}

\newtheorem{theorem}{Theorem}

\newtheorem{remark}{Remark}

\hypersetup{
    colorlinks=true,
    linkcolor=blue,
    citecolor=blue,
    filecolor=magenta,
    urlcolor=cyan,
}

\DeclareMathOperator*{\argmin}{arg\,min}

\newcommand{\bm}{\mathbf}
\newcommand{\bbm}{\boldsymbol}

\newcommand{\cm}[1]{\mbox{\boldmath$\mathscr{#1}$}}

\date{}

\mathtoolsset{showonlyrefs}

\begin{document}

\setlength{\parindent}{16pt}
\linespread{1.4}
\maketitle
\vspace{-1cm}
\noindent
${}^{a}$\emph{\small School of Mathematical Sciences, Shanghai Jiao Tong University, China}\\
${}^{b}$\emph{\small Department of Statistics, University of Connecticut, United States of America}\\
${}^{c}$\emph{\small Department of Statistics and Actuarial Science, University of Hong Kong, China}
\vspace{-3mm}

\begin{abstract}
Modern technological advances have  enabled an unprecedented amount of structured data with complex temporal dependence, urging the need for new methods to efficiently model and forecast high-dimensional tensor-valued time series. This paper provides a new modeling framework to accomplish this task via autoregression (AR).  By considering a low-rank Tucker decomposition for the transition tensor, the proposed tensor AR can flexibly capture the underlying low-dimensional tensor dynamics, providing both substantial dimension reduction and meaningful multi-dimensional dynamic factor interpretations.  
For this model, we first study several nuclear-norm-regularized estimation methods  and derive their non-asymptotic properties under the approximate low-rank setting. In particular, by leveraging the special balanced structure of the transition tensor, a novel convex regularization approach  based on the sum of nuclear norms of square matricizations  is proposed to efficiently encourage low-rankness of the coefficient tensor. To further improve the estimation efficiency under exact low-rankness, a non-convex estimator is proposed with a gradient descent algorithm, and its computational and statistical convergence guarantees are established. Simulation studies and an empirical analysis of  tensor-valued time series data from multi-category import-export networks demonstrate the advantages of the proposed approach.
  
\end{abstract}


\linespread{1.5}
\newpage
\section{Introduction}\label{sec:intro}

The rapid improvement in data collection capability has enabled the generation of increasingly more comprehensive economic datasets. Meanwhile, significant progress has been made in unifying data collection standards. These advances  have led to an abundance of comparable disaggregated time series datasets across countries, which are further categorized by various dimensions like regions, industries, goods categories, and demographics.
Such multidimensional datasets can often be organized as multi-way arrays, forming tensor-valued time series.  Moreover, this type of detailed time series data is common in finance, where it can be formed, e.g., by categorizing stock returns based on various firm characteristics dimensions, or asset returns across asset classes, regions, and sectors. 
The availability of extensive disaggregated data, in turn, provides new opportunities to advance techniques for modeling  complex dynamic systems like the global economy and financial markets. 

The motivation behind the study of tensor-valued time series stems from the modeling of temporal and cross-sectional dependencies in panel data. To illustrate, first consider the panel data $\bm{y}_t=(\bm{y}_{1,t}^\top, \dots, \bm{y}_{N,t}^\top)^\top$ for $N$ countries, where $1\leq t\leq T$ represents time, and $\bm{y}_{i,t}$ is a vector containing observations of different economic variables for country $i$ with $1\leq i\leq N$. For example,
 \cite{BCS12} models the international trade by fitting a Global Vector Autoregressive (GVAR) model \citep{PSW04} to $\bm{y}_{i,t}$ which includes the aggregate export and import volumes of country $i$, $ex_{i,t}$ and $im_{i,t}$, as key variables. Compared to previous methods, \cite{BCS12}'s approach has two major strengths: (1) it captures interdependencies across countries, i.e. cross-country spillovers; (2) it jointly models exports and imports, allowing for co-movements between them, which is important as exporting firms typically import components.

However, aggregate-level exports and imports data are limited in providing a comprehensive understanding of global trade dynamics. They do not provide information about how exports from one country are distributed among importing countries. By contrast, a much more detailed perspective on the trade flows can be gained from the disaggregated data $\bm{Y}_t=(ex_{i,j,t})_{1\leq i,j\leq N}$, where $ex_{i,j,t}$ represents exports from one specific country (country $i$) to another specific country (country $j$) for all possible pairs of countries. Note that $\bm{Y}_t$ is an $N\times N$ matrix-valued time series, and by convention $ex_{i,i,t}$'s are set to zero.  Since $ex_{i,t}=\sum_{j\neq i} ex_{i,j,t}$ and $im_{i,t}=\sum_{j\neq i} ex_{j,i,t}$, the disaggregated series $\bm{Y}_t$ contains all information of the aggregate series. Furthermore, we can have an even more granular view by further breaking $\bm{Y}_t$ down into additional dimensions. For instance, the data can be divided into $K$ different product categories, resulting in the tensor-valued time series $\cm{Y}_t=(ex_{i,j,k,t})_{1\leq i,j\leq N, 1\leq k\leq K}$, where $ex_{i,j,k,t}$ represents  exports of product category $k$ from country $i$ to country $j$; see  Figure \ref{fig:illustration} for an illustration.

\begin{figure}
\includegraphics[width=0.95\textwidth]{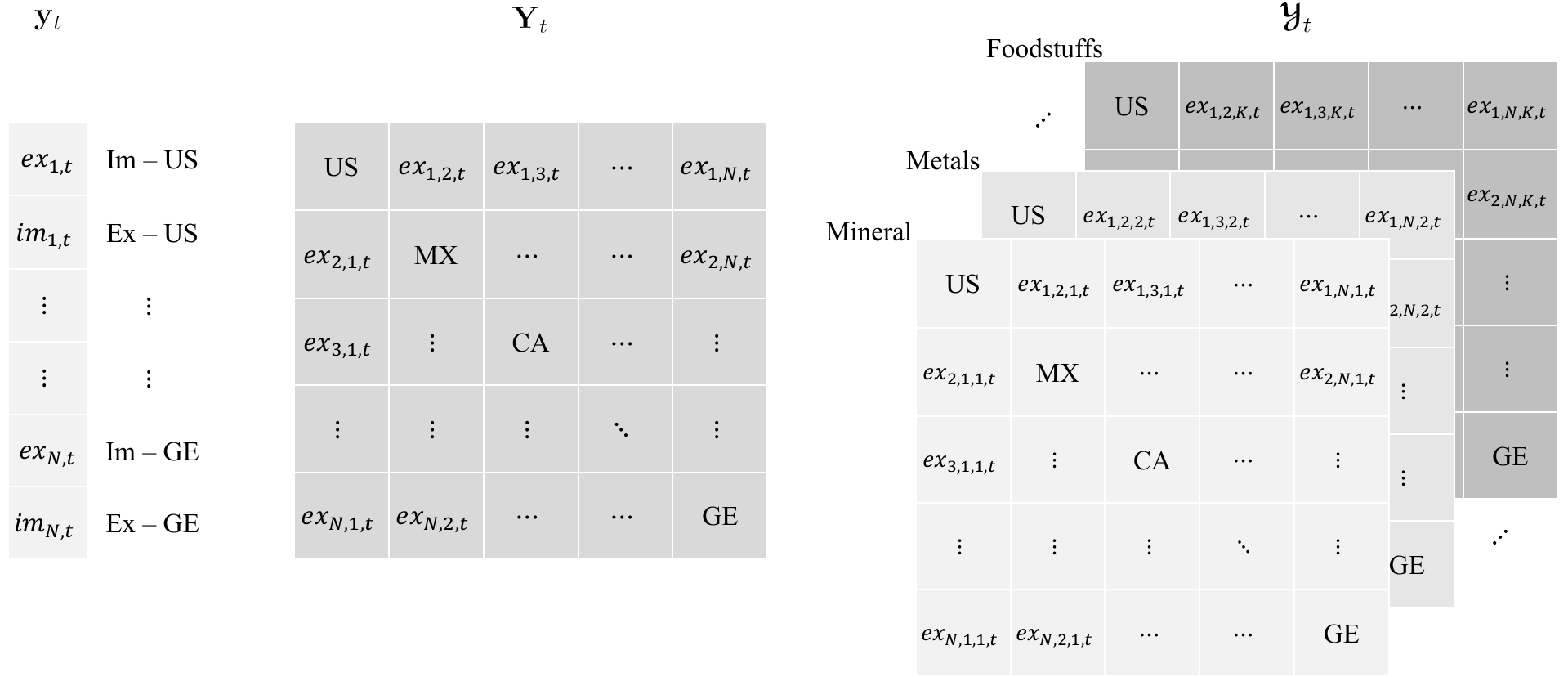}
\caption{Illustration of vector-, matrix- and tensor-valued time series  $\bm{y}_t=(ex_{1,t}, im_{1,t}, \dots, ex_{N,t}, im_{N,t})^\top$,  $\bm{Y}_t=(ex_{i,j,t})_{1\leq i,j\leq N}$, and $\cm{Y}_t=(ex_{i,j,k,t})_{1\leq i,j\leq N, 1\leq k\leq K}$ in the context of modeling international trade dynamics via import-export data, where $N$ is the number of countries, and $K$ is the number of product categories.}
\label{fig:illustration}
\end{figure}

More broadly, this paper considers autoregressive modeling of a general tensor-valued time series  $\cm{Y}_t\in\mathbb{R}^{p_1\times\cdots\times p_d}$, where the total number of series $p:=\prod_{i=1}^dp_i$ can be much larger than $T$.  A na\"{\i}ve method is to apply models for panel data to the vectorized series $\bm{y}_t=\text{vec}(\cm{Y}_t)$, such as the vector autoregressive (VAR) model,
\begin{equation}
    \label{eq:VAR}
    \text{vec}(\cm{Y}_t)=\bm{A}\text{vec}(\cm{Y}_{t-1})+\text{vec}(\cm{E}_t),
\end{equation}
and then perform dimension reduction for  the unknown transition matrix $\bm{A}\in\mathbb{R}^{p\times p}$ via generic regularization methods such as the Lasso \citep{basu2015regularized, han2015direct}, or  data-specific methods \citep{PSW04, CC13, Guo2016, Zhu2017, zheng20}  which impose parameter restrictions based on  pre-determined network structures. However, the vectorization undermines the model interpretability that could have been a valuable advantage of multi-dimensional data. For example, it would be much easier to gain meaningful insights into the global trade flow from the multi-category import-export data $\cm{Y}_t\in\mathbb{R}^{N\times N\times K}$ mentioned above if patterns across exporting countries,  importing countries, and product categories can be separately interpreted. In particular, adopting a multi-dimensional approach, as proposed in this paper, enables us to address the following questions, which cannot be answered using the vector model:
\begin{itemize}
\item[(i)] Among all countries, whose exporting activities are the driving forces of the global trade flow? Are there any geographical groupings among them? 
\item[(ii)] Similar to (i), what about the importing activities?
\item[(iii)] Among all product categories, which ones are the driving forces of the global trade flow? Are there any grouping patterns?
\item[(iv)] Do the past and present states of the dynamic system (i.e., predictor and response) have the same grouping patterns across exporting countries, importing countries, and product categories?
\end{itemize}

Specifically, for the tensor-valued time series  $\cm{Y}_t\in\mathbb{R}^{p_1\times\cdots\times p_d}$, this paper proposes the Low-Rank Tensor Autoregressive (LRTAR) model by folding the $p\times p$ transition matrix $\bm{A}$ in \eqref{eq:VAR}, with  $p=\prod_{i=1}^dp_i$, into the $2d$-th-order transition tensor $\cm{A}\in\mathbb{R}^{p_1\times\cdots\times p_d\times p_1\times\cdots\times p_d}$ which is assumed to have Tucker ranks $(r_1,\dots, r_{2d})$ with $r_i$ being possibly much smaller than $p_i$, where $p_{d+i}=p_i$ for $i=1,\dots, d$. This implies the Tucker decomposition $\cm{A}=\cm{G}\times_{i=1}^{2d}\bm{U}_i$, where $\bm{U}_i\in\mathbb{R}^{p_i\times r_i}$ and $\cm{G}\in\mathbb{R}^{r_1\times\cdots\times r_{2d}}$, and consequently the  low-dimensional structure of the process $\cm{Y}_t$ as follows:
\begin{equation*}
    \cm{Y}_t\times_{i=d+1}^{2d}\bm{U}_{i}^\top =\langle\cm{G},\cm{Y}_{t-1} \times_{i=1}^d\bm{U}_i^\top\rangle+\cm{E}_t\times_{i=d+1}^{2d}\bm{U}_{i}^\top,
\end{equation*}
where $\cm{Y}_t\times_{i=d+1}^{2d}\bm{U}_{i}^\top$ and $\cm{Y}_{t-1} \times_{i=1}^d\bm{U}_i^\top$ can be viewed as $r_{d+1}\times\cdots\times r_{2d}$ and $r_{1}\times\cdots\times r_{d}$  factors summarizing the dynamic information across all dimensions. Moreover, each loading matrix $\bm{U}_i$ reveals interpretable patterns for a particular dimension of the present or past state of $\cm{Y}_t$; see Section  \ref{sec:example} for more detailed  descriptions in the context of import-export data. The proposed model has the following features:
\begin{itemize}
\item Similar to panel data models, it captures both cross-sectional and temporal dependencies. However, by leveraging the tensor structure, it dissects the cross-sectional information into $d$ different dimensions, allowing for separate interpretations in each dimension.
\item Simultaneous dimension reduction is achieved across all dimensions of the transition tensor $\cm{A}$ via the low-Tucker-rank assumption. This approach does not rely on predetermined parameter restrictions derived from the user's prior knowledge or beliefs about the network structure. 
\item The low-Tucker-rank assumption implies that factors are extracted across all dimensions of the response and its lagged predictor. The factor loadings facilitate the discernment of patterns in each dimension (i.e., mode) of the tensor-valued observation.
\end{itemize}

For the proposed model, this paper introduces two types of high-dimensional estimation methods: (i) convex estimators via nuclear norm regularizations and (ii) the non-convex estimator.  
For (i), we consider the general setting where the transition tensor $\cm{A}$ is approximately low-rank, and  develop convex estimation methods based on different nuclear norm regularizations. Firstly, to encourage low-rankness along all modes, we study the widely-used Sum of Nuclear (SN) norm regularizer, defined as the sum of nuclear norms of all one-mode matricizations. However, due to the \textit{fat-and-short} shape of the one-mode matricizations, the SN regularized estimator suffers from serious efficiency loss and hence performs even worse than the conventional Matrix Nuclear (MN) norm regularized estimator \citep{negahban2011estimation} which simply penalizes the nuclear norm of the transition matrix $\bm{A}$ in \eqref{eq:VAR}. Thus, we further introduce a novel Sum of Square-matrix Nuclear (SSN) norm regularizer, defined as the sum of nuclear norms of all   $p\times p$ square matricizations of $\cm{A}$. The SSN reguarlized estimator is provably more efficient than the SN regularized one; see Theorem \ref{thm:SSN} and the first simulation experiment in Section \ref{sec:sim}. In addition, we propose a truncated variants of the SSN estimator and prove its rank selection consistency when $\cm{A}$ is exactly low-Tucker-rank under  mild conditions.

However, the consistency of  the SSN estimator requires that $T$  grows faster than $p=\prod_{i=1}^dp_i$. Thus, it may not be applicable to high-dimensional tensor-valued time series datasets with large $p_i$'s. This motivates us to consider a  non-convex estimation method to further improve the estimation efficiency and relax the sample size requirement. Specifically, under the assumption that $\cm{A}$ is exactly low-rank,   this paper develops an estimator via non-convex (NC) optimization based on the explicit Tucker decomposition structure. A gradient descent algorithm is proposed for the NC estimator, with rigorous statistical and computational convergence guarantees. Compared with the  convex estimators via nuclear norm regularizations, the consistency of the NC estimator   only requires that $T$ grows faster than $\max_{1\leq i\leq d}p_i$, which makes it attractive under high dimensionality. Although this approach requires initial values and known tensor ranks, the ridge-type ratio estimator can be used for determination of tensor ranks and initialization of the gradient descent algorithm.

This work is related to the literature on matrix-variate regression and tensor regression for independent data. The matrix-variate regression in \cite{DingCook18} has the same basic bilinear form, while an envelope method was introduced to further reduce the dimension.
\citet{raskutti2019convex} proposed a multi-response tensor regression model, where they mainly studied the third-order coefficient tensor and the SN regularization which is known to be statistically sub-optimal for higher-order tensor estimation. By contrast, we study the model for general higher-order tensor-valued time series. Moreover, our SSN   estimator has a much faster statistical convergence rate than the SN estimator.  For the non-convex tensor estimation problem, \citet{chen2019non} and \citet{han2020optimal} studied non-convex projected gradient descent methods for tensor regression. Our NC estimator can be viewed as a higher-order extension of the estimation approach in \citet{han2020optimal}.
In addition, existing literature on tensor regression has only considered independent data or Gaussian time series data, whereas we allow sub-Gaussianity of the time series. This is a non-trivial relaxation, since unlike the Gaussian case,  sub-Gaussian time series cannot be linearly transformed into independent samples.

The rest of the paper is organized as follows. Section \ref{sec:prelim} introduces basic notation and tensor algebra. Section \ref{sec:LRTAR} presents the proposed LRTAR model. A series of nuclear-norm-regularized estimation methods are covered in Section \ref{sec:HDM}, where we develop the non-asymptotic theory for three regularized estimators and rank selection consistency for the truncated estimator. Section \ref{sec:rank_constrained} proposes a non-convex estimation approach and presents its computational guarantees and statistical efficiency improvement.
Section \ref{sec:numerical} presents simulation studies and a real data analysis. Section \ref{sec:conclusion} concludes with a brief discussion.  We provide all technical proofs,  algorithms, and additional discussions in a separate online supplementary file.

\section{Tensor Decomposition and Tensor Autoregression}
\subsection{Preliminaries: Notation and Tensor Algebra \label{sec:prelim}}
Tensors, also known as multi-dimensional arrays, are natural higher-order extensions of  matrices. The order of a tensor is known as the dimension, way or mode, so a multi-dimensional array $\cm{X}\in\mathbb{R}^{p_1\times\cdots\times p_m}$ is called an $m$-th order tensor. We introduce some important notations and concepts of tensor operation in this subsection, and refer readers to \citet{kolda2009tensor} for a detailed review of basic tensor algebra. 

Throughout this paper, we denote vectors by boldface small letters, e.g. $\bm{x}$, $\bm{y}$,  matrices by boldface capital letters, e.g. $\bm{X}$, $\bm{Y}$, and tensors by boldface Euler capital letters, e.g. $\cm{X}$, $\cm{Y}$. For any two real-valued sequences $x_k$ and $y_k$, we write $x_k\gtrsim y_k$ if there exists a  constant $c>0$ such that $x_k\geq cy_k$ for all $k$,  and write $x_k\gg y_k$ if $\lim_{k\rightarrow\infty}y_k/x_k=0$. In addition, write $x_k\asymp y_k$ if $x_k\gtrsim y_k$ and $y_k\gtrsim x_k$. We use $C$ to denote a generic positive constant, which is independent of the dimensions and the sample size.

For a generic matrix $\bm{X}$,  we let $\bm{X}^\top$, $\|\bm{X}\|_{\text{F}}$, $\|\bm{X}\|_{\text{op}}$, $\|\bm{X}\|_\text{nuc}$,  $\text{vec}(\bm{X})$, and $\sigma_j(\bm{X})$ denote its transpose, Frobenius norm, operator norm, nuclear norm,  vectorization, and $j$-th largest singular value, respectively. For any matrix $\bm{X}\in\mathbb{R}^{p\times q}$, recall that  the nuclear norm and its dual norm, the operator norm, are defined as
\begin{equation}
	\|\bm{X}\|_\text{nuc}=\sum_{j=1}^{\min(p,q)}\sigma_j(\bm{X})~~\text{and}~~\|\bm{X}\|_{\text{op}}=\sigma_1(\bm{X}).
\end{equation}
For any square matrix $\bm{X}$, we let $\lambda_{\min}(\bm{X})$ and $\lambda_{\max}(\bm{X})$ denote its minimum and maximum eigenvalues. 
For any real symmetric matrices $\bm{X}$ and $\bm{Y}$, we write $\bm{X} \leq \bm{Y}$ if $\bm{Y}-\bm{X}$ is a positive semidefinite matrix.

Matricization, also known as unfolding, is the process of reordering the elements of a third- or higher-order tensor into a matrix. The most commonly used matricization is the one-mode matricization defined as follows.  For any $m$-th-order tensor $\cm{X}\in\mathbb{R}^{p_1\times\cdots\times p_m}$, its mode-$s$  matricization $\cm{X}_{(s)}\in\mathbb{R}^{p_s\times p_{-s}}$, with $p_{-s}=\prod_{i=1, i\neq s}^{m}p_i$, is the matrix obtained by setting the $s$-th tensor mode as its rows and collapsing all the others into  its columns, for $s=1,\dots, m$. Specifically,  the $(i_1,\dots,i_d)$-th element of $\cm{X}$  is mapped to the $(i_s,j)$-th element of $\cm{X}_{(s)}$, where
\begin{equation}
	j=1+\sum_{\substack{k=1\\k\neq s}}^m(i_k-1)J_k~~\text{with}~~J_k=\prod_{\substack{\ell=1\\\ell\neq s}}^{k-1}p_\ell.
\end{equation}

The above one-mode matricization can be extended to the multi-mode matricization by combining multiple modes to rows and combining the rest to columns of a matrix. For any index subset $S\subset\{1,\dots,m\}$, the multi-mode matricization $\cm{X}_{[S]}$ is the $\prod_{i\in S}p_i$-by-$\prod_{i\notin S}p_i$ matrix whose $(i,j)$-th element is mapped from the $(i_1,\dots,i_d)$-th element of $\cm{X}$, where
\begin{equation}
	i=1+\sum_{k\in S}(i_k-1)I_k ~\text{and}~ j=1+\sum_{k\notin S}(i_k-1)J_k,~\text{with}~I_k=\prod_{\substack{\ell\in S\\\ell<k}}p_\ell~\text{and}~J_k=\prod_{\substack{\ell\notin S\\\ell<k}}p_\ell.
\end{equation}
Note that the modes in the multi-mode matricization are collapsed following their original order $1, \dots, m$. 
Moreover, it holds $\cm{X}_{[S]}=\cm{X}_{[S^\complement]}^\top$, where $S^\complement=\{1,\dots, m\}\setminus S$ is the complement of $S$.
In addition, the one-mode matricization $\cm{X}_{(s)}$ defined above is simply $\cm{X}_{[\{s\}]}$. 

We next review the concepts of tensor-matrix multiplication, tensor generalized inner product and norm. For any $m$-th-order tensor $\cm{X}\in\mathbb{R}^{p_1\times\cdots\times p_m}$ and matrix $\bm{Y}\in\mathbb{R}^{q_k\times p_k}$ with $1\leq k\leq m$,
the mode-$k$ multiplication $\cm{X}\times_k\bm{Y}$ produces an $m$-th-order tensor in $\mathbb{R}^{p_1\times\cdots\times p_{k-1}\times q_k\times p_{k+1}\times\cdots\times p_m}$  defined by
\begin{equation}
	\left(\cm{X}\times_k\bm{Y}\right)_{i_1\cdots i_{k-1}ji_{k+1}\dots i_d}=\sum_{i_k=1}^{p_k}\cm{X}_{i_1\cdots i_m}\bm{Y}_{ji_k}.
\end{equation}
For any two tensors $\cm{X}\in\mathbb{R}^{p_1\times p_2\times\cdots\times p_m}$ and $\cm{Y}\in\mathbb{R}^{p_1\times p_2\times\cdots p_n}$ with $m\geq n$, their generalized inner product $\langle\cm{X},\cm{Y}\rangle$ is the $(m-n)$-th-order tensor  in $\mathbb{R}^{p_{n+1}\times\dots\times p_m}$ defined by
\begin{equation}\label{eq:tensorinner}
	\langle\cm{X},\cm{Y}\rangle_{i_{n+1}\dots i_{m}}=\sum_{i_1=1}^{p_1}\sum_{i_2=1}^{p_2}\dots\sum_{i_n=1}^{p_n}\cm{X}_{i_1i_2\dots i_ni_{n+1}\dots i_m}\cm{Y}_{i_1i_2\dots i_n},
\end{equation}
where  $1\leq i_{n+1}\leq p_{n+1},\dots,1\leq i_m\leq p_m$. In particular, when $m=n$, it reduces to the conventional real-valued inner product.  In addition,  the Frobenius norm of any tensor $\cm{X}$ is defined as $\|\cm{X}\|_{\text{F}}=\sqrt{\langle\cm{X},\cm{X}\rangle}$. 

Some basic properties of the tensor generalized inner product are as follows. Let $\cm{X}\in\mathbb{R}^{p_1\times p_2\times\cdots\times p_m}$, $\cm{Y}\in\mathbb{R}^{p_1\times p_2\times\cdots p_n}$, and $\cm{Z}\in\mathbb{R}^{p_1\times\cdots \times p_{k-1}\times q_k\times p_{k+1}\cdots \times p_n}$ be tensors with $m\geq n\geq k\geq1$. If $\bm{Y}\in\mathbb{R}^{q_k\times p_k}$, then $\langle\cm{X}\times_k\bm{Y},\cm{Z}\rangle=\langle\cm{X},\cm{Z}\times_k\bm{Y}^\top\rangle$.
If $\bm{Z}\in\mathbb{R}^{q_{n+j}\times p_{n+j}}$  with $1\leq j\leq m-n$, then $\langle\cm{X},\cm{Y}\rangle\times_j\bm{Z}=\langle\cm{X}\times_{n+j}\bm{Z},\cm{Y}\rangle$. Moreover, 
\begin{equation}\label{eq:vecinner}
	\text{vec}(\langle\cm{X},\cm{Y}\rangle)=\cm{X}_{[S]}\text{vec}(\cm{Y}),
\end{equation}
where $S=\{n+1,\dots, m\}$, and when $m=n$, $\cm{X}_{[\emptyset]}=\text{vec}(\cm{X})^\top$.

Finally, we summarize some concepts and useful results of the Tucker decomposition \citep{tucker1966some,deLathauwer2000multilinear}. For any tensor $\cm{X}\in\mathbb{R}^{p_1\times\cdots\times p_m}$, its Tucker ranks $(r_1,\dots,r_m)$ are defined as the matrix ranks of its one-mode matricizations, namely $r_i=\text{rank}(\cm{X}_{(i)})$, for $i=1,\dots,m$. Note that $r_i$'s are analogous to the row and column ranks of a matrix, but are not necessarily equal for third- and higher-order tensors.  However, the Tucker ranks must satisfy the condition 
\begin{equation}\label{eq:rankcond}
\left(\max_{1\leq i\leq m}r_i\right)^2\leq\prod_{i=1}^mr_i.
\end{equation}
If only one of the $r_i$'s is equal to the maximum rank $r_{\max}:=\max_{1\leq i\leq m}r_i$, \eqref{eq:rankcond}
is equivalent to $r_{\max} \leq \prod_{i=1, r_i\neq r_{\max}}^mr_i$; that is, the maximum Tucker rank must be no greater than the product of the other ranks.

Suppose that $\cm{X}$ has Tucker ranks $(r_1,\dots,r_m)$. Then $\cm{X}$ has the following  Tucker decomposition:
\begin{equation}\label{eq:Tucker}
	\cm{X}=\cm{Y}\times_1\bm{Y}_1\times_2\bm{Y}_2\cdots\times_m\bm{Y}_m=\cm{Y}\times_{i=1}^m\bm{Y}_i,
\end{equation}
where $\bm{Y}_i\in\mathbb{R}^{p_i\times r_i}$ for $i=1, \dots, m$ are the factor matrices  and $\cm{Y}\in\mathbb{R}^{r_1\times\cdots\times r_m}$  is the core tensor.
If $\cm{X}$ has the Tucker decomposition in \eqref{eq:Tucker}, then we have the following results for its one- and multi-mode matricizations:
\begin{equation}
	\cm{X}_{(s)}=\bm{Y}_s\cm{Y}_{(s)}(\bm{Y}_d\otimes\cdots\otimes\bm{Y}_{s+1}\otimes\bm{Y}_{s-1}\cdots\otimes\bm{Y}_1)^\top=\bm{Y}_s\cm{Y}_{(s)}(\otimes_{i\neq s}\bm{Y}_i)^\top, \quad s=1,\dots, m,
\end{equation}
and 
\begin{equation}\label{eq:mmmat}
	\cm{X}_{[S]}=(\otimes_{i\in S}\bm{Y}_i)\cm{Y}_{[S]}(\otimes_{i\notin S}\bm{Y}_i)^\top, \quad S\subset\{1,\dots,m\},
\end{equation}
where $\otimes_{i\neq s}, \otimes_{i\in S}$ and $\otimes_{i\notin S}$ are matrix Kronecker products operating in the reverse order within the corresponding index sets.

\subsection{Low-Rank Tensor Autoregression\label{sec:LRTAR}}
For the tensor-valued time series $\{\cm{Y}_t\}_{t=1}^T$, we propose the following Low-Rank Tensor Autoregressive (LRTAR) model:
\begin{equation}
	\label{eq:LTRTAR_model}
	\cm{Y}_t=\langle\cm{A},\cm{Y}_{t-1}\rangle+\cm{E}_t,
\end{equation}
where $\cm{A}\in\mathbb{R}^{p_1\times\cdots\times p_d\times p_1\times\cdots\times p_d}$ is the $2d$-th-order transition tensor which is assumed to have Tucker ranks $(r_1, \dots, r_{2d})$ with  $r_i=\text{rank}(\cm{A}_{(i)})$, $\langle\cdot,\cdot\rangle$ is the generalized tensor inner product defined in \eqref{eq:tensorinner} with $m=2d$ and $n=d$, and $\cm{E}_t \in\mathbb{R}^{p_1\times\cdots\times p_d}$ is the mean-zero random error at time $t$ with possible dependencies among its contemporaneous elements. 

By Section \ref{sec:prelim},  $\cm{A}$ admits the following  Tucker decomposition:
\begin{equation}
	\label{eq:Tucker_model}
	\cm{A}=\cm{G}\times_{i=1}^{2d}\bm{U}_i,
\end{equation}
where $\cm{G}\in\mathbb{R}^{r_1\times\cdots\times r_{2d}}$ is the core tensor, and $\bm{U}_i\in\mathbb{R}^{p_i\times r_i}$ are factor matrices for $1\leq i\leq 2d$.
Note that for any nonsingular matrices $\bm{O}_i\in\mathbb{R}^{r_i\times r_i}$ for $i=1, \dots, 2d$, it holds 
\[
\cm{G}\times_{i=1}^{2d}\bm{U}_i=(\cm{G}\times_{i=1}^{2d}\bm{O}_i) \times_{i=1}^{2d} (\bm{U}_i\bm{O}_i^{-1}).\]
Thus, although  the coefficient tensor $\cm{A}$ in \eqref{eq:LTRTAR_model} is identifiable, its Tucker decomposition in \eqref{eq:Tucker_model} suffers from rotational indeterminacy. 
To pin down the rotation matrices $\bm{O}_i$'s, a special Tucker decomposition, called the higher-order singular value decomposition (HOSVD), is commonly considered \citep{kolda2009tensor}.
In the HOSVD, the factor matrix $\bm{U}_i$ is defined as the tall orthonormal matrix consisting of the top $r_i$ left singular vectors of $\cm{A}_{(i)}$, for $i=1, \dots, 2d$. This further implies that the core tensor $\cm{G}=\cm{A}\times_{i=1}^{2d}\bm{U}_i^\top$ has the all-orthogonal property as follows: $\cm{G}_{(i)}\cm{G}_{(i)}^\top$ is a diagonal matrix for $i=1,\dots, 2d$.  
We will formally discuss the identification conditions of $\cm{G}$ and $\bm{U}_i$'s in Section \ref{sec:id}.

Denote $S_1=\{1,2,\dots,d\}$ and $S_2=\{d+1,d+2,\dots,2d\}$. Note that by \eqref{eq:vecinner}, model \eqref{eq:LTRTAR_model} can be written into the VAR form in \eqref{eq:VAR} with transition matrix $\bm{A}=\cm{A}_{[S_2]}$, i.e.,
\begin{equation} \label{eq:VAR_rep}
	\underbrace{\text{vec}(\cm{Y}_t)}_{\bm{y}_t}=\underbrace{(\otimes_{i\in S_2}\bm{U}_i)\cm{G}_{[S_2]}(\otimes_{i\in S_1}\bm{U}_i)^\top}_{\scalebox{0.7}{\cm{A}}_{[S_2]}} \underbrace{\text{vec}(\cm{Y}_{t-1})}_{\bm{y}_{t-1}}+\underbrace{\text{vec}(\cm{E}_t)}_{\bm{e}_t},
\end{equation}
where $\bm{y}_t=\text{vec}(\cm{Y}_t)$ and $\bm{e}_t=\text{vec}(\cm{E}_t)$.

By the VAR representation in  \eqref{eq:VAR_rep}, we immediately have the necessary and sufficient condition for the existence of a unique strictly stationary solution to model \eqref{eq:LTRTAR_model}  as follows.
\begin{assumption}
	\label{asmp:stationary}
	The spectral radius of $\cm{A}_{[S_2]}$ is strictly less than one.
\end{assumption}

\subsection{Model Identification}\label{sec:id}

To measure the extent of dimension reduction for the parameter space through the low-Tucker-rank assumption on $\cm{A}$, it is necessary to rule out the  rotational indeterminacy of the Tucker decomposition.
As mentioned in Section \ref{sec:LRTAR}, the HOSVD can be used to solve the rotational indeterminacy. Specifically, under the HOSVD, we have
\begin{equation}\label{eq:constraints}
\bm{U}_i^\top\bm{U}_i=\bm{I}_{r_i}\quad\text{and}\quad \cm{G}_{(i)}\cm{G}_{(i)}^\top \text{ is a diagonal matrix},
\end{equation}
for $i=1,\dots, 2d$. Thus, \eqref{eq:constraints} provides a convenient way for us to compute the effective number of degrees of freedom for the proposed LRTAR model.  Specifically, by  subtracting the number of constraints induced by \eqref{eq:constraints} from the total number of parameters in $\cm{G}$ and $\bm{U}_i$'s, we can obtain that the effective number of degrees of freedom for model \eqref{eq:LTRTAR_model} is
\begin{equation}\label{eq:dim}
	\prod_{i=1}^{2d}r_i+\sum_{i=1}^dr_i(p_i-r_i)+\sum_{i=1}^dr_{d+i}(p_i-r_{d+i}).
\end{equation}
This is substantially smaller than the total number of parameters in $\cm{A}$, i.e., $p^2$, with $p=\prod_{i=1}^dp_i$.  For the example with $d=3$ and $p_1=p_2=p_3=20$, if $r_1=\cdots=r_6=2$, then the number of parameters will be reduced from $p^2=64,000,000$ to 280.

While the HOSVD avoids the rotational indeterminacy, it is still not necessarily unique in general. It is possible, however, to guarantee the uniqueness of the HOSVD under the additional assumption that  the singular values of each one-mode matricization $\cm{A}_{(i)}$ are distinct for $i=1,\dots,2d$. 
Under this assumption, each $\bm{U}_i$ contains the $r_i$ left singular vectors of $\cm{A}_{(i)}$ corresponding to the  largest $r_i$ singular values which are all distinct. To further avoid the indeterminacy due to sign switches of the singular vectors,  it suffices to require that the first nonzero element in each column of $\bm{U}_i$ is positive. As a result, such an HOSVD will be unique.

However, despite the non-uniqueness of the Tucker decomposition, the transition tensor $\cm{A}$ itself is uniquely defined. Thus, the identification problem  will not be an issue for the estimation of the low-Tucker-rank tensor $\cm{A}$. Indeed, in Sections \ref{sec:HDM} and \ref{sec:rank_constrained}, we will introduce two types of methods to estimate $\cm{A}$. None of them requires a unique  Tucker decomposition of $\cm{A}$. In practice, we can first obtain a consistent estimator $\cm{\widehat{A}}$ by the methods in Sections \ref{sec:HDM} and \ref{sec:rank_constrained}, i.e., $\cm{\widehat{A}}=\cm{\widehat{A}}_{\textrm{SN}}, \cm{\widehat{A}}_{\textrm{MN}}, \cm{\widehat{A}}_{\textrm{SSN}}, \cm{\widehat{A}}_{\textrm{TSSN}}$, or $\cm{\widehat{A}}_{\textrm{NC}}$, and then apply the HOSVD to $\cm{\widehat{A}}$ to obtain the corresponding unique estimates $\cm{\widehat{G}}$ and $\bm{\widehat{U}}_i$'s. That is, $\bm{\widehat{U}}_i$ is calculated as the top $r_i$ left singular vectors of $\cm{\widehat{A}}_{(i)}$ with the first nonzero element in each column being positive, and  $\cm{\widehat{G}}=\cm{\widehat{A}}\times_{i=1}^{2d}\bm{\widehat{U}}_i^\top$.

Furthermore, it is worth noting that the column space of $\bm{U}_i$ is unique and identifiable, although  $\bm{U}_i$ suffers from rotational indeterminacy; this is similar to the loading matrix in factor models.
Thus, we can treat $\bm{\widehat{U}}_i$ and $\bm{\widehat{U}}_i\bm{O}$ to be equivalent for any orthogonal rotation $\bm{O}\in\mathbb{R}^{r_i\times r_i}$, as they correspond to the same factor interpretation. Moreover, for the orthonormal matrix $\bm{\widehat{U}}_i$, $\bm{\widehat{U}}_i\bm{\widehat{U}}_i^\top$ is the projection matrix of its column space. This projection matrix is unique and identifiable as $\bm{\widehat{U}}_i\bm{\widehat{U}}_i^\top=(\bm{\widehat{U}}_i\bm{O})(\bm{\widehat{U}}_i\bm{O})^\top$ for any orthogonal matrix $\bm{O}$. Hence, in practice, we can use the unique projection matrix $\bm{\widehat{U}}_i\bm{\widehat{U}}_i^\top$ to interpret the estimated low-dimensional factor loadings; see  the empirical analysis in Section \ref{sec:real_data}.

\subsection{Multi-Dimensional Dynamic Factor Interpretations}\label{sec:example}

To illustrate the interpretation of the proposed LRTAR model, we consider the monthly import-export data among 22 countries for 15 product categories studied in \citet{chen2019factor}, where $\cm{Y}_t\in\mathbb{R}^{22\times 22\times 15}$ is the observed Export-Import-Product tensor in month $t$, with the $(i,j,k)$-th entry of $\cm{Y}_t$ corresponding to the export of product $k$ from country $i$ to country $j$; see  Section \ref{sec:real_data} for a detailed analysis of this  dataset.

For simplicity, we first consider the proposed model with $d=2$ for the data obtained by aggregating all 15 product categories, denoted $\bm{Y}_t\in\mathbb{R}^{22\times 22}$, where each row represents an exporting country   and each column represents an importing country. In this case,  $\cm{A}=\cm{G}\times_{i=1}^{4}\bm{U}_i\in \mathbb{R}^{22\times 22\times 22\times 22}$, where $\bm{U}_i\in\mathbb{R}^{22\times r_i}$, with $r_i$ being much smaller than 22, and $\cm{G}\in\mathbb{R}^{r_1\times\cdots\times r_4}$. Suppose that $\bm{U}_i$'s satisfy \eqref{eq:constraints}. Then the proposed LRTAR model  for the matrix-valued time series $\bm{Y}_t$ implies that
\begin{equation}\label{eq:matfactor}
	\bm{U}_3^\top\bm{Y}_t\bm{U}_4 = \langle\cm{G},\bm{U}_1^\top\bm{Y}_{t-1}\bm{U}_2\rangle+\bm{U}_3^\top\bm{E}_t\bm{U}_4.
\end{equation}
Note that in \eqref{eq:matfactor}, 
$\bm{Y}_t$ and $\bm{E}_t$ are both  projected onto a low-dimensional space via  $\bm{U}_{3}$ and $\bm{U}_{4}$, while $\bm{Y}_{t-1}$ is projected onto another low-dimensional space via $\bm{U}_{1}$ and  $\bm{U}_{2}$. This provides a multi-dimensional dynamic factor interpretation of the import-export data as follows. According to \eqref{eq:matfactor},  the dynamic of the international market is driven by the low-dimensional lagged (predictor) matrix factor $\bm{P}_{t}:=\bm{U}_1^\top\bm{Y}_{t-1}\bm{U}_2\in\mathbb{R}^{r_1\times r_2}$, whereas the effect of the past information---encapsulated by the predictor tensor factor $\bm{P}_{t}$---on the present state of the market is manifested through the low-dimensional (response) matrix factor $\bm{R}_{t}:=\bm{U}_3^\top\bm{Y}_t\bm{U}_4\in\mathbb{R}^{r_3\times r_4}$. For the predictor factor,  $\bm{U}_1$ and $\bm{U}_2$ provide factor loadings along the directions of exporting  and importing countries, respectively. Similarly,  $\bm{U}_3$ and $\bm{U}_4$ provide those for the response factor. From a dynamical system point of view, the predictor factor $\bm{P}_{t}$ and the response factor $\bm{R}_{t}$ can be interpreted as the input and output of the economic system, respectively, while the core tensor $\cm{G}$ characterizes the predictive relationship between  $\bm{P}_t$ and $\bm{R}_t$.

The factor interpertation also applies to the general case with $d\geq3$. For the multi-category import-export data $\cm{Y}_t\in\mathbb{R}^{22\times 22\times 15}$ with $d=3$, \eqref{eq:matfactor} is extended to
\begin{equation}\label{eq:tenfactor}
	\cm{Y}_t\times_{i=4}^6\bm{U}_{i}^\top =\langle\cm{G},\cm{Y}_{t-1} \times_{i=1}^3\bm{U}_i^\top\rangle+\cm{E}_t\times_{i=4}^6\bm{U}_{i}^\top,
\end{equation}
where $(\bm{U}_1,\bm{U}_2,\bm{U}_3)$ and $(\bm{U}_4,\bm{U}_5,\bm{U}_6)$
can be viewed as  loadings of exporting countries, importing countries,  and product categories for the predictor and response tensor factors, respectively. 
The predictor tensor factor $\cm{P}_t:=\cm{Y}_{t-1}\times_{i=1}^3\bm{U}_{i}^\top\in\mathbb{R}^{r_1\times r_2\times r_3}$  drives the dynamic of the market, and the response tensor factor $\cm{R}_t:=\cm{Y}_t\times_{i=4}^6\bm{U}_{i}^\top\in\mathbb{R}^{r_4\times r_5\times r_6}$ reflects the reaction of the market to the past information.

\begin{remark}
In the literature on high-dimensional VAR models, a popular dimension reduction method is to impose sparsity on coefficients; see a recent review in \citet{Basu2021}. It is especially suitable for high-dimensional data where only a small subset of the variables are correlated, which is often the case in biological applications, e.g., the discovery of gene regulatory networks \citep{Shojaie2012}.
However, in some economic and financial applications, most variables are expected to be somewhat correlated. This will often lead to many small but nonzero coefficient estimates  under sparse estimation. As a result, the estimated sparse model could be hard to interpret. Rather, when pervasive cross-sectional dependency is observed in the data, it is probably more reasonable to assume that the variables in an economic or financial system are driven by some common factors. The LRTAR model provides the supervised multi-dimensional dynamic factor interpretation, which is the key advantage of the proposed model over the sparse modeling approach.
\end{remark}

\begin{remark}
While we focus on the lag-one tensor autoregression for simplicity, the proposed model can be readily extended to the case with a general lag order; see the discussion in  Section \ref{sec:conclusion}.
\end{remark}

\begin{remark}
In Appendix \ref{append:discuss} of the supplementary file, we further explore the relationship between  the proposed LRTAR model and the tensor factor model in \cite{chen2019factor}.   Note that the latter is an unsupervised learning method and cannot be used directly for forecasting, unless an explicit dynamic structure is imposed on the latent factor process. To build a connection with our model, we adapt the   tensor factor model by assuming that their latent factor process follows an autoregressive model. We can show that the proposed model is more flexible than the tensor factor model with autoregressive factors.
\end{remark}

\section{Convex Estimation via Nuclear Norm Regularization\label{sec:HDM}}

In Sections \ref{subsec:1mode}--\ref{subsec:trunc}, we consider  a series of convex  estimation methods for the proposed model via different nuclear norm regularizations. Throughout the rest of this paper,  the true value of the coefficient tensor $\cm{A}$ is denoted by $\cm{A}^*$.  While our estimation methods in Sections \ref{subsec:1mode} and \ref{subsec:sqrmode} will be developed from the exact low-Tucker-rank structure  of  the transition tensor, our theoretical analysis will allow  $\cm{A}^*$ to be approximately low-Tucker-rank, which includes the exact low-rankness as a special case. In other words, the proposed LRTAR model will be used as a working model. 

\subsection{Regularization via One-Mode Matricization \label{subsec:1mode}}

In model \eqref{eq:LTRTAR_model}, the exactly low-rank transition tensor $\cm{A}\in\mathbb{R}^{p_1\times\cdots\times p_{2d}}$  is subject to the constraints $r_i=\text{rank}(\cm{A}_{(i)})$, for $i=1,\dots, 2d$. A commonly used convex relaxation of such Tucker rank constraints is the regularization via the sum of nuclear (SN) norms of all the one-mode matricizations,
\begin{equation}\label{eq:SN_norm}
	\|\cm{A}\|_{\text{SN}}=\sum_{i=1}^{2d}\|\cm{A}_{(i)}\|_\textup{nuc}.
\end{equation}
The SN norm has been widely used in the literature  \citep{gandy2011tensor, tomioka2011statistical, liu2013tensor, raskutti2019convex} to simultaneously encourage the low-rankness for all modes of a tensor. 
This leads us to the SN norm regularized estimator
\begin{equation}
	\cm{\widehat{A}}_{\text{SN}}=\underset{\scalebox{0.7}{\cm{A}}}{\argmin}\left\{\frac{1}{T}\sum_{t=1}^T\|\cm{Y}_t-\langle\cm{A},\cm{Y}_{t-1}\rangle\|_\text{F}^2+\lambda_{\textup{SN}}\|\cm{A}\|_{\text{SN}}\right\},
\end{equation}
where $\lambda_{\text{SN}}>0$ is the tuning parameter. Note that if instead of  $\|\cm{A}\|_{\text{SN}}$, only one single nuclear norm, say $\|\cm{A}_{(1)}\|_\textup{nuc}$, is penalized, then the resulting estimator will only encourage the low-rankness for the first mode of $\cm{A}$,  while failing to do so for all the other $2d-1$ modes, and hence will be less effective than the above SN estimator.

To derive the estimation error bound for $\cm{\widehat{A}}_{\text{SN}}$, we make the following assumption on the random error $\bm{e}_t=\textup{vec}(\cm{E}_t)$.

\begin{assumption}
	\label{asmp:gaussian}
	Let $\bm{e}_t=\bm{\Sigma}_{\bm{e}}^{1/2}\bbm{\xi}_t$, where $\{\bbm{\xi}_t\}$ is a sequence of $i.i.d.$ random vectors, with $\mathbb{E}(\bbm{\xi}_t)=\bm{0}$ and $\textup{var}(\bbm{\xi}_t)=\bm{I}_p$, and $\bm{\Sigma_e}=\textup{var}(\bm{e}_t)$ is  a positive definite matrix. In addition, the entries $(\xi_{it})_{1\leq i\leq p}$ of  $\bbm{\xi}_t$ are mutually independent and $\kappa^2$-sub-Gaussian, i.e.,
	$\mathbb{E}(e^{\mu\xi_{it}})\leq e^{\kappa^2\mu^2/2}$, for any $\mu\in\mathbb{R}$ and  $i=1,\dots, p$.
\end{assumption}

Assumption \ref{asmp:gaussian} implies that  $\cm{E}_t$ are $i.i.d.$, which is standard in the literature on high-dimensional time series models. It may be relaxed to the weakly dependent case through strong mixing conditions as in \citet{wong2017lasso}.
The sub-Gaussianity condition in Assumption \ref{asmp:gaussian} is milder than the commonly used normality assumption in the literature  \citep{basu2015regularized, raskutti2019convex}. This relaxation is made possible through establishing a novel martingale-based concentration bound in the proof of the deviation bound; see Lemma \ref{lemma:deviation} in Appendix \ref{subsec:auxlemma} of the supplementary file. The covariance matrix $\bm{\Sigma_e}$ captures the contemporaneous dependency in $\cm{E}_t$, and the constant $\kappa$ controls the tail heaviness of the marginal distributions.

For any $z\in\mathbb{C}$, let $\mathcal{A}(z)=\bm{I}_{p}-\cm{A}^*_{[S_2]}z$ be a matrix polynomial, where $\mathbb{C}$ is the set of complex numbers. Let
$\mu_{\min}(\mathcal{A})=\underset{|z|=1}{\min}\,\lambda_{\min}(\mathcal{A}^\dagger(z)\mathcal{A}(z))$ and $\mu_{\max}(\mathcal{A})=\underset{|z|=1}{\max}\,\lambda_{\max}(\mathcal{A}^\dagger(z)\mathcal{A}(z))$, where $\mathcal{A}^\dagger(z)$ is the conjugate transpose of $\mathcal{A}(z)$. It can be shown that $\mu_{\min}(\mathcal{A})>0$ under Assumption \ref{asmp:stationary}; see also \cite{basu2015regularized} for more discussions on the connection between the spectral density of the VAR process and the two quantities. In addition, define the positive constants
\[
\alpha_{\textup{RSC}}=\frac{\lambda_{\min}(\bm{\Sigma_e})}{\mu_{\max}(\mathcal{A})}, \quad M_1=\frac{\lambda_{\max}(\bm{\Sigma_e})}{\mu^{1/2}_{\min}(\mathcal{A})}, \quad\text{and}\quad M_2=\frac{\lambda_{\min}(\bm{\Sigma}_{\bm{e}})\mu_{\max}(\mathcal{A})}{\lambda_{\max}(\bm{\Sigma_e})\mu_{\min}(\mathcal{A})}.
\]
Note that our theoretical analysis does not require $\alpha_\textup{RSC}$, $M_1$ and $M_2$ to be fixed as the dimension grows.  

In practice, it could be too stringent to assume that   $\cm{A}^*$ is exactly low-rank. In this section, we relax it to  the following approximately low-rank assumption:  We assume that all one-mode matricizations of the underlying true transition tensor $\cm{A}^*$ belong to the set of approximately low-rank matrices, namely $\cm{A}^*_{(i)}\in\mathbb{B}_q(r_q^{(i)};p_i,p_{-i}p)$ for some $q\in[0,1)$, where $r_q^{(1)},\dots,r_q^{(2d)}>0$ are the radii for all modes,
\begin{equation}
	\mathbb{B}_q(r;d_1,d_2):=\left\{\bm{M}\in\mathbb{R}^{d_1\times d_2}:\sum_{i=1}^{\min(d_1,d_2)}\sigma_i(\bm{M})^q\leq r\right\}
\end{equation}
is the set of approximately low-rank matrices defined by the $\ell_q$ norm of the singular values, $p_{-i}=p/p_i=\prod_{j=1,j\neq i}^{d}p_j$ for $i=1,\dots,d$, and $p_{-i}=p_{-i+d}$ for $i=d+1,\dots,2d$. For the convenience of notation, we let $0^0=0$. Note that when $q=0$, $\mathbb{B}_0(r;d_1,d_2)$ is the set of $d_1$-by-$d_2$ rank-$r$ matrices. For $q>0$, the restriction on $\sum_{i=1}^{\min(d_1,d_2)}\sigma_i(\bm{M})^q\leq r$ requires that the singular values decay to zero under a polynomial rate, and it is more general  than the exactly low-rank assumption.



\begin{theorem} \label{thm:SN}
Suppose that $\cm{A}^*_{(i)}\in\mathbb{B}_q(r_q^{(i)};p_i,p_{-i}p)$ for some  $q\in[0,1)$ and radii $r_q^{(i)}>0$ for $i=1,\dots, 2d$. If $T\gtrsim \max_{1\leq i\leq d}p_{-i}p+\max(\kappa^2,\kappa^4)M_2^{-2}p$ and $\lambda_{\textup{SN}}\gtrsim\kappa^2M_1d^{-2}\sum_{i=1}^{d}\sqrt{ p_{-i} p/T}$, then under 	Assumptions \ref{asmp:stationary} and \ref{asmp:gaussian}, 
	\begin{equation}
		\|\cm{\widehat{A}}_{\textup{SN}}-\cm{A}^*\|_{\textup{F}}\lesssim\sqrt{r_q}\left(\frac{2d\cdot\lambda_{\textup{SN}}}{\alpha_{\textup{RSC}}}\right)^{1-q/2}
	\end{equation}
 with probability at least $1-2\sum_{i=1}^{d}\exp(-Cp_{-i}p)-\exp[-C\min(\kappa^{-2},\kappa^{-4})M_2^2p]$,	where $r_q=(2d)^{-1}\sum_{i=1}^{2d}r_q^{(i)}$ is the average radius for  all one-mode matricizations.
\end{theorem}

By Theorem \ref{thm:SN}, when $\lambda_{\textup{SN}}\asymp\kappa^2M_1d^{-2}\sum_{i=1}^{d}\sqrt{p_{-i} p/T}$, the estimation error bound scales as
$\sqrt{r_q}(\kappa^4M_1^2\alpha_\textup{RSC}^{-2}\max_{1\leq i\leq d}p_{-i}p/T)^{1/2-q/4}$; note that the factor $d$ in the error bounds is canceled by the $d^{-2}$ in the rate of $\lambda_{\textup{SN}}$. When $\kappa$, $\alpha_\textup{RSC}^{-1}$ and $M_1$ are bounded, and $q=0$, namely $\cm{A}^*$ is exactly low-rank with Tucker ranks $(r_0^{(1)},\dots,r_0^{(2d)})$, the error bound reduces to $\sqrt{r_0\max_{1\leq i\leq d}p_{-i}p/T}$ and it is comparable to that in \citet{tomioka2011statistical} for $i.i.d.$ tensor regression. 

However, recent research in tensor analysis \citep[e.g.,][]{mu2014square,raskutti2019convex} shows that the SN norm regularization approach can be suboptimal. For our model,  this is mainly because $\cm{A}^*_{(i)}$ is an unbalanced \textit{fat-and-short} matricization of a higher-order tensor.  Specifically,  an essential intermediate step in the proof of Theorem \ref{thm:SN}  is to establish the deviation bound, where we need to upper bound the operator norm of a sub-Gaussian random matrix with the same dimensions as $\cm{A}^*_{(i)}$; see  Lemma \ref{lemma:deviation} in Appendix \ref{subsec:lemma} of the supplementary file. The order of this operator norm  will be dominated by the larger of the  row and column dimensions of the matrix $\cm{A}_{(i)}\in\mathbb{R}^{p_i \times p_{-i}p}$, and hence by the column dimension $p_{-i}p$, which eventually appears in the  error bound.  As a result, the imbalance of the matricization leads to the efficiency bottleneck of the SN estimator.

On the other hand, since the reduced-rank VAR  model can be regarded as an overparameterization of the proposed LRTAR model, alternatively one may focus on the  low-rankness of the transition matrix $\cm{A}_{[S_2]}$ in the VAR representation in \eqref{eq:VAR_rep}, and adopt the matrix nuclear (MN)  estimator \citep{negahban2011estimation} to estimate $\cm{A}$,
\begin{equation}\label{eq:matrixnuclear}
	\cm{\widehat{A}}_{\text{MN}}=\underset{\scalebox{0.7}{\cm{A}}}{\argmin}\left\{\frac{1}{T}\sum_{t=1}^T\|\cm{Y}_t-\langle\cm{A},\cm{Y}_{t-1}\rangle\|_\text{F}^2+\lambda_{\textup{MN}}\|\cm{A}_{[S_1]}\|_\textup{nuc}\right\},
\end{equation}
where $\lambda_{\text{MN}}>0$ is the tuning parameter.
Note that the multi-mode matricization $\cm{A}_{[S_2]}=\cm{A}_{[S_1]}^\top$ is a $p\times p$ square matrix. Thus, the loss of efficiency due to the unbalanced matricization can be avoided, which is confirmed by the following theorem.

\begin{theorem}\label{thm:MN}
Suppose that $\cm{A}^*_{[S_1]}\in\mathbb{B}_q(s_q^{(1)};p,p)$ for some $q=[0,1)$ and radius $s_q^{(1)}>0$. If $T\gtrsim [1+\max(\kappa^2,\kappa^4)M_2^{-2}]p$ and $\lambda_{\textup{MN}}\gtrsim\kappa^2M_1\sqrt{p/T}$, then under Assumptions \ref{asmp:stationary} and \ref{asmp:gaussian}, 	\begin{equation*}
		\|\cm{\widehat{A}}_{\textup{MN}}-\cm{A}^*\|_{\textup{F}}\lesssim\sqrt{s_q^{(1)}}\left(\frac{\lambda_{\textup{MN}}}{\alpha_{\textup{RSC}}}\right)^{1-q/2}
	\end{equation*}
with probability at least $1-\exp(-Cp)-\exp[-C\min(\kappa^{-2},\kappa^{-4})M_2^2p]$.
\end{theorem}

Theorem \ref{thm:MN} shows that, with  $\lambda_{\textup{MN}}\asymp\kappa^2M_1\sqrt{p/T}$, the estimation error bound for $\cm{\widehat{A}}_{\text{MN}}$ scales as $\sqrt{s_q^{(1)}}(\kappa^4M_1^2\alpha_\text{RSC}^{-2}p/T)^{1/2-q/4}$, where $s_q^{(1)}$ is the singular value radius of $\cm{A}^*_{[S_1]}$. This result is comparable to that in \citet{negahban2011estimation} for reduced-rank VAR models, yet we relax both the singular value constraint $\|\cm{A}_{[S_1]}^*\|_{\text{op}}<1$ and the normality assumption on the random error in their paper. This estimation error bound is clearly smaller than that in Theorem \ref{thm:SN}, as $(\max_{1\leq i\leq d}p_{-i}p/T)^{1/2-q/4}$ in general can be much larger than $(p/T)^{1/2-q/4}$ when $s_q^{(1)}\asymp r_q$. Therefore, adopting square matricization can indeed improve the estimation performance. 

The idea of using square matricization to improve efficiency was adopted by \citet{mu2014square} in low-rank tensor completion problems. Their proposed method, called the square deal, is to first unfold a general higher-order tensor into a matrix with similar numbers of rows and columns, and then use the MN norm as the regularizer. However,  despite the advantage of $\cm{\widehat{A}}_{\textup{MN}}$  over $\cm{\widehat{A}}_{\textup{SN}}$, Theorem \ref{thm:MN} reveals another drawback of $\cm{\widehat{A}}_{\textup{MN}}$. 
That is, the  error bounds for $\cm{\widehat{A}}_{\textup{MN}}$ will increase as the radius $s_q^{(1)}$ for the singular values of $\cm{A}^*_{[S_1]}$ becomes larger.  
In other words, unless we have prior knowledge that the $\ell_q$-``norm'' of singular values of $\cm{A}^*_{[S_1]}$ is truly small, $\cm{\widehat{A}}_{\textup{MN}}$ may not be desirable in practice. 

On the other hand, although the SN regularizer in \eqref{eq:SN_norm} suffers from  inefficiency due to the imbalance of one-mode matricizations,  it has the attractive feature of simultaneously encouraging low-rankness across all modes of  $\cm{A}$, and thus is more  efficient than its counterpart which considers only one single one-mode matricization, say, $\|\cm{A}_{(1)}\|_\textup{nuc}$. Similarly, if we can encourage  low-rankness across all possible square matricizations of $\cm{A}$, the estimation performance may be further improved upon $\cm{\widehat{A}}_{\textup{MN}}$.  This motivates us to propose  a new regularization approach in the next subsection.

\begin{remark}	
Since our statistical theory is  non-asymptotic,  the dimensions $p_i$'s, approximate or exact Tucker ranks such as $r_q^{(i)}$'s in Theorem \ref{thm:SN}, and any other quantities appearing in the error bounds  are all allowed to diverge to infinity. 
Our results show how these quantities explicitly affect the error bounds. However, for simplicity of understanding the convergence rates, one may assume that $\alpha_\textup{RSC}$, $M_1$ and $M_2$ are fixed; see Table \ref{table1}. For example, it is common to assume that  $0<c\leq \lambda_{\min}(\bm{\Sigma_e})\leq \lambda_{\max}(\bm{\Sigma_e})\leq C$. In addition, when the spectral radius of $\cm{A}^*_{[S_2]}$ is bounded away from one, it can be shown that $\mu_{\min}(\mathcal{A})$ is also bounded away from zero.
\end{remark}

\subsection{Regularization via Square Matricization \label{subsec:sqrmode}}
Motivated by the discussion at the end of Section \ref{subsec:1mode}, we propose a novel convex regularizer which improves upon both   SN and MN regularizers in \eqref{eq:SN_norm} and \eqref{eq:matrixnuclear},  by simultaneously encouraging low-rankness across all possible square matricizations of   $\cm{A}$.

For any $2d$-th-order tensor $\cm{A}\in\mathbb{R}^{p_1\times\cdots\times p_{d}\times p_1\times\cdots\times p_{d}}$, its multi-mode matricization $\cm{A}_{[I]}$ will be a $p\times p$ square  matrix, with $p=\prod_{i=1}^{d}p_i$, if the index set is  chosen as \[I=\{\ell_1,\dots, \ell_d\},\] 
where each index $\ell_i$ is set to either $i$ or $d+i$,  for $i=1,\dots, d$.  For instance, $\cm{A}_{[S_1]}$ is the square matricization formed by setting $\ell_i=i$ for all $i=1,\dots, d$.  Moreover, if $\cm{A}$ has Tucker ranks $(r_1,\dots,r_{2d})$, then the rank of the matricization $\cm{A}_{[I]}$ is  at most $\min(\prod_{i=1,i\in I}^{2d}r_i,\prod_{i=1,i\notin I}^{2d}r_i)$. Therefore, if we penalize the sum of nuclear norms of all such squares matricizations, which we call the sum of square-matrix nuclear (SSN) norms for simplicity, then the resulting estimator would  enjoy the efficiency gain from both  the use of square matricizations and simultaneous  incorporation of many rank constraints.

Obviously, there are $2^d$ possible choices of the index set $I$ that corresponds to a square matricization $\cm{A}_{[I]}$. However, since $\cm{A}_{[I]}=\cm{A}_{[I^\complement]}^\top$, when defining the SSN norm, we only need to include one of $I$ and its complement $I^\complement$. A simple way to do so is to choose only sets containing the index one. That is, fix $\ell_1=1$ and choose  $\ell_i=i$ or  $d+i$ for $i=2,\dots, d$. This results in totally $2^{d-1}$ chosen index sets, denoted by $I_1,I_2,\dots,I_{2^{d-1}}$. Note that $I_1=S_1=\{1,\dots, d\}$. For example, when $d=3$, we have four chosen index sets, $I_1=\{1,2,3\}, I_2=\{1,5,3\}, I_3=\{1,2,6\}$ and $I_4=\{1,5,6\}$.

Based on the above choice of the $2^{d-1}$ index sets,  we introduce the following SSN norm,
\begin{equation} \label{eq:SSN_norm}
	\|\cm{A}\|_{\text{SSN}}=\sum_{k=1}^{2^{d-1}}\left\|\cm{A}_{[I_k]}\right\|_\textup{nuc}.
\end{equation}
For a tuning parameter $\lambda_{\textup{SSN}}>0$, the corresponding estimator is defined as
\begin{equation}\label{eq:SSN_est}
	\cm{\widehat{A}}_{\text{SSN}}=\underset{\scalebox{0.7}{\cm{A}}}{\argmin}\left\{\frac{1}{T}\sum_{t=1}^T\|\cm{Y}_t-\langle\cm{A},\cm{Y}_{t-1}\rangle\|_\text{F}^2+\lambda_{\textup{SSN}}\|\cm{A}\|_{\text{SSN}}\right\}.
\end{equation}

If the rank of one-mode matricizations $\text{rank}(\cm{A}_{(i)})=r_i$, each square matricization $\cm{A}_{[I_k]}$ is also low-rank with $\text{rank}(\cm{A}_{[I_k]})\leq\min(\prod_{i=1,i\in I_k}^{2d}r_i,\prod_{i=1,i\notin I_k}^{2d}r_i)$. Similarly, if all $\cm{A}_{(i)}$s are approximately low-rank, the square matricizations are approximately low-rank as well. In contrast to the SN norm in \eqref{eq:SN_norm} which directly matches the Tucker ranks $\text{rank}(\cm{A}_{(i)})$ for $i=1,\dots, d$,  the SSN norm encourages the low-Tucker-rank structure of $\cm{A}$ by simultaneously encouraging low-rankness of all   square matricizations $\cm{A}_{[I_k]}$'s.
The following theorem gives the theoretical results for $\cm{\widehat{A}}_{\text{SSN}}$.

\begin{theorem} \label{thm:SSN}
Suppose that $\cm{A}^*_{[I_k]}\in\mathbb{B}_q(s_q^{(k)};p,p)$ for  some $q\in[0,1)$ and radii $s_q^{(k)}>0$ for $k=1,\dots,2^{d-1}$. If $T\gtrsim [1+\max(\kappa^2,\kappa^4)M_2^{-2}]p$ and $\lambda_{\textup{SSN}}\gtrsim\kappa^2M_12^{1-d}\sqrt{p/T}$, 
under Assumptions \ref{asmp:stationary} and \ref{asmp:gaussian}, then  with probability at least $1-\exp[-C(p-d)]-\exp[-C\min(\kappa^{-2},\kappa^{-4})M_2^2p]$,
\begin{equation}
	\|\cm{\widehat{A}}_{\textup{SSN}}-\cm{A}^*\|_{\textup{F}}\lesssim\sqrt{s_q}\left(\frac{2^{d-1}\lambda_{\textup{SSN}}}{\alpha_{\textup{RSC}}}\right)^{1-q/2}
\end{equation}
where $s_q=2^{1-d}\sum_{k=1}^{2^{d-1}}s_q^{(k)}$ is the average radius for all square matricizations.
\end{theorem}

\begin{table}
	\begin{center}
		\begin{tabular}{c|ccc}\hline
			& SN & MN & SSN\\\hline
			Sample size & $T\gtrsim (\max_{1\leq i\leq d}p_{-i}+ M_2^{-2})p$ & $T\gtrsim (1+M_2^{-2})p$ & $T\gtrsim (1+M_2^{-2})p$\\
			Estimation error & $\sqrt{r_q}(\max_{1\leq i\leq d}p_{-i}p/T)^{1/2-q/4}$ & $\sqrt{s_q^{(1)}}(p/T)^{1/2-q/4}$ & $\sqrt{s_q}(p/T)^{1/2-q/4}$ \\\hline
		\end{tabular}
		\caption{Summary of sample size conditions and error  bounds in Theorems \ref{thm:SN}--\ref{thm:SSN}, where  $p_{-i}=\prod_{j=1, j\neq i}^{d}p_j$, $r_q=(2d)^{-1}\sum_{i=1}^{2d}r_q^{(i)}$, and $s_q=2^{1-d}\sum_{k=1}^{2^{d-1}}s_q^{(k)}$, assuming that $\kappa$, $\alpha_\textup{RSC}^{-1}$ and $M_2$ are bounded. 
			\label{table1}} 
	\end{center}
\end{table}

By Theorem \ref{thm:SSN}, when $\lambda_{\textup{SSN}}\asymp\kappa^2M_12^{1-d}\sqrt{p/T}$, the estimation error bound scales as $\sqrt{s_q}(\kappa^4M_1^2\alpha_\text{RSC}^{-2}p/T)^{1/2-q/4}$, and reduces to $\sqrt{s_0p/T}$ in the exactly low-rank setting for $q=0$ when $\kappa$, $\alpha_\textup{RSC}^{-1}$ and $M_2$ are bounded. For a clearer comparison among the three estimators $\cm{\widehat{A}}_{\textup{SN}}, \cm{\widehat{A}}_{\textup{MN}}$ and $\cm{\widehat{A}}_{\textup{SSN}}$, we summarize the main results of Theorems \ref{thm:SN}--\ref{thm:SSN} in Table \ref{table1}. 
First, both $\cm{\widehat{A}}_{\textup{SSN}}$ and $\cm{\widehat{A}}_{\textup{MN}}$ have much smaller error bounds and less stringent sample size requirements than $\cm{\widehat{A}}_{\textup{SN}}$, due to the diverging dimension $p_{-i}$ in the results of the latter. This reaffirms the advantage of the square matricizations.

Secondly, comparing $\cm{\widehat{A}}_{\textup{SSN}}$ to $\cm{\widehat{A}}_{\textup{MN}}$, since the factor $s_q$ in the error bounds of $\cm{\widehat{A}}_{\textup{SSN}}$ is the average of all $s_q^{(k)}$ for $k=1,\dots, 2^{d-1}$, $\cm{\widehat{A}}_{\textup{SSN}}$ can protect us from the bad scenarios where the $\ell_q$-``norm'' of the singular values of $\cm{A}^*_{[S_1]}$ is relatively large.  If all the $s_q^{(k)}$'s are of the same order, then the error upper bounds for $\cm{\widehat{A}}_{\textup{SSN}}$ and $\cm{\widehat{A}}_{\textup{MN}}$ in Table \ref{table1} will be similar. However, our simulation results in Section \ref{sec:sim} show that $\cm{\widehat{A}}_{\textup{SSN}}$ clearly outperforms $\cm{\widehat{A}}_{\textup{MN}}$ under various settings, even when $s_q^{(1)}=\cdots=s_q^{(2d)}$.
Indeed, the  error bounds for $\cm{\widehat{A}}_{\textup{SSN}}$ in Theorem \ref{thm:SSN} is likely to be loose, which is believed to be caused by taking the upper bounds on the dual norm of the SSN norm in the proof of Lemma \ref{lemma:dualnorm}; see Appendix \ref{subsec:lemma} of the supplementary file for details. By contrast, the error bounds for $\cm{\widehat{A}}_{\textup{MN}}$ are minimax-optimal \citep{negahban2011estimation}.
Therefore, although our theoretical results are not sharp enough to  distinguish clearly between the error rates of $\cm{\widehat{A}}_{\textup{SSN}}$  and $\cm{\widehat{A}}_{\textup{MN}}$,  we conjecture that the actual rate of the former is generally smaller than that of the latter. Methodologically, this is also easy to understand because, unlike $\cm{\widehat{A}}_{\textup{MN}}$, $\cm{\widehat{A}}_{\textup{SSN}}$ simultaneously encourages the low-rankness across all square matricizations of $\cm{A}$ rather than just on  $\cm{A}_{[S_1]}$.

\begin{remark}
While our SSN regularization is proposed in the time series context, the idea of imposing joint penalties on all (close to) square matricizations of the coefficient tensor may  be extended to general higher-order tensor estimation problems.  It can also be refined to accommodate particular structures of the data. For example, 
if some of the $d$ modes of the tensor-value time series $\cm{Y}_t$, namely $p_1,\dots, p_d$, are equal, then even a greater number of possible square matricizations of $\cm{A}$ can be formed, resulting in improved estimation efficiency. 
\end{remark}

\subsection{Truncated Regularized Estimation \label{subsec:trunc}}
While the  estimation methods in Sections \ref{subsec:1mode} and \ref{subsec:sqrmode} do not require exact low-rankness of the true transition tensor $\cm{A}^*$, sometimes imposing exact low-rankness is more desirable if one wants to interpret the underlying dynamic tensor factor structures. As  discussed in Section \ref{sec:example}, the Tucker ranks determine the dimensions of the dynamic factors. For greater model interpretability,  we further  consider the case that   $\cm{A}^*$ is exactly low-rank and propose a truncation method to consistently estimate its true Tucker ranks $(r_1, \dots, r_{2d})$.

Let $\gamma>0$ be a threshold parameter to be chosen properly. Given the estimator  $\cm{\widehat{A}}_{\text{SSN}}$, for each $i=1, \dots, 2d$, we calculate the singular value decomposition (SVD) of the mode-$i$ matricization $(\cm{\widehat{A}}_{\text{SSN}})_{(i)}$ with the singular values arranged in descending order. Next we truncate the SVD by retaining only singular values greater than $\gamma$, and  take their corresponding left singular vectors to define the matrix $\bm{\widetilde{U}}_i$. Then, the truncated core tensor is defined as \[\cm{\widetilde{G}}=\cm{\widehat{A}}_{\text{SSN}}\times_{i=1}^{2d}\bm{\widetilde{U}}_i^\top,\]
based on which we propose the truncated sum of square-matrix nuclear (TSSN) estimator
\[
\cm{\widehat{A}}_{\text{TSSN}}=\cm{\widetilde{G}}\times_{i=1}^{2d}\bm{\widetilde{U}}_i.
\]

To derive the theoretical results on rank selection, we make the following assumption on the exact Tucker ranks and the magnitude of the singular values.

\begin{assumption}\label{asmp:truncate}
	For all $i=1,\dots,2d$, $\sigma_r(\cm{A}_{(i)}^*)=0$ for all $r>r_i$, and there exists a constant $C>1$ such that $\min_{1\leq i\leq 2d}\sigma_{r_i}\left(\cm{A}_{(i)}^*\right)\geq C\gamma$. As $T\rightarrow\infty$, the threshold parameter satisfies $\gamma \gg (\kappa^2M_1/\alpha_{\textup{RSC}}) \sqrt{s_0p/T}$, where $s_0=2^{1-d}\sum_{k=1}^{2^{d-1}}\textup{rank}(\cm{A}^*_{[I_k]})$.
\end{assumption}

Assumption \ref{asmp:truncate} requires that $\cm{A}^*$ has exact Tucker ranks $(r_1,\dots,r_{2d})$ which do not diverge too fast. The smallest positive singular value for each $\cm{A}_{(i)}^*$ is assumed to be bounded away from the threshold $\gamma$ when the sample size is sufficiently large. Since Assumption  \ref{asmp:truncate} involves unknown quantities, it cannot be used directly for determining $\gamma$ in practice.  Instead, we recommend using a data-driven threshold parameter $\gamma$ to be described  below.

The rank selection consistency of the truncation method and the asymptotic estimation error rate of $\cm{\widehat{A}}_{\text{TSSN}}$ are given by the following theorem. 

\begin{theorem} \label{thm:rankconsistency}
	Under the conditions of Theorem \ref{thm:SSN} and Assumption \ref{asmp:truncate}, if the tuning parameter $\lambda_{\textup{SSN}}\asymp\kappa^2M_12^{1-d}\sqrt{p/T}$, then
	\begin{equation*}
		\mathbb{P}\left\{\textup{rank}\left((\cm{\widehat{A}}_{\textup{TSSN}})_{(i)}\right)=\textup{rank}(\cm{A}^*_{(i)}),~\text{for}~i=1,\dots,2d\right\}\to 1,
	\end{equation*}
	as  $T\to\infty$, and for any fixed $d$,
	\[
	\|\cm{\widehat{A}}_{\textup{TSSN}}-\cm{A}^*\|_{\textup{F}}=O_p\left (\sqrt{s_0p/T}\right ), 
	\]
	where $s_0$ is defined as in Assumption \ref{asmp:truncate}.
\end{theorem}

Similar to \citet{gandy2011tensor}, the SSN estimation can be solved by the alternating direction method of multipliers (ADMM) algorithm, while the truncation can be done by the standard HOSVD; see  Appendix \ref{append:ADMM} of the supplementary file for details. For the tuning parameter selection, since the cross-validation method is unsuitable for time series or intrinsically ordered data, we apply the Bayesian information criterion (BIC) to select the optimal $\lambda_\text{SSN}$, where the number of degrees of freedom is defined as $2^{-(d-1)}\sum_{k=1}^{2^{d-1}}s_k(2p-s_k)$. For the threshold parameter $\gamma$ of the truncated estimator, we recommend $\gamma=2^{d-1}\lambda_\text{SSN}/4$ to practitioners, where $\lambda_{\textup{SSN}}$ is the optimal tuning parameter selected by the BIC. Similarly, the BIC can be used for tuning parameter selection for SN and MN estimators as well.

\begin{remark}\label{rmk:rankadjust}
The Tucker ranks of $\cm{A}^*$ must satisfy $r_{\max}^2\leq\prod_{i=1}^{2d}r_i$, where $r_{\max}=\max_{1\leq i\leq 2d}r_i$; see also the discussion below  \eqref{eq:rankcond}.  In practice, if the ranks selected by the truncated estimator fail to satisfy this condition, that is, when $\widehat{r}_{\max}:=\max_{1\leq i\leq 2d}\widehat{r}_i$  exceeds the product of the other ranks (i.e., $\prod_{i=1, \widehat{r}_i\neq \widehat{r}_{\max}}^{2d}\widehat{r}_i$), where $\widehat{r}_i=\textup{rank}((\cm{\widehat{A}}_\textup{TSSN})_{(i)})$, we recommend the following selection procedure for rank adjustments. First, for each $\widehat{r}_i$ that is not equal to $\widehat{r}_{\max}$, we increase it until the above condition on Tucker ranks is met, while fixing the other ranks, and obtain the adjusted truncated estimator. Next, for all adjusted estimators, we select the most suitable ranks via BIC. For example, if the TSSN estimator produces the Tucker ranks $(3,2,1,1,1,1)$, we consider adjusted ranks $(3,3,1,1,1,1)$, $(3,2,2,1,1,1)$, $(3,2,1,2,1,1)$, $(3,2,1,1,2,1)$ or $(3,2,1,1,1,2)$, and then select the one with the smallest BIC.
\end{remark}

\section{Non-convex Tensor Regression Estimation \label{sec:rank_constrained}} 

\subsection{Non-convex Estimation}

While significant efficiency improvement can be achieved by the square matricization in Section \ref{subsec:sqrmode}, the consistency of the SSN and TSSN estimators still requires the sample size $T$ grows faster than the overall dimension $p=\prod_{j=1}^dp_j$. To further lower the sample size requirement and improve the estimation efficiency, this section proposes a non-convex estimation method for the LRTAR model under the assumption that $\cm{A}^*$ is exactly low-rank. 

First, we assume that the true Tucker ranks $(r_1,r_2,\dots,r_{2d})$ are known.
Following \cite{han2020optimal}, we can estimate the transition tensor via the non-convex (NC) optimization:
\begin{equation}\label{eq:ORLTR}
	\begin{split}
		\cm{\widehat{A}}_\textup{NC}&=\cm{\widetilde{G}}\times_{i=1}^{2d}\widetilde{\bm{U}}_i\\
		&=\underset{\substack{\scalebox{0.7}{\cm{G}}\in\mathbb{R}^{r_1\times\cdots\times r_{2d}}\\\bm{U}_i\in\mathbb{R}^{p_i\times r_i}}}{\arg\min}\left\{\frac{1}{2T}\sum_{t=1}^T\|\cm{Y}_t-\langle\cm{G}\times_{i=1}^{2d}\bm{U}_i,\cm{Y}_{t-1}\rangle\|_\text{F}^2+\frac{a}{2}\sum_{i=1}^{2d}\|\bm{U}_i^\top\bm{U}_i-b^2\bm{I}_{r_i}\|_\text{F}^2\right\},
	\end{split}
\end{equation}
where the regularization terms $\|\bm{U}_i^\top\bm{U}_i-b^2\bm{I}_{r_i}\|_\text{F}^2$ are used to prevent $\bm{U}_i$'s from being singular and balance the scale of tensor decomposition components, and $a,b>0$ are tuning parameters. 

To further understand the regularization terms for $\bm{U}_i$'s, let $\overline{\mathcal{L}}(\cm{A}):=(2T)^{-1}\sum_{t=1}^T\|\cm{Y}_t-\langle\cm{A},\cm{Y}_{t-1}\rangle\|_\text{F}^2$ and $\mathcal{L}(\cm{G},\bm{U}_1,\dots,\bm{U}_{2d}):=\overline{\mathcal{L}}([\![\cm{G};\bm{U}_1,\dots,\bm{U}_{2d}]\!])$ be the squared loss functions with respect to $\cm{A}$ and its Tucker decomposition, respectively. While the optimization in  \eqref{eq:ORLTR} is  unconstrained, any solution $\widetilde{\bm{U}}_i$ will satisfy $\widetilde{\bm{U}}_i^\top\widetilde{\bm{U}}_i=b^2\bm{I}_{r_i}$. Otherwise, we can always find some nonsingular matrices $\bm{Q}_i\in\mathbb{R}^{r_i\times r_i}$, for $i=1,\dots,2d$, such that $\widetilde{\bm{U}}_i\bm{Q}_i=\overline{\bm{U}}_i$ and $\overline{\bm{U}}_i^\top\overline{\bm{U}}_i=b^2\bm{I}_{r_i}$. In this case, $\mathcal{L}(\widetilde{\cm{G}},\widetilde{\bm{U}}_1,\dots,\widetilde{\bm{U}}_{2d})=\mathcal{L}(\widetilde{\cm{G}}\times_{i=1}^{2d}\bm{Q}_i^{-1},\overline{\bm{U}}_1,\dots,\overline{\bm{U}}_{2d})$, while the regularization terms for $\overline{\bm{U}}_i$ reduce to zero. This will result in a contradiction with the definition of minimizers. 
Note that we do not require $b=1$, i.e., $\widetilde{\bm{U}}_i$ may not be orthonormal. Moreover, we do not require the uniqueness of   $\widetilde{\cm{G}}$ and $\widetilde{\bm{U}}_i$'s, since we only need the resulting $\cm{\widehat{A}}_\textup{NC}$ from \eqref{eq:ORLTR}; see also \cite{han2020optimal}. However, as discussed in Section \ref{sec:id}, after obtaining  $\cm{\widehat{A}}_\textup{NC}$, we can apply the HOSVD to $\cm{\widehat{A}}_\text{NC}$ to obtain the uniquely defined orthonormal estimates $\widehat{\bm{U}}_i$'s and all-orthogonal estimate $\widehat{\cm{G}}$.

The partial gradients of the squared loss $\mathcal{L}$ with respect to $\bm{U}_i$ and $\cm{G}$ are defined as
\begin{equation}
	\nabla_{\bm{U}_i}\mathcal{L}=\nabla\overline{\mathcal{L}}_{(i)}(\otimes_{j=1,j\neq i}^{2d}\bm{U}_j)\cm{G}_{(i)}^\top=\nabla\overline{\mathcal{L}}_{(i)}[\cm{G}\times_{j=1,j\neq i}^{2d}\bm{U}_j]_{(i)}~~\text{and}~~\nabla_{\scalebox{0.75}{\cm{G}}}\mathcal{L}=\nabla\overline{\mathcal{L}}\times_{i=1}^{2d}\bm{U}_i^\top,
\end{equation}
where $\nabla\overline{\mathcal{L}}=T^{-1}\sum_{t=1}^T\cm{Y}_{t-1}\circ[\langle\cm{A},\cm{Y}_{t-1}\rangle-\cm{Y}_t]$, and $\circ$ denotes the tensor outer product. The problem in \eqref{eq:ORLTR} can be solved by the gradient descent algorithm: for $s=0,\dots,I-1$,
\begin{equation}
	\begin{split}
		\bm{U}_i^{(s+1)} & = \bm{U}_i^{(s)} - \eta\nabla_{\bm{U}_i}\mathcal{L}^{(s)}-\eta a\left[\bm{U}_i^{(s)}(\bm{U}_i^{(s)\top}\bm{U}_i^{(s)}-b^2\bm{I}_{r_i})\right],\quad i=1,\dots,2d,\\
		\text{and}~\cm{G}^{(s+1)} & = \cm{G}^{(s)} - \eta\nabla_{\scalebox{0.75}{\cm{G}}}\mathcal{L}^{(s)},
	\end{split}
\end{equation}
with the initial values $(\cm{G}^{(0)},\bm{U}_1^{(0)},\dots,\bm{U}_{2d}^{(0)})$, where $I$ is the total number of iterations and $\eta>0
$ is the step size of each iteration. The final output is $\cm{\widehat{A}}_\textup{NC}=[\![\cm{G}^{(I)};\bm{U}_1^{(I)},\dots,\bm{U}_{2d}^{(I)}]\!]$, and we may apply the HOSVD to $\cm{\widehat{A}}_\textup{NC}$ to ensure the identifiability of the Tucker decomposition.


\subsection{Computational and Statistical Convergence Analysis}
\label{sec:4.2}

In this subsection, we present the main properties of the NC estimation method. Theoretical analysis of this method is challenging due to the non-convex nature of the problem. To show that the proposed method is valid, we derive the linear convergences of gradient descent iterates to the ground truth up to a statistical error. First, we introduce some regulatory conditions, namely the restricted strong convexity, restricted strong smoothness, and deviation bound conditions.

\begin{definition}\label{def:RSC_RSS}
	The squared loss function $\overline{\mathcal{L}}$ is restricted strongly convex with parameter $\alpha$ and restricted strongly smooth with parameter $\beta$, such that for any low-rank tensors $\cm{A}_1,\cm{A}_2\in\mathbb{R}^{p_1\times\cdots\times p_{2d}}$ with Tucker ranks $(r_1,\dots,r_{2d})$,
	\begin{equation}
		\frac{\alpha}{2}\|\cm{A}_1-\cm{A}_2\|_\textup{F}^2\leq\overline{\mathcal{L}}(\cm{A}_1)-\overline{\mathcal{L}}(\cm{A}_2)-\langle\nabla\overline{\mathcal{L}}(\cm{A}_2),\cm{A}_1-\cm{A}_2\rangle\leq\frac{\beta}{2}\|\cm{A}_1-\cm{A}_2\|_\textup{F}^2.
	\end{equation}
\end{definition}

\begin{definition}\label{def:deviation}
	For the given Tucker ranks $(r_1,\dots,r_{2d})$, denote $\xi(r_1,\dots,r_{2d})$ as
	\begin{equation}
		\xi(r_1,\dots,r_{2d})=\sup_{\|\scalebox{0.75}{\cm{G}}\|_\textup{F}=1,\bm{U}_i^\top\bm{U}_i=\bm{I}_{r_i}}\left\langle\nabla\overline{\mathcal{L}}(\cm{A}^*),[\![\cm{G};\bm{U}_1,\dots,\bm{U}_{2d}]\!]\right\rangle.
	\end{equation}
\end{definition}

The restricted strong convexity and smoothness consitions are essential for convergence analysis of a large number of non-convex optimization problems \citep{jain2017non}. The deviation bound $\xi(r_1,\dots,r_{2d})$ characterizes the magnitude of the statistical noise projected onto the low-rank tensor spaces. Moreover, for the true value $\cm{A}^*$, denote by $\overline{\sigma}=\max_{1\leq i\leq 2d}\sigma_1(\cm{A}^*_{(i)})$, $\underline{\sigma}=\min_{1\leq i\leq 2d}\sigma_{r_i}(\cm{A}^*_{(i)})$, and $\rho=\overline{\sigma}/\underline{\sigma}$ the largest singular value, the smallest nonzero singular value, and condition number of $\cm{A}^*$ along all modes, respectively. Now, we state a deterministic upper bound on the estimation error and a linear rate of convergence for the proposed gradient descent algorithm.

\begin{theorem}\label{thm:comp_conv}
	Suppose that the squared loss function $\overline{\mathcal{L}}$ satisfies the restricted strong convexity, restricted strong smoothness and deviation bound conditions in Definitions \ref{def:RSC_RSS} and \ref{def:deviation}, and $\cm{A}^*$ is low-rank with known Tucker ranks $(r_1,\dots,r_{2d})$. For the gradient descent iterates with parameters $a\asymp\alpha\rho^{-2}\overline{\sigma}$, $b\asymp\overline{\sigma}^{1/(2d+1)}$, and size step $\eta=\eta_0\beta^{-1}\rho^{-2}\overline{\sigma}^{(-4d)/(2d+1)}$ from some small $\eta_0>0$, if the initial bound $\|\cm{A}^{(0)}-\cm{A}^*\|_\textup{F}\lesssim\underline{\sigma}$ is satisfied, for all $i=1,2\dots,$
	\begin{equation}\label{eq:comp_conv}
		\|\cm{A}^{(i)}-\cm{A}^*\|_\textup{F}^2\lesssim\rho^2(1-C\eta_0\alpha\beta^{-1}\rho^{-2})^i\|\cm{A}^{(0)}-\cm{A}^*\|_\textup{F}^2+\rho^2\alpha^{-2}\xi^2(r_1,\dots,r_{2d}),
	\end{equation}
	with $0<C\eta_0\alpha\beta^{-1}\rho^{-2}<1$.
\end{theorem}

Theorem \ref{thm:comp_conv} presents a set of conditions for the convergence of the gradient descent iteratives for $\cm{\widehat{A}}_\text{NC}$. The first term in the right hand side of \eqref{eq:comp_conv} corresponds to the optimization error, whereas the second term corresponds to the statistical error. This bound shows that the estimation error of the iterates decreases exponentially to a statistical limit. When the RSC parameter $\alpha$ and all nonzero singular values of $\cm{A}^*$ along all modes are bounded and bounded away from zero, the rates of paramters $a,b,\eta$ remain constant.

\begin{remark}
	For the initialization, if $\underline{\sigma}$ is a constant number, the initial bound can be satisfied for any consistent intial value $\cm{A}^{(0)}$. When $p_i$'s are large, $\underline{\sigma}$ may diverge to infinity as $p_i$'s increase, and hence the initial condition for $\cm{A}^{(0)}$ could be relaxed.
\end{remark}

Following the spectral dependency measure in Section \ref{sec:HDM}, we define the restricted strong smoothness parameter for the tensor AR process
$\beta_\text{RSS}=[3\lambda_{\max}(\bm{\Sigma}_{\bm{e}})]/[2\mu_{\min}(\mathcal{A})]$.
For the tensor AR process satisfying Assumptions \ref{asmp:stationary} and \ref{asmp:gaussian}, we can derive the following statistical convergence results of the gradient descent iteratives.
\begin{theorem}\label{thm:stat_conv}
	Under Assumptions \ref{asmp:stationary}-\ref{asmp:gaussian} and conditions in Theorem \ref{thm:comp_conv} with $\alpha=\alpha_\textup{RSC}$ and $\beta=\beta_\textup{RSS}$, if $T\gtrsim\max(\kappa^4,\kappa^2)M_2^{-2}\max_{1\leq i\leq d}p_i$, when
	\begin{equation}
		I\gtrsim\log(\alpha_\textup{RSC}^{-2}T^{-1}\max_{1\leq i\leq2d}(p_ir_i)\|\cm{A}^{(0)}-\cm{A}^*\|_\textup{F}^{-2})/\log(1-C\eta_0\alpha_\textup{RSC}\beta_\textup{RSS}^{-1}\rho^{-2})
	\end{equation}
	with probability at least $1-2\exp[-CM_2^2\min(\kappa^{-2},\kappa^{-4})T]-C\exp(-C\max_{1\leq i\leq d}p_i)$,
	\begin{equation}
		\|\cm{\widehat{A}}_\textup{NC}-\cm{A}^*\|_\textup{F}\lesssim\rho\alpha_\textup{RSC}^{-1}\kappa^2M_1\sqrt{\frac{\sum_{i=1}^{2d}p_ir_i+\prod_{i=1}^{2d}r_i}{T}}.
	\end{equation}
\end{theorem}

Theorem \ref{thm:stat_conv} presents the estimation error upper bound after a sufficient number of iterations. When $\rho$, $\kappa$, $\alpha^{-1}_\textup{RSC}$ and $M_1$ are bounded, the statistical convergence rate scales as $O_p(\sqrt{(\sum_{i=1}^{2d}p_ir_i+\prod_{i=1}^{2d}r_i)/T})$. Under the exact low-rank condition in $\cm{A}^*$, compared with the SSN estimator with a rate of $O_p(\sqrt{s_0p/T})$, the rate of the non-convex NC estimator is improved significantly. In other words, to achieve consistency, the sample size requirement is reduced from $T\gtrsim p=\prod_{i=1}^dp_i$ to $T\gtrsim \max_{i=1}^dp_i$. For high-dimensional matrix-valued ($d=2$) and tensor-valued ($d\geq3$) time series data, the relaxation of sample size requirement is essential, since it is usually difficult or even impossible to collect a large number of samples, when $p$ is large as in the import-export network data discussed in Section \ref{sec:intro}.

\subsection{Rank Selection and Initialization}

In practice, we need to determine the Tucker ranks in order to apply the proposed non-convex estimation method. When the sample size is sufficiently large, i.e., $T\gtrsim p$, one may apply the TSSN method described in Section \ref{subsec:trunc} to select the ranks. When the dimensions $p_i$'s are large, we recommend giving a pre-specified upper bound $\bar{r}_i> r_i$, and then calculate the estimate $\cm{\widetilde{A}}$ based on the rank upper bounds $\bar{r}_1,\bar{r}_2,\dots,\bar{r}_{2d}$. Denote the singular values of its mode-$i$ matricization by $\widetilde{\sigma}_{i1}\geq\widetilde{\sigma}_{i2}\geq\cdots\geq\widetilde{\sigma}_{i,r_i}$, and each rank $r_i$ can be selected by the ridge-type ratio estimator \citep{wang2019high}
\begin{equation}
	\widehat{r}_i=\argmin_{1\leq j\leq \bar{r}_i-1}\frac{\widetilde{\sigma}_{i,j+1}+s(p_{\max},T)}{\widetilde{\sigma}_{ij}+s(p_{\max},T)},
\end{equation}
where $s(p_{\max},T)$ is a positive sequence depending on $p_{\max}=\max_{1\leq i\leq d}p_i$ and $T$.

The proposed method is not sensitive to the choice of $\bar{r}_i$ as long as it is greater than $r_i$. Thus, by the multidimensional factor interpretation, we can choose $\bar{r}_i$ to be reasonably large. For example, for the import-export network data described in Section \ref{sec:example}, we may set $\bar{r}_1=\cdots=\bar{r}_6=5$. The ridge parameter $s(p_{\max},T)$ is essential for consistent rank selection, and we suggest using $s(p_{\max},T)=\sqrt{p_{\max}\log(T)/(10T)}$ based on the simulation experiments in Section \ref{sec:numerical}. 
Similar to the TSSN estimator, the Tucker ranks selected by the ridge-type ratio estimator may not satisfy condition \eqref{eq:rankcond} and we may adjust the selected ranks by the approach in Remark \ref{rmk:rankadjust}.
The rank selection consistency is established in the following theorem.

\begin{theorem}\label{thm:rankconsistency2}
	Suppose that all conditions in Theorem \ref{thm:stat_conv} hold, $T\gtrsim \max(\kappa^2,\kappa^4)M_2^{-2}p_{\max}$, $\alpha_{\textup{RSC}}^{-1}\kappa^2M_1\sqrt{p_{\max}\bar{r}_{\max}/T}\ll s(p_{\max},T)$, $s(p_{\max},T)\ll \underline{\sigma}^{-1}\min_{1\leq i\leq 2d,1\leq j\leq r_i-1}\sigma_{j+1}(\cm{A}^*_{(i)})/\sigma_j(\cm{A}^*_{(i)})$, and $r_i\leq\bar{r}_i$, where $\bar{r}_{\max}=\max_{1\leq i\leq 2d}\bar{r}_i$. Then $\mathbb{P}(\widehat{r}_i=r_i)\to1$, for $i=1,\dots,2d$, as $T\to\infty$.
\end{theorem}

The conditions in this theorem reduce to $s^{-1}(p_{\max},T)\sqrt{p_{\max}/T}\to0$ and $s(p,T)\to0$ as $T\to\infty$, when $\overline{\sigma}$, $\underline{\sigma}^{-1}$, $\alpha_{\text{RSC}}^{-1}$, $\kappa$ and $M_1$ are bounded. Thus, the sample size requirement is reduced to $T\gtrsim p_{\max}$, which significantly relaxes that in Theorem \ref{thm:rankconsistency} for the TSSN method.

Moreover, for the initialization of the proposed estimation methodology, we may first select the rank upper bounds $\bar{r}_i$ and randomly initialize the algorithm by adding a random perturbation to $\cm{\widetilde{A}}$ obtained under the rank upper bounds. The refined tensor ranks are selected by the ridge-type ratio estimator, and then the HOSVD is applied to the previous initial value to obtain $\cm{A}^{(0)}$. The satisfactory performance of this initialization procedure is observed in our simulation experiments.

\section{Numerical Studies}\label{sec:numerical}

In this section, we present numerical studies to support the methodological and theoretical results obtained in the previous sections. In Section \ref{sec:sim}, we present the finite-sample performance of various estimation methods proposed in Sections \ref{sec:HDM} and \ref{sec:rank_constrained}. In Section \ref{sec:real_data}, we model the import-export network data via the LRTAR and other vector-valued and tensor-valued time series models in the literature.

\subsection{Simulation Experiments}\label{sec:sim}

We present two simulation experiments to examine the finite-sample performance of the proposed high-dimensional estimation methods. Throughout this section, we generate the data from model \eqref{eq:LTRTAR_model} with $\textup{vec}(\cm{E}_t)\overset{i.i.d.}{\sim} N(\bm{0},\bm{I}_p)$. The entries of $\cm{G}$ are generated independently from $N(0,1)$ and rescaled such that $\|\cm{G}\|_\text{F}=5$. The matrices $\bm{U}_i's$ are generated by extracting the leading singular vectors from Gaussian random matrices while ensuring the stationarity condition in Assumption \ref{asmp:stationary}. In these two experiments, we consider four cases of data generating processes. In cases (a) and (b), we consider $d=2$ and Tucker ranks $(r_1,r_2,r_3,r_4)=(1,1,1,1)$, $(2,2,1,1)$ or $(2,2,2,2)$; while in cases (c) and (d), we consider $d=3$ and Tucker ranks $(r_1,r_2,r_3,r_4,r_5,r_6)=(1,1,1,1,1,1)$, $(2,2,2,1,1,1)$ or $(2,2,2,2,2,2)$. Both pairs of cases differ in the setting for $p_i$'s: (a) $p_1=p_2=10$ and (b) $p_1=p_2=20$; (c) $p_1=p_2=p_3=7$ and (d) $p_1=p_2=p_3=15$.

The first experiment aims to compare the performance of four nuclear-norm-penalized estimators discussed in Section \ref{sec:HDM}, namely the SN, MN, SSN and TSSN estimators, when the sample size is relatively large. For each setting, we repeat 500 times and conduct the estimation using SN, MN, SSN, and TSSN. The nuclear-norm-penality tuning parameter and truncation parameter are selected by the BIC described in Section \ref{subsec:trunc}. In Figure \ref{fig:2}, the average estimation errors are plotted against $T\in\{800,1000,1200,1400\}$ for cases (a) and (b), and $T\in\{1000,1200,1400,1600\}$ for cases (c) and (d). First, it can be seen that the SN estimator is much inferior to the other three estimators, which is due to its use of the unbalanced one-mode matricizations. Secondly, the SSN and TSSN estimators outperform the MN estimator in all cases, and their advantage is remarkably clear even when $r_1=\cdots=r_{2d}$. In addition, the rank selection performance of the TSSN method is summarized in Table \ref{tbl:rank_ex1}. In general, the TSSN estimator can consistently select the tensor ranks when $T$ is large, and performs the best among these four, probably because it yields a more parsimonious model which improves the estimation efficiency. The results in experiment 1 verify the efficiency improvement in the proposed SSN and TSSN estimators.

\begin{figure}[!htp]
    \includegraphics[width=\textwidth]{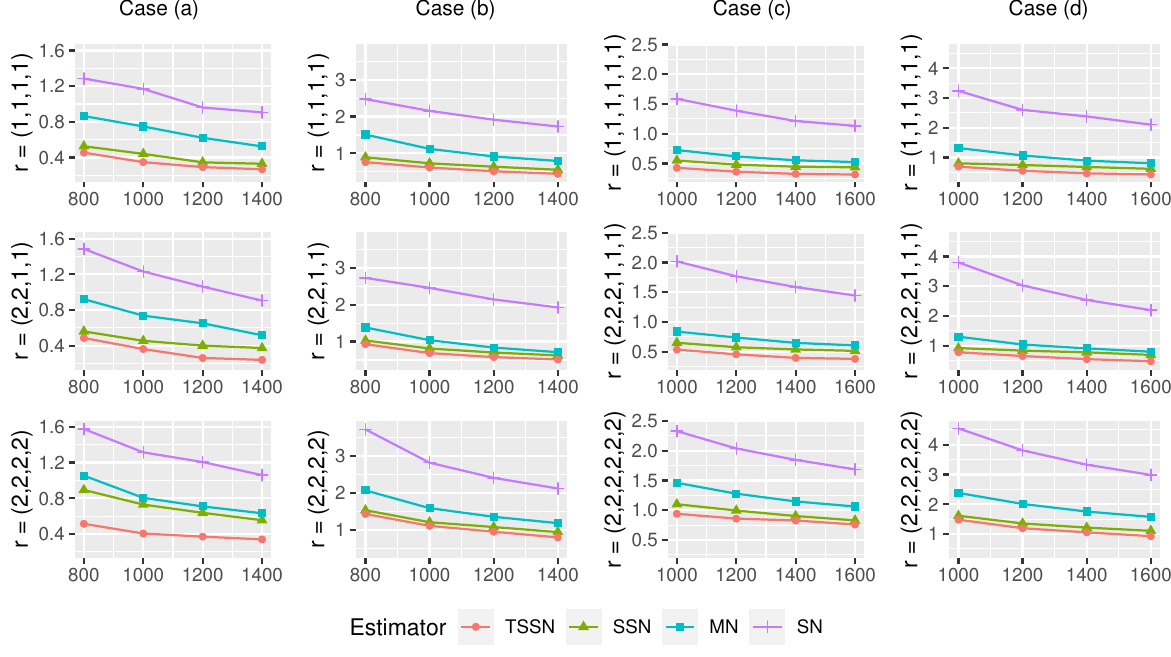}
    \caption{Average estimation error for TSSN, SSN, MN, and SN estimators for data generated with different $d$, $p_i$'s and Tucker ranks in experiment 1}
    \label{fig:2}
\end{figure}

\begin{table}[!htp]
    \begin{center}
        \caption{\small{Percentages of correct rank determination by TSSN in experiment 1}}
        \label{tbl:rank_ex1}
        \vspace{-0.4cm}
        \renewcommand{\arraystretch}{0.85}
        \small{\begin{tabular}{R{1.3cm} C{2cm} C{2cm} C{2cm} c C{2cm} C{2cm} C{2cm}}
                \hline\hline
                $d=2$& \multicolumn{3}{c}{Case (a)} && \multicolumn{3}{c}{Case (b)}\\
                \cline{2-4}\cline{6-8}
                $T\backslash\text{rank}$ & (1,1,1,1) & (2,2,1,1) & (2,2,2,2) && (1,1,1,1) & (2,2,1,1) & (2,2,2,2) \\
                \hline
                800 & 96.2 & 93.8 & 90.0 && 82.6 & 79.6 & 75.2 \\
                1000& 98.4 & 98.0 & 94.8 && 86.4 & 84.4 & 81.8 \\
                1200& 100  & 100  & 99.2 && 93.2 & 94.0 & 88.0 \\
                1400& 100  & 99.8 & 100  && 98.4 & 97.8 & 96.2 \\
                \hline
                $d=3$& \multicolumn{3}{c}{Case (c)} && \multicolumn{3}{c}{Case (d)}\\
                \cline{2-4}\cline{6-8}
                $T\backslash\text{rank}$ & (1,1,1,1,1,1) & (2,2,2,1,1,1) & (2,2,2,2,2,2) && (1,1,1,1,1,1) & (2,2,2,1,1,1) & (2,2,2,2,2,2) \\
                \hline
                1000 & 93.2 & 91.8 & 92.2 && 81.4 & 81.0 & 77.6\\
                1200 & 96.6 & 93.2 & 92.6 && 88.2 & 90.4 & 85.6\\
                1400 & 99.4 & 98.8 & 99.0 && 91.4 & 93.6 & 92.8\\
                1600 & 99.6 & 99.2 & 98.8 && 96.2 & 97.0 & 97.2\\
                \hline
        \end{tabular}}
    \end{center}
\end{table}

The second experiment aims to verify the performance of NC estimator when the sample size is relatively small. We consider $T\in(50,100,150,200)$ for cases (a) and (b), and $T\in(80,160,240,320)$ for cases (c) and (d). Since the NC estimator requires the pre-determined Tucker ranks, we consider two estimators, namely the NC estimator with the true Tucker ranks (denoted by NC-true) and NC estimator with Tucker ranks estimated by the ridge-type ratio estimator in Section \ref{sec:rank_constrained} (denoted by NC-est). When applying the gradient descent algorithm, we simply set $a=b=1$ and use the TSSN estimator to obtain the initial values of $\cm{G}$ and $\bm{U}_i$'s. The default gradient descent step size is $\eta=10^{-4}$, and it will be reduced to $10^{-5}$ if the default one fails to converge. In addition, the random initialization method is also adopted for the NC-est estimator. The average estimation errors of the non-convex methods are summarized in Figure \ref{fig:3} and the rank determination of the ridge-type ratio estimator is collected in Table \ref{tbl:rank_ex2}. As the ridge-type ratio estimator can consistently estimate the Tucker ranks, the performance of NC-true and NC-est estimators is quite similar. When the sample size is small, NC-true method performs slightly better.

\begin{figure}[!htp]
	\includegraphics[width=\textwidth]{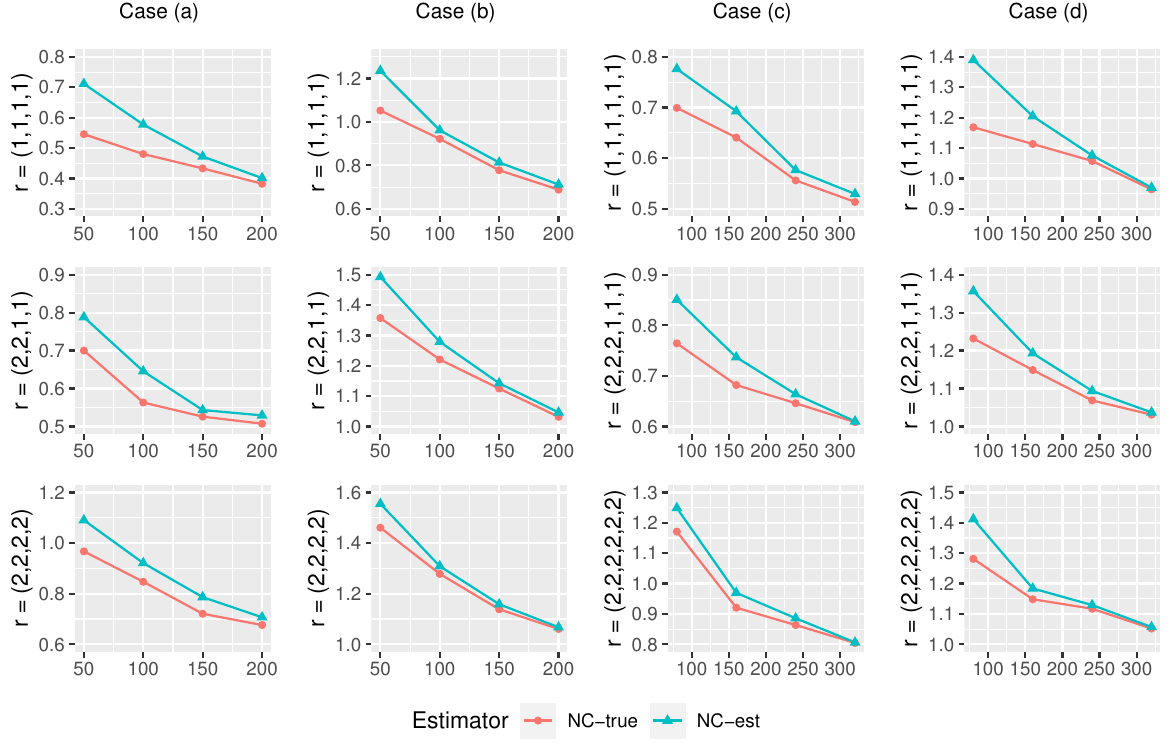}
	\caption{Average estimation error for NC-true and NC-est estimators for data generated with different $d$, $p_i$'s and Tucker ranks in experiment 2}
	\label{fig:3}
\end{figure}

\begin{table}[!htp]
	\begin{center}
		\caption{\small{Percentages of correct rank determination by ridge-type ratio estimator in experiment 2}}
		\label{tbl:rank_ex2}
		\vspace{-0.4cm}
		\renewcommand{\arraystretch}{0.85}
		\small{\begin{tabular}{R{1.3cm} C{2cm} C{2cm} C{2cm} c C{2cm} C{2cm} C{2cm}}
				\hline\hline
				$d=2$& \multicolumn{3}{c}{Case (a)} && \multicolumn{3}{c}{Case (b)}\\
				\cline{2-4}\cline{6-8}
				$T\backslash\text{rank}$ & (1,1,1,1) & (2,2,1,1) & (2,2,2,2) && (1,1,1,1) & (2,2,1,1) & (2,2,2,2) \\
				\hline
				50  & 77.2 & 69.2 & 70.4 && 69.6 & 57.4 & 55.4 \\
				100 & 82.2 & 77.2 & 78.6 && 74.2 & 62.2 & 63.8 \\
				150 & 88.4 & 85.0 & 87.2 && 82.8 & 70.6 & 72.4 \\
				200 & 94.0 & 92.4 & 93.2 && 90.8 & 87.0 & 78.8 \\
				\hline
				$d=3$& \multicolumn{3}{c}{Case (c)} && \multicolumn{3}{c}{Case (d)}\\
				\cline{2-4}\cline{6-8}
				$T\backslash\text{rank}$ & (1,1,1,1,1,1) & (2,2,2,1,1,1) & (2,2,2,2,2,2) && (1,1,1,1,1,1) & (2,2,2,1,1,1) & (2,2,2,2,2,2) \\
				\hline
				 80 & 71.2 & 69.4 & 68.4 && 68.2 & 70.4 & 73.2\\
				160 & 82.8 & 79.8 & 78.0 && 79.8 & 76.0 & 55.8\\
				240 & 88.2 & 84.0 & 84.4 && 88.2 & 87.8 & 87.8\\
				320 & 97.0 & 92.0 & 92.6 && 98.0 & 99.2 & 97.6\\
				\hline
		\end{tabular}}
	\end{center}
\end{table}

\subsection{Real Data Analysis}\label{sec:real_data}

We analyze the multi-category import-export network data in \citet{chen2019factor}, which consists of the monthly export data among 22 countries, including 19 European countries (Belgium, Bulgaria, Denmark, Finland, France, Germany, Greece, Hungary, Iceland, Ireland, Italy, Norway, Poland, Portugal, Spain, Sweden, Switzerland, Turkey, and the United Kingdom) and 3 North American countries (Canada, Mexico, and the United States). The products are classified into 15 categories, including industrial and algricultural products. Hence, the import-export network data in each month form a $22\times22\times15$ Export-Import-Product tensor, and the data is collected from January 2010 to December 2016. Following \citet{chen2019factor}, the missing diagonal values for the export from any country to itself are treated as zero. A three-month moving average of the series is applied to alleviate the possible effect of incidental transactions, so the total available sample size is $T=84$ which is much smaller than the overall dimension of the data $p=22\times22\times15=7,260$.

Let $\cm{Y}_t\in\mathbb{R}^{22\times 22\times 15}$ be the tensor-valued time series and denote $\bm{y}_t=\text{vec}(\cm{Y}_t)$. For comparison, we consider the following seven candidate models:
\begin{itemize}
	\item[1.] The proposed LRTAR model: $\cm{Y}_t=\langle\cm{A},\cm{Y}_{t-1}\rangle+\cm{E}_t$, with $\cm{A}=\cm{G}\times_{i=1}^6\bm{U}_i$. The model is estimated using the SSN, TSSN and NC methods, respectively.
	\item[2.] Sparse vector autoregression (SVAR): $\bm{y}_t=\bm{A}\bm{y}_t+\bm{e}_t$, where $\bm{A}\in\mathbb{R}^{7260\times7260}$ is a sparse matrix. We estimate the sparse VAR model via the Lasso estimator discussed in \citet{basu2015regularized}.
	\item[3.] Low-rank vector autoregression (LRVAR): $\bm{y}_t=\bm{A}\bm{y}_t+\bm{e}_t$, where $\bm{A}\in\mathbb{R}^{7260\times7260}$ is a low-rank matrix. The model is estimated by the MN estimator in Section \ref{sec:HDM}.
	\item[4.] Vector factor model (VFM): $\bm{y}_t=\bm{\Lambda}\bm{f}_t+\bm{e}_t$, where $\bm{f}_t$ is the low-dimensional vector-valued latent factor, and $\bm{\Lambda}$ is the loading matrix. The model is estimated by the method in \citet{lam2012factor}, and for prediction, the estimated factors $\widehat{\bm{f}}_t$ are then fitted by a VAR(1) model.
	\item[5.] Multilinear tensor autoregression (MTAR): $\cm{Y}_t=\cm{Y}_{t-1}\times_{i=1}^3\bm{B}_i+\cm{E}_t$, where $\bm{B}_1,\bm{B}_2\in\mathbb{R}^{22\times22}$ and $\bm{B}_3\in\mathbb{R}^{15\times 15}$ are coefficient matrices. The model is estimated by the iterative least squares method similar to \citet{chen2018autoregressive}.
	\item[6.] Tensor factor model (TFM): $\cm{Y}_t=\cm{F}_t\times_{i=1}^3\bm{U}_i+\cm{E}_t$, where $\cm{F}_t$ is the low-dimensional tensor-valued latent factor, and $\bm{U}_i$'s are the loading matrices. The TFM is esimated by the method in \citet{chen2019factor}, and for prediction, the estimated factors $\cm{\widehat{F}}_t$ are fitted by a VAR(1) model.
	\item[7.] Factor augmented vector autoregressive model (FAVAR): the vectorized time series is decomposed into two parts $\bm{y}_t=(\bm{y}_{1t}^\top,\bm{y}_{2t}^\top)^\top$, where $\bm{y}_{1t}\in\mathbb{R}^{4}$ contains the trading data between the United States and Germany under the categories of the largest volume ``Machinery and Electrical'' and ``Transportation'', and $\bm{y}_{2t}\in\mathbb{R}^{7256}$ contains the rest of data. The FAVAR model \citep{bernanke2005measuing,stock2016dynamic} with $(\bm{f}_t^\top,\bm{y}_{1t}^\top)^\top=\bm{A}(\bm{f}_{t-1}^\top,\bm{y}_{1,t-1}^\top)^\top+\bm{w}_t$ and $\bm{y}_{2t}=\bm{\Lambda}\bm{f}_t+\bm{\Gamma}\bm{y}_{1t}+\bm{g}_t$ is used to model the data.
\end{itemize}

We first focus on the results of the proposed LRTAR model. The overall dimension $p=7260$ is much larger than the sample size $T=82$, which violates the sample size requirements of nuclear-norm-regularized estimators. Hence, we try all combinations of Tucker ranks, with each rank ranging from 1 to 3, and the best ranks selected by the BIC are $(1,1,2,2,2,2)$. By the multi-dimensional dynamic factor interpretation in \eqref{eq:tenfactor}, these six ranks indicate the numbers of factors for ``export predictor'', ``import predictor'', ``product predictor'', ``export response'', ``import response'', and ``product response.'' In other words, the total number of factors for predictors ($1\times1\times2=2$) is smaller than that for responses ($2\times2\times2=8$), showing that the low-dimensional information summarized from predictors is more compact than that of responses. It is also interesting to see that the numbers of factors for predictors and responses selected by LRTAR are smaller than those selected by the tensor factor model, $(4,4,6)$, in \citet{chen2019factor}.

As the factor matrices $\bm{U}_i$'s are not uniquely defined, we   present the estimates of the identifiable projection matrices $\bm{U}_i\bm{U}_i^\top$ by LRTAR-NC with $\bm{U}_i$ being orthonormal in Figure \ref{fig:trade}. The estimated projection matrices of these six factor loadings offer a clear and interesting interpretation of inter-regional trading flow, which helps us answer the four questions in Section \ref{sec:intro}. 
For the first two questions about the driving forces of the exporting and importing activities, the estimated factor matrices $\widehat{\bm{U}}_1$, $\widehat{\bm{U}}_2$, $\widehat{\bm{U}}_4$ and $\widehat{\bm{U}}_5$ present some numerical hints.
Specifically, for the responses of export and import (first two plots in the left panel of Figure \ref{fig:trade}), the exporting countries are clearly classified into two geographical factors, one for European countries and one for North American countries, while the import countries are categorized into another two factors, United States factor and Germany factor.
For the predictors, the factor loadings for exporting and importing countries (see the first two plots in the right panel of Figure \ref{fig:trade}) are both dominated by the United States. In other words, to forecast the trading volume in Europe and North America, the historical trading data of the United States, in both import and export, are most predictive. However, the future import and export value have a clear geographical grouping pattern. 

In addition, for the third question, the factor loadings for product category, $\widehat{\bm{U}}_3$ and $\widehat{\bm{U}}_6$, also have a clear grouping pattern. For both responses and predictors, the product categories can be classified into two factors, ``heavy industry factor'' (mineral, chemical, machinery, electrical and transportation products) and ``agricultural and light industry factor'' (animal, vegetable, leather, wood, textiles products). Hence, we may interpret the estimated factor matrices in LRTAR as variable grouping patterns in export, import, and product categories for responses and predictors, respectively.
Finally, by comparing the predictor and response factor loadings, we observe that the geographical grouping patterns of both exporting and importing countries are significantly different between past and present states (i.e., predictor vs. response), whereas the grouping patterns of product categories remain almost the same.

Next, we compare the forecasting performance of seven candidate models through both average in-sample and out-of-sample forecasting errors. The average in-sample forecasting error is calculated based on the fitted models for the entire data, while the average out-of-sample forecasting error is calculated based on the rolling forecast procedure as follows. From January 2015 ($t=59$) to December 2016 ($t=84)$, we fit the models using all the available data until time $t$ and obtain the one-step-ahead forecast $\cm{\widehat{Y}}_t$. Then, we obtain the average of the rolling forecasting errors, excluding the missing diagonal entries. The number of parameters in each candidate model (LRTAR-SSN is excluded because it produces shrinkage of singular values instead of exactly low-rank structure) and the average in-sample and out-of-sample forecasting errors are summarized in Table \ref{tbl:forecast_error}.

\begin{table}
    \begin{center}
        \small{\begin{tabular}{cc|rrrrrrrrrrr}
                \hline\hline
                &\multirow{2}{*}{Model} && \multicolumn{3}{c}{LRTAR} && \multirow{2}{*}{SVAR} & \multirow{2}{*}{LRVAR} &\multirow{2}{*}{VFM} & \multirow{2}{*}{MTAR} & \multirow{2}{*}{TFM} & \multirow{2}{*}{FAVAR}\\
                \cline{4-6}
                &&&SSN & TSSN & NC &&&&&&&\\
                \hline
                \multicolumn{2}{r|}{No. of par.} && NA & 2643 & 190 && 9543 & 43560 & 21789 & 2386 & 2525 & 36305\\
                \hline
                \multirow{2}{*}{\rotatebox[origin=c]{90}{IS}}& $\ell_2$ norm && 1362 & 1409 & 1563 && \bf{713} & 896 & 906 & 1263 & 1076 & 943 \\
                & $\ell_\infty$ norm && 74 & 79 & 89 && 62 & 88 & \textbf{61} & 83 & 79 & 67\\
                \hline
                \multirow{2}{*}{\rotatebox[origin=c]{90}{OOS}}& $\ell_2$ norm && 2018 & 1533 & \textbf{1083} && 2362 & 2218 & 1545 & 1432 & 1211 & 1224\\
                & $\ell_\infty$ norm && 127 & 109 & \textbf{99} && 176 & 218 & 134 & 123 & 114 & 119\\
                \hline
        \end{tabular}}
        \caption{Number of parameters (No. of par.), average in-sample (IS) forecasting error and out-of-sample (OOS) rolling forecasting error for the import-export network data by various models and methods. The best cases are marked in bold.}
        \label{tbl:forecast_error}
    \end{center}
\end{table}

As shown in Table \ref{tbl:forecast_error}, all vector time series models have smaller in-sample forecasting errors and larger out-of-sample forecasting errors than their tensor counterparts, as they fail to utilize the multi-dimensional structure of the tensor data. For out-of-sample forecasting, LRTAR-NC significantly outperforms the other models in terms of average and maximum errors, as this model is much more parsimonious and can prevent overfitting effectively.

\begin{figure}
	\includegraphics[width=\textwidth]{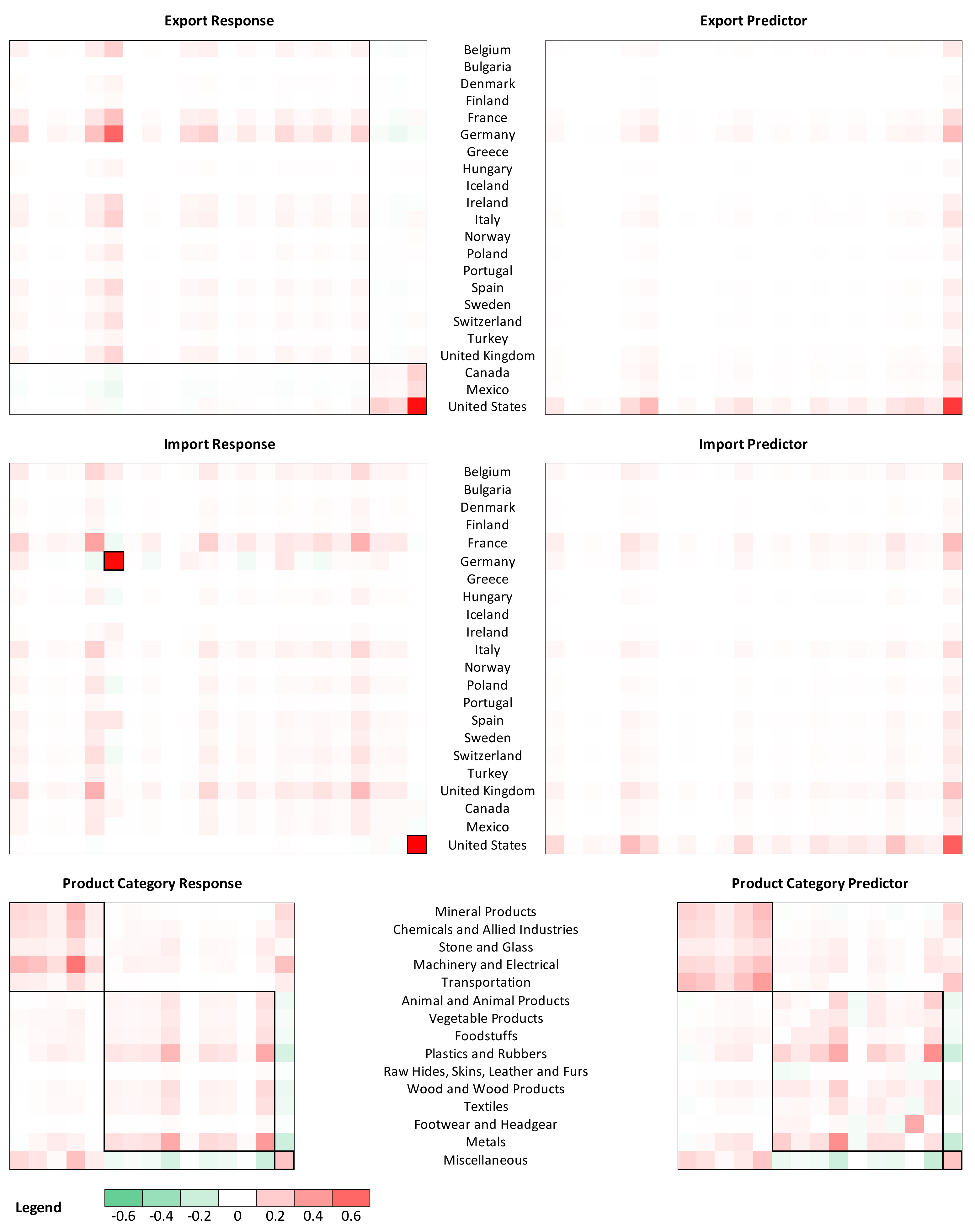}
	\caption{NC estimates of response and predictor factor projection matrices $\bm{U}_i\bm{U}_i^\top$.}
	\label{fig:trade}
\end{figure}

\section{Conclusion and Discussion}\label{sec:conclusion}

Efficient modeling and forecasting of high-dimensional tensor time series data is an important and emerging research topic. This paper makes the first thorough attempt to address this problem from the perspective of autoregressive modeling. By assuming the exact or approximately low-Tucker-rank structure of the transition tensor, the model exploits the low-dimensional tensor dynamic structure of the high-dimensional time series data, and summarizes the complex temporal dependencies into interpretable dynamic factors.

Under the high-dimensional setting, we investigate two estimation approaches, nuclear-norm-regularized methods and non-convex methods. For the former, based on the special structure of the transition tensor, a novel convex regularizer, the SSN, is proposed, gaining efficiencies from both the square matricization and simultaneous penalization across modes. For the latter, an integrated computational and statistical analysis is provided for the gradient descent algorithm. The nuclear-norm-regularized estimators can handle the general case with approximate low-rankness, and the non-convex estimator gains efficiency improvement under the exactly low-rank setting.

We discuss several directions for future research. First, in addition to the low-rank models, sparse plus low-rank models \citep{basu2019low,miao2023high} have been extensively studied in the literature of high-dimensional vector autoregression. It is also of interest to extend the proposed model in this direction, i.e., the parameter tensor $\cm{A}$ can be decomposed into two components, the low-rank component $\cm{L}$ and sparse component $\cm{S}$. Specically, $\cm{L}$ is low-Tucker-rank as we discuss in this paper, and $\cm{S}$ can capture the additional sparse autoregressive relationship between responses and predictors.

Second, while this paper focuses on the pure autoregressive model, the fundamental idea of leveraging the tensor-valued data and imposing the low-Tucker-rank assumption for dimension reduction can be extended to more complex settings. Similar to panel data models, exogenous variables can be further added into the regression, resulting in LRTAR-X models. For example, for the multi-category import-export data in Section \ref{sec:real_data}, it is possible to consider
$\cm{Y}_t=\langle\cm{A},\cm{Y}_{t-1}\rangle+\boldsymbol{\beta}^\prime \bm{x}_t+\langle\bm{B},\cm{X}_{t}\rangle+\cm{E}_t$, where the vector $\bm{x}_t$ may contain global variables such as the return of the oil price, and the matrix or tensor $\cm{X}_{t}$ may contain other country-level macroeconomic indicators such as the GDP growth rate.
When the dimensions of $\bm{x}_t$ and $\cm{X}_t$ are high, a low-dimensional structure, such as sparsity, group sparsity or low-rankness, can be imposed on $\boldsymbol{\beta}$ and $\cm{B}$ to improve the estimation efficiency.

Third, in the proposed  model, all variables in $\cm{Y}_t$ are treated with equal importance because the primary objective is to capture the complex dependence structures of a global system using granular data. However, if there are other priority variables to forecast, represented by the vector $\bm{x}_t\in\mathbb{R}^{p_x}$, then we may extend the proposed method to the joint model,
\begin{align*}
\cm{Y}_t & =\langle\cm{A}_y,\cm{Y}_{t-1}\rangle+\boldsymbol{\beta}_y^\prime \bm{x}_t+\cm{E}_t\\
\bm{x}_t &= \langle\cm{A}_x,\cm{Y}_{t-1}\rangle+\boldsymbol{\beta}_x^\prime \bm{x}_{t-1}+\cm{E}_t,
\end{align*}
where $\cm{A}_x\in\mathbb{R}^{p_x\times p_1\times\cdots\times p_d}$ can be assumed to have low Tucker ranks. 

Fourth, the proposed methods can be generalized to the LRTAR model of finite lag order $L$, defined as
$\cm{Y}_t=\langle\cm{A}_1,\cm{Y}_{t-1}\rangle+\cdots+\langle\cm{A}_L,\cm{Y}_{t-L}\rangle+\cm{E}_t$,
where $\cm{A}_1, \dots, \cm{A}_L$ are $2d$-th-order Tucker low-rank coefficient tensors. Then, one may consider the SSN regularized estimation by minimizing
$ T^{-1}\sum_{t=1}^T\|\cm{Y}_t-\sum_{j=1}^L\langle\cm{A}_j,\cm{Y}_{t-j}\rangle\|_{\text{F}}^2+\sum_{j=1}^L\lambda_j\|\cm{A}_j\|_{\text{SSN}}$.
In addition, 
the NC estimator can be defined as the  minimizer of $(2T)^{-1}\sum_{t=1}^T\|\cm{Y}_t-\langle\cm{G}\times_{i=1}^{2d}\bm{U}_i,\cm{Y}_{t-1}\rangle\|_\text{F}^2+(a/2)\sum_{i=1}^{2d}\|\bm{U}_i^\top\bm{U}_i-b^2\bm{I}_{r_i}\|_\text{F}^2$, which can be implemented by the gradient descent algorithm.

Finally,  heavy-tailed distributions and outliers are commonly observed in empirical economic and financial datasets. Robust estimation methods against the heavy-tailed distribution for high-dimensional VAR models have been investigated recently \citep{liu2021robust,wang2021robust}, and it is of practical importance to investigate the robust methods for the proposed model.

\section*{Acknowledgement}

We are  grateful to the editor, Serena Ng, the associate editor, and three anonymous referees for their valuable comments that led to the substantial improvement of this paper. We would also like to thank Dan Yang for sharing the multi-category import-export network data for the empirical analysis. Wang was supported by the National Natural Science Foundation of China Grant 12301352 and Shanghai Sailing Program for Youth Science and Technology Excellence (23YF1420300). Zheng was supported by the National Science Foundation Grant DMS-2311178. Li was supported by the Hong Kong Research Grant Council Grants 17306519 and 17313722.

\bibliography{mybib}
\bibliographystyle{apalike}

\clearpage
\newpage
\begin{center}
	{\Large \bf Supplementary Material for \\ ``High-Dimensional Low-Rank Tensor Autoregressive Time Series Modeling"}
\end{center}

\begin{abstract}		
	This supplementary material provides all technical proofs and details about the algorithms for the proposed LTR and (T)SSN estimators. To be specific, Appendix \ref{append:high-dim} presents the proofs of theoretical results for the nuclear-norm-regularized estimators in Section \ref{sec:HDM} of the main paper, while Appendix \ref{append:OR} gives the proofs of the non-convex approach in Section \ref{sec:rank_constrained}. Appendix \ref{append:ADMM} presents the ADMM algorithm for the (T)SSN estimator. Finally, Appendix \ref{append:discuss}  discusses two special cases of the proposed LRTAR model and their connections with some existing models in the literature.
\end{abstract}

\renewcommand{\thesection}{S\arabic{section}}
\renewcommand{\theequation}{S\arabic{equation}}
\renewcommand{\thelemma}{S\arabic{lemma}}
\setcounter{section}{0}
\setcounter{equation}{0}
\setcounter{lemma}{0}
\setcounter{figure}{0}
\setcounter{table}{0}

\section{Proofs for Convex Regularized Estimation}\label{append:high-dim}

In this appendix, we provide the proofs of Theorems \ref{thm:SN}--\ref{thm:rankconsistency} in Section \ref{sec:HDM}. We start with a preliminary analysis in Appendix \ref{sec:prelimproof} which lays out the common technical framework for proving the estimation and prediction error bounds of the SN, MN and SSN regularized estimators, and four lemmas, Lemmas \ref{lemma:ApproxCone}--\ref{lemma:RSC}, are introduced herein.
Then in Appendix \ref{subsec:proof_theorem} we give the proofs of Theorems \ref{thm:SN}--\ref{thm:rankconsistency}. The proofs of Lemmas \ref{lemma:ApproxCone}--\ref{lemma:RSC} are provided in Appendix \ref{subsec:lemma}, and three  auxiliary lemmas are collected  in Appendix \ref{subsec:auxlemma}

\subsection{Preliminary Analysis \label{sec:prelimproof}}
The technical framework for proving the error bounds in Theorem \ref{thm:SN}--\ref{thm:SSN} consists of two main steps, a deterministic analysis  and a stochastic analysis, given in Sections \ref{subsec:deterministic} and \ref{subsec:stochastic}, respectively. The goal of the first one is to derive the error bounds given the deterministic realization of the time series, assuming that the parameters satisfy certain regularity conditions. The goal of the second one is to verify that  under stochasticity these regularity conditions are satisfied with high probability.

\subsubsection{Deterministic Analysis \label{subsec:deterministic}}

Throughout the appendix, we adopt the following notations. We use $C$ to denote a generic positive constant, which is independent of the dimensions and the sample size.  For any matrix $\bm{M}$ and a compatible subspace $\mathcal{S}$, we denote by $\bm{M}_{\mathcal{S}}$ the projection of $\bm{M}$ onto $\mathcal{S}$. In addition, let $\text{col}(\bm{M})$ be the column space of $\bm{M}$, and let $\mathcal{S}^\perp$ be the complement of the subspace $\mathcal{S}$. For a generic tensor $\cm{W}\in\mathbb{R}^{p_1\times\cdots\times p_{2d}}$, the dual norms of its SSN norm and SN norm, denoted by $\|\cm{W}\|_{\text{SSN}^*}$ and $\|\cm{W}\|_{\text{SN}^*}$, respectively, are defined as
\begin{equation}
	\begin{split}
		\|\cm{W}\|_{\text{SSN}^*}=\sup_{\scalebox{0.7}{\cm{T}}\in\mathbb{R}^{p_1\times\cdots\times p_{2d}},\|\scalebox{0.7}{\cm{T}}\|_{\text{SSN}}\leq1}\langle\cm{W},\cm{T}\rangle,~~\text{and}~~\|\cm{W}\|_{\text{SN}^*}=\sup_{\scalebox{0.7}{\cm{T}}\in\mathbb{R}^{p_1\times\cdots\times p_{2d}},\|\scalebox{0.7}{\cm{T}}\|_{\text{SN}}\leq1}\langle\cm{W},\cm{T}\rangle.
	\end{split}
\end{equation}
Moreover, for any two tensors $\cm{X}\in\mathbb{R}^{p_1\times\cdots\times p_m}$ and $\cm{Y}\in\mathbb{R}^{p_{m+1}\times\cdots\times p_{m_n}}$, their tensor outer product is defined as $(\cm{X}\circ\cm{Y})\in\mathbb{R}^{p_1\times\cdots\times p_m\times p_{m+1}\times\cdots\times p_{m+n}}$ where
\begin{equation}
	(\cm{X}\circ\cm{Y})_{i_1\dots i_{m}i_{m+1}\dots i_{m+n}}=\cm{X}_{i_1\dots i_m}\cm{Y}_{i_{m+1}\dots i_{m+n}},
\end{equation}
for any $1\leq i_1\leq p_1$, $\dots$, $1\leq i_{m+n}\leq p_{m+n}$.

For the theory of regularized $M$-estimators, restricted error sets  and restricted strong convexity are essential definitions. To define the former, we need to first introduce the following restricted model subspaces.

For $i=1,\dots,2d$,  denote by $\widetilde{\mathcal{U}}_i$ and $\widetilde{\mathcal{V}}_i$ the spaces spanned by the first $r_i$ left and right singular vectors in the SVD of $\cm{A}_{(i)}$, respectively.  
Define the  collections of subspaces
\begin{equation*}
	\mathcal{N}=(\mathcal{N}_1,\dots,\mathcal{N}_{2d}) \quad \text{and} \quad  \overline{\mathcal{N}}^\perp=(\overline{\mathcal{N}}_1^\perp,\dots,\overline{\mathcal{N}}_{2d}^\perp),
\end{equation*}
where
\begin{equation}\begin{split}
		\mathcal{N}_i&=\{\bm{M}\in\mathbb{R}^{p_i\times p_{-i}p}|\text{col}(\bm{M})\subset\widetilde{\mathcal{U}}_i,\text{col}(\bm{M}^\top)\subset\widetilde{\mathcal{V}}_i\},\\
		\overline{\mathcal{N}}_i^\perp&=\{\bm{M}\in\mathbb{R}^{p_i\times p_{-i}p}|\text{col}(\bm{M})\perp\widetilde{\mathcal{U}}_i,\text{col}(\bm{M}^\top)\perp\widetilde{\mathcal{V}}_i\}, \label{eq:subspace_N}
\end{split}\end{equation}
for $i=1,\dots,2d$.
Note that $\mathcal{N}_i\subset\overline{\mathcal{N}}_i$.

Furthermore, for $k=1,\dots, 2^{d-1}$, denote by $\mathcal{U}_{k}$ and $\mathcal{V}_k$ the spaces spanned by the first $s_k^*$ left and right singular vectors in the SVD of the square matricization $\cm{A}^*_{[I_k]}$, respectively, where $s_k^*=\textup{rank}(\cm{A}_{[I_k]}^*)$. Similarly, define the collections of subspaces
\begin{equation*}
	\mathcal{M}:=(\mathcal{M}_1,\dots,\mathcal{M}_{2^{d-1}})
	\quad \text{and} \quad 
	\overline{\mathcal{M}}^\perp=(\overline{\mathcal{M}}_1^\perp,\dots,\overline{\mathcal{M}}_{2^{d-1}}^\perp),
\end{equation*}
where
\begin{equation}\begin{split}
		&\mathcal{M}_k=\{\bm{M}\in\mathbb{R}^{p\times p}|\text{col}(\bm{M})\subset \mathcal{U}_k,~\text{col}(\bm{M}^\top)\subset \mathcal{V}_k\},\\
		&\overline{\mathcal{M}}_k^\perp=\{\bm{M}\in\mathbb{R}^{p\times p}|\text{col}(\bm{M})\perp\mathcal{U}_k,~\text{col}(\bm{M}^\top)\perp \mathcal{V}_k\}, \label{eq:subspace_M}
\end{split}\end{equation}
for $k=1,\dots,2^{d-1}$.  In particular, as described in Section \ref{subsec:sqrmode},  $I_1=S_1=\{1,\dots,d\}$.  Thus, $\mathcal{M}_1$ and $\overline{\mathcal{M}}_1^\perp$ are the subspaces associated with the square matricization $\cm{A}^*_{[S_1]}$.

Then, for simplicity, for any $\cm{W}\in\mathbb{R}^{p_1\times\cdots\times p_{2d}}$, we denote
\begin{equation}\begin{split}
		&\cm{W}_{\mathcal{N}}^{(i)}=(\cm{W}_{(i)})_{\mathcal{N}_i},\quad \cm{W}_{\mathcal{N}^\perp}^{(i)}=(\cm{W}_{(i)})_{\mathcal{N}_i^\perp}, \quad
		\cm{W}_{\overline{\mathcal{N}}}^{(i)}=(\cm{W}_{(i)})_{\overline{\mathcal{N}}_i}, \quad
		\cm{W}_{\overline{\mathcal{N}}^\perp}^{(i)}=(\cm{W}_{(i)})_{\overline{\mathcal{N}}_i^\perp}\\
		&\cm{W}_{\mathcal{M}}^{(k)}=(\cm{W}_{[I_k]})_{\mathcal{M}_k}, \quad \cm{W}_{\mathcal{M}^\perp}^{(k)}=(\cm{W}_{[I_k]})_{\mathcal{M}_k^\perp}, \quad
		\cm{W}_{\overline{\mathcal{M}}}^{(k)}=(\cm{W}_{[I_k]})_{\overline{\mathcal{M}}_k}, \quad
		\cm{W}_{\overline{\mathcal{M}}^\perp}^{(k)}=(\cm{W}_{[I_k]})_{\overline{\mathcal{M}}_k^\perp},
\end{split}\end{equation}
where $i=1,\dots,2d$ and $k=1,\dots,2^{d-1}$. Based on the subspaces defined in \eqref{eq:subspace_N} and \eqref{eq:subspace_M}, we can define the  restricted error sets corresponding to the three regularized estimators as follows.

%

\begin{definition}
	The restricted error set corresponding to $\overline{\mathcal{M}}$ is defined as
	\begin{equation}
		\mathbb{C}_{\textup{SSN}}(\overline{\mathcal{M}}):=\left\{\bm{\Delta}\in\mathbb{R}^{p_1\times\cdots\times p_{2d}}:\sum_{k=1}^{2^{d-1}}\|\bm{\Delta}_{\overline{\mathcal{M}}^\perp}^{(k)}\|_\textup{nuc}\leq
		3\sum_{k=1}^{2^{d-1}}\|\bm{\Delta}^{(k)}_{\overline{\mathcal{M}}}\|_\textup{nuc}+4\sum_{k=1}^{2^{d-1}}\|\cm{A}^{*(k)}_{\mathcal{M}^\perp}\|_\textup{nuc}\right\}.
	\end{equation}
	The restricted error set corresponding to $\overline{\mathcal{N}}$ is defined as
	\begin{equation}
		\mathbb{C}_{\textup{SN}}(\overline{\mathcal{N}}):=\left\{\bm{\Delta}\in\mathbb{R}^{p_1\times\cdots\times p_{2d}}:\sum_{i=1}^{2d}\|\bm{\Delta}_{\overline{\mathcal{N}}^\perp}^{(i)}\|_\textup{nuc}\leq
		3\sum_{i=1}^{2d}\|\bm{\Delta}^{(i)}_{\overline{\mathcal{N}}}\|_\textup{nuc}+4\sum_{i=1}^{2d}\|\cm{A}^{*(i)}_{\mathcal{N}^\perp}\|_\textup{nuc}\right\}.
	\end{equation}
	The restricted error set corresponding to $\overline{\mathcal{M}}_1$ is defined as
	\begin{equation}
		\mathbb{C}_{\textup{MN}}(\overline{\mathcal{M}}_1):=\left\{\bm{\Delta}\in\mathbb{R}^{p_1\times\cdots\times p_{2d}}:\|\bm{\Delta}_{\overline{\mathcal{M}}^\perp}^{(1)}\|_\textup{nuc}\leq
		3\|\bm{\Delta}^{(1)}_{\overline{\mathcal{M}}}\|_\textup{nuc}+4\|\cm{A}^{*(1)}_{\mathcal{M}^\perp}\|_\textup{nuc}\right\}.
	\end{equation}
\end{definition}

The first lemma shows that if the tuning parameter is well chosen for each regularized estimator, the estimation error belongs to the corresponding restricted error set.

\begin{lemma} \label{lemma:ApproxCone}
	For the SSN estimator, if the regularization parameter $\lambda_{\textup{SSN}}\geq4\|T^{-1}\sum_{t=1}^T\cm{Y}_{t-1}\circ\cm{E}_t\|_{\textup{SSN}^*}$, the error $\bm{\Delta}_{\textup{SSN}}=\cm{\widehat{A}}_{\textup{SSN}}-\cm{A}$ belongs to the set $\mathbb{C}_{\textup{SSN}}(\overline{\mathcal{M}})$. 
	
	For the SN estimator, if the regularization parameter $\lambda_{\textup{SN}}\geq4\|T^{-1}\sum_{t=1}^T\cm{Y}_{t-1}\circ\cm{E}_t\|_{\textup{SN}^*}$, the error $\bm{\Delta}_{\textup{SN}}=\cm{\widehat{A}}_{\textup{SN}}-\cm{A}$ belongs to the set $\mathbb{C}_{\textup{SN}}(\overline{\mathcal{N}})$.
	
	For the MN estimator, if the regularization parameter $\lambda_{\textup{MN}}\geq4\|T^{-1}\sum_{t=1}^T\textup{vec}(\cm{Y}_{t-1})\textup{vec}(\cm{E}_t)^\top\|_\textup{nuc}$, the error $\bm{\Delta}_{\textup{MN}}=\cm{\widehat{A}}_{\textup{MN}}-\cm{A}$ belongs to the set $\mathbb{C}_{\textup{MN}}(\overline{\mathcal{M}}_{(1)})$.
\end{lemma}

Following \citet{negahban2012restricted} and \citet{negahban2012unified}, a restricted strong convexity (RSC) condition for the square loss function can be defined as follows.

\begin{definition} \label{dfn:rsc}
	The loss function satisfies the RSC condition with curvature $\alpha_{\textup{RSC}}>0$ and restricted error set $\mathbb{C}$, if
	\begin{equation}
		\frac{1}{T}\sum_{t=1}^T	\|\langle\bm{\Delta},\cm{Y}_{t-1}\rangle\|_\text{F}^2\geq\alpha_{\textup{RSC}}\|\bm{\Delta}\|_{\textup{F}}^2,~~~\forall\bm{\Delta}\in\mathbb{C}.
	\end{equation}
\end{definition}

Based on the restricted error sets and RSC conditions, the estimation errors have the following deterministic upper bounds.

\begin{lemma} \label{lemma:errorbound}
	Suppose that $\lambda_{\textup{SSN}}\geq4\|T^{-1}\sum_{t=1}^T\cm{Y}_{t-1}\circ\cm{E}_t\|_{\textup{SSN}^*}$, the RSC condition holds with the parameter $\alpha_{\textup{RSC}}$ and restricted error set $\mathbb{C}_{\textup{SSN}}(\overline{\mathcal{M}})$, and $\cm{A}^*_{[I_k]}\in\mathbb{B}_q(s_q^{(k)};p,p)$ for some $q\in[0,1)$ and all $k=1,\dots,2^{d-1}$,
	\begin{equation}
		\|\bm{\Delta}_{\textup{SSN}}\|_\textup{F}\lesssim\sqrt{s_q}\left(\frac{2^{d-1}\lambda_{\textup{SSN}}}{\alpha_{\textup{RSC}}}\right)^{1-q/2},
	\end{equation}
	where $s_q=2^{1-d}\sum_{k=1}^{2^{d-1}}s_q^{(k)}$.
	
	Suppose that $\lambda_{\textup{SN}}\geq4\|T^{-1}\sum_{t=1}^T\cm{Y}_{t-1}\circ\cm{E}_t\|_{\textup{SN}^*}$, the RSC condition holds with the parameter $\alpha_{\textup{RSC}}$ and restricted error set $\mathbb{C}_{\textup{SN}}(\overline{\mathcal{N}})$, and $\cm{A}^*_{(i)}\in\mathbb{B}_q(r_q^{(i)};p_i,p_{-i}p)$ for some $q\in[0,1)$ and all $i=1,\dots,2d$,
	\begin{equation}
		\|\bm{\Delta}_{\textup{SN}}\|_\textup{F}\lesssim\sqrt{r_q}\left(\frac{2d\cdot\lambda_{\textup{SN}}}{\alpha_{\textup{RSC}}}\right)^{1-q/2},
	\end{equation}
	where $r_q=(2d)^{-1}\sum_{i=1}^{2d}r_q^{(i)}$.
	
	Suppose that $\lambda_{\textup{MN}}\geq4\|T^{-1}\sum_{t=1}^T\textup{vec}(\cm{Y}_{t-1})\textup{vec}(\cm{E}_t)\|_\textup{nuc}$, the RSC condition holds with the parameter $\alpha_{\textup{RSC}}$ and restricted error set $\mathbb{C}_{\textup{MN}}(\overline{\mathcal{M}}_1)$, and $\cm{A}^*_{[S_1]}\in\mathbb{B}_q(s_q^{(1)};p,p)$ for some $q\in[0,1)$,
	\begin{equation}
		\|\bm{\Delta}_{\textup{MN}}\|_\textup{F}\lesssim\sqrt{s_q^{(1)}}\left(\frac{\lambda_{\textup{MN}}}{\alpha_{\textup{RSC}}}\right)^{1-q/2}.
	\end{equation}
\end{lemma}
Note that Lemma \ref{lemma:errorbound} is deterministic and the radius $s_q$, $r_q$, and $s_q^{(1)}$ can also diverge to infinity.

\subsubsection{Stochastic Analysis \label{subsec:stochastic}}
We continue with the stochastic analysis to show that the deviation bound and the RSC condition hold simultaneously with high probability. 


\begin{lemma}[Deviation bound]
	\label{lemma:dualnorm}
	Suppose that Assumptions \ref{asmp:stationary} and \ref{asmp:gaussian} hold. If $T\gtrsim p$ and $\lambda_{\textup{SSN}}\gtrsim \kappa^2M_12^{1-d}\sqrt{p/T}$, with probability at least $1-\exp[-C(p-d)]$,
	\begin{equation}
		\left\|\frac{1}{T}\sum_{t=1}^T\cm{Y}_{t-1}\circ\cm{E}_t\right\|_{\textup{SSN}^*}\leq\frac{\lambda_{\textup{SSN}}}{4}
	\end{equation}
	where $M_1=\lambda_{\max}(\bm{\Sigma_e})/\mu^{1/2}_{\min}(\mathcal{A})$.
	
	If  $T\gtrsim \max_{1\leq i\leq d}p_{-i}p$ and $\lambda_{\textup{SN}}\gtrsim \kappa^2M_1d^{-2}\sum_{i=1}^d\sqrt{p_{-i}p/T}$, with probability at least $1-2\sum_{i=1}^d\exp(-Cp_{-i}p)$,
	\begin{equation}
		\left\|\frac{1}{T}\sum_{t=1}^T\cm{Y}_{t-1}\circ\cm{E}_t\right\|_{\textup{SN}^*}\leq\frac{\lambda_{\textup{SN}}}{4}.
	\end{equation} 
	
	Moreover, if $T\gtrsim p$ and $\lambda_{\textup{MN}}\gtrsim \kappa^2M_1\sqrt{p/T}$, with probability at least $1-\exp(-Cp)$,
	\begin{equation}
		\left\|\frac{1}{T}\sum_{t=1}^T\textup{vec}(\cm{Y}_{t-1})\textup{vec}(\cm{E}_t)^\top\right\|_\textup{op}\leq\frac{\lambda_{\textup{MN}}}{4}.
	\end{equation} 
\end{lemma}

Next, we prove the restricted strong convexity for regularized estimators. According to Lemma \ref{lemma:dualnorm}, we need the sample size  $T\gtrsim p$ for all three estimators. In this case, we can establish the strong convexity condition that is stronger than the RSC condition.

\begin{lemma}[Strong convexity]
	\label{lemma:RSC}
	Under Assumptions \ref{asmp:stationary} and \ref{asmp:gaussian}, for $T\gtrsim\max(\kappa^2,\kappa^4)M_2^{-2}p$, with probability at least $1-\exp[-C\min(\kappa^{-2},\kappa^{-4})M_2^2p]$,
	\begin{equation}
		\frac{1}{T}\sum_{t=1}^T	\|\langle\bm{\Delta},\cm{Y}_{t-1}\rangle\|_\textup{F}^2\geq\alpha_{\textup{RSC}}\|\bm{\Delta}\|_{\textup{F}}^2,
	\end{equation}
	where $M_2=[\lambda_{\min}(\bm{\Sigma}_{\bm{e}})\mu_{\max}(\mathcal{A})]/[\lambda_{\max}(\bm{\Sigma}_{\bm{e}})\mu_{\min}(\mathcal{A})]$ and $\alpha_{\textup{RSC}}=\lambda_{\min}(\bm{\Sigma}_{\bm{e}})/(2\mu_{\max}(\mathcal{A}))$.
	
\end{lemma}

\subsection{Proofs of Theorems \ref{thm:SN}--\ref{thm:rankconsistency} \label{subsec:proof_theorem}}

\begin{proof}[Proof of Theorems \ref{thm:SN} and \ref{thm:MN}]	
	Theorems \ref{thm:SN} and \ref{thm:MN} can be proved based on Lemmas \ref{lemma:errorbound}--\ref{lemma:RSC} following the same line of the proof of Theorem \ref{thm:SSN} given below. Therefore, we omit the details here.	
\end{proof}

\begin{proof}[Proof of Theorem \ref{thm:SSN}]
	
	The proof of Theorem \ref{thm:SSN} has been split into Lemmas \ref{lemma:errorbound}--\ref{lemma:RSC}.	
	By Lemma \ref{lemma:errorbound}, for deterministic realization with sample size $T$ of a tensor autoregressive process, if we choose $\lambda_{\text{SSN}}\geq4\|T^{-1}\sum_{t=1}^T\cm{Y}_{t-1}\circ\cm{E}_t\|_{\text{SSN}^*}$ and RSC condition holds for the square loss with the parameter $\alpha_{\text{RSC}}$, the following error upper bound can be established
	\begin{equation}
		\|\bm{\Delta}\|_\textup{F}\lesssim\sqrt{s_q}\left(\frac{2^{d-1}\lambda_{\textup{SSN}}}{\alpha_{\textup{RSC}}}\right)^{1-q/2}.
	\end{equation}
	Denote the events $E_1(\beta)=\{\beta\geq4\|T^{-1}\sum_{t=1}^T\cm{Y}_{t-1}\circ\cm{E}_t\|_{\text{SSN}^*}\}$ and $E_2(\alpha)=\{\lambda_{\min}(\bm{X}\bm{X}^\top/T)\geq\alpha\}$.
	If we take $\lambda_{\text{SSN}}\gtrsim \kappa^2M_12^{1-d}\sqrt{p/T}$, it suffices to show that $E_1(C\kappa^2M_12^{1-d}\sqrt{p/T})$ and $E_2(\alpha_{\text{RSC}}/2)$ occur simultaneously with high probability.
	
	By Lemma \ref{lemma:dualnorm}, when $T\gtrsim p$,
	\begin{equation}
		\left\|\frac{1}{T}\sum_{t=1}^T\cm{Y}_{t-1}\circ\cm{E}_t\right\|_{\text{SSN}^*}\lesssim\kappa^2M_1  2^{1-d} \sqrt{\frac{p}{T}}
	\end{equation}
	with probability at least $1-\exp[-C(p-d)]$. 
	
	By Lemma \ref{lemma:RSC}, when $T\gtrsim\max(\kappa^2,\kappa^4)M_2^{-2}p$, for any $\bm{\Delta}\in\mathbb{R}^{p_1\times\cdots\times p_{2d}}$,
	\begin{equation}
		\frac{1}{T}\sum_{t=1}^T	\|\langle\bm{\Delta},\cm{Y}_{t-1}\rangle\|_\text{F}^2\geq\frac{\lambda_{\min}(\bm{\Sigma_e})}{2\mu_{\max}(\mathcal{A})}\|\bm{\Delta}\|_{\text{F}}^2
	\end{equation}
	with probability at least $1-\exp[-C\min(\kappa^{-2},\kappa^{-4})M_2^2p]$.
	
	Hence, when $T\gtrsim [1+\max(\kappa^2,\kappa^4)M_2^{-2}]p$ and $\lambda\gtrsim\kappa^2M_1 2^{1-d}\sqrt{p/T}$, with probability at least $1-\exp[-C(p-d)]-\exp[-C\min(\kappa^{-2},\kappa^{-4})M_2^2p]$, the condition $\lambda\geq4\|T^{-1}\sum_{t=1}^T\cm{Y}_{t-1}\circ\cm{E}_t\|_{\text{SSN}^*}$ and the RSC condition with the parameter $\alpha_{\text{RSC}}=\lambda_{\min}(\bm{\Sigma_e})/\mu_{\max}(\mathcal{A})$ hold.
\end{proof}

\begin{proof}[Proof of Theorem \ref{thm:rankconsistency}]
	Theorem \ref{thm:SSN} gives the Frobenius estimation error bound. For simplicity, we write $\cm{\widehat{A}}=\cm{\widehat{A}}_{\text{SSN}}$ and $\cm{\widetilde{A}}=\cm{\widehat{A}}_{\text{TSSN}}$ in this proof.
	By definition, for any tensor $\cm{A}\in\mathbb{R}^{p_1\times\cdots\times p_{2d}}$,
	\begin{equation}
		\|\cm{A}\|_{\text{F}}^2=\|\cm{A}_{(i)}\|_{\text{F}}^2=\sum_{j=1}^{p_i}\sigma_j^2(\cm{A}_{(i)}),~~i=1,2,\dots,2d.
	\end{equation}
	In other words, the Frobenius norm of the error tensor is equivalent to the $\ell_2$ norm of singular values of the one-mode matricization. By Mirsky's singular value inequality \citep{mirsky1960symmetric},
	\begin{equation}\label{eq:mirsky}
		\sum_{j=1}^{p_i}[\sigma_j(\cm{\widehat{A}}_{(i)})-\sigma_j(\cm{A}^*_{(i)})]^2\leq\sum_{j=1}^{p_i}\sigma_j^2(\cm{\widehat{A}}_{(i)}-\cm{A}^*_{(i)})=\|\cm{\widehat{A}}-\cm{A}^*\|_{\text{F}}^2,~~i=1,2,\dots,2d.
	\end{equation}
	
	Obviously, the $\ell_\infty$ error bound is smaller than the $\ell_2$ error bound, so it follows the same upper bound. By Theorem \ref{thm:SSN}, when $\lambda_{\textup{SSN}}\asymp\kappa^2M_12^{1-d}\sqrt{p/T}$, with probability approaching one,
	\begin{align}
		\max_{1\leq i\leq 2d}\max_{1\leq j\leq p_i}|\sigma_j(\cm{\widehat{A}}_{(i)})-\sigma_j(\cm{A}^*_{(i)})|\leq& \max_{1\leq i\leq 2d}\left\{\sum_{j=1}^{p_i}[\sigma_j(\cm{\widehat{A}}_{(i)})-\sigma_j(\cm{A}^*_{(i)})]^2\right\}^{1/2} \notag\\
		\leq&\|\cm{\widehat{A}}-\cm{A}^*\|_{\text{F}}\lesssim \frac{\kappa^2M_1}{\alpha_{\text{RSC}}}\sqrt{\frac{s_0p}{T}}.\label{eq:singul}
	\end{align}
	Therefore, by Assumption \ref{asmp:truncate},	as $T\rightarrow\infty$,
	\begin{equation}
		\gamma \gg \max_{1\leq i\leq 2d}\max_{1\leq j\leq p_i}|\sigma_j(\cm{\widehat{A}}_{(i)})-\sigma_j(\cm{A}^*_{(i)})|.
	\end{equation}
	Then, for any $j>r_i$, since $\sigma_j(\cm{A}^*_{(i)})=0$, we have $\gamma \gg \sigma_j(\cm{\widehat{A}}_{(i)})$. Thus,  for all $i=1,\dots, 2d$, $\sigma_j(\cm{\widehat{A}}_{(i)})$ will be truncated for all $j>r_i$.
	Meanwhile, by Assumption \ref{asmp:truncate} and \eqref{eq:singul}, we have $\sigma_{r_i}(\cm{\widehat{A}}_{(i)})>\gamma$ for $T$ sufficiently large, for all $i=1,\dots, 2d$. Therefore, the rank selection consistency of the truncated estimator $\cm{\widetilde{A}}$ can be established.
	
	Denote the event $E=\{\text{rank}(\cm{\widetilde{A}}_{(i)})=r_i,~\text{for}~i=1,\dots,2d\}$. For a generic tensor $\cm{T}\in\mathbb{R}^{p_1\times\cdots\times p_{2d}}$, denote the sub-tensor $\cm{T}_{i_k=j}$, a $p_1\times\cdots\times p_{k-1}\times 1\times p_{k+1}\times\cdots\times p_{2d}$ tensor such that
	\begin{equation}
		(\cm{T}_{i_k=j})_{i_1\dots i_{k-1}1i_{k+1}\dots i_{2d}}=\cm{T}_{i_1\dots i_{k-1}ji_{k+1}\dots i_{2d}},
	\end{equation}
	and sub-tensor $\cm{T}_{i_k>j}$, a $p_1\times\cdots\times p_{k-1}\times (p_k-j)\times p_{k+1}\times\cdots\times p_{2d}$ tensor such that
	\begin{equation}
		(\cm{T}_{i_k>j})_{i_1\dots i_{k-1}\ell i_{k+1}\dots i_{2d}}=\cm{T}_{i_1\dots i_{k-1}(\ell+j)i_{k+1}\dots i_{2d}}.
	\end{equation}
	
	Let the HOSVD of $\cm{\widehat{A}}$ be $\cm{\widehat{G}}\times_{i=1}^{2d}\bm{\widehat{U}}_i$. By definition,  $\cm{\widehat{G}}$ is a $p_1\times\cdots\times p_{2d}$ all-orthogonal and sorted tensor such that
	\begin{equation}
		\|\cm{\widehat{G}}_{i_k=1}\|_{\text{F}}\geq\|\cm{\widehat{G}}_{i_k=2}\|_{\text{F}}\geq\cdots\geq\|\cm{\widehat{G}}_{i_k=p_k}\|_{\text{F}},
	\end{equation}
	for $k=1,\dots,2d$. On $E$, the truncation procedure is equivalent to truncating all the sub-tensors $\cm{\widehat{G}}_{i_k>r_k}$ to zeros. Thus, $\|\cm{\widehat{A}}-\cm{\widetilde{A}}\|_{\text{F}}=\|\cm{\widehat{G}}-\cm{\widetilde{G}}\|_{\text{F}}^2\leq\sum_{k=1}^{2d}\|\cm{\widehat{G}}_{i_k>r_k}\|_{\text{F}}^2$. 
	
	By the definition of HOSVD, $\|\cm{\widehat{G}}_{i_k=j}\|_{\text{F}}=\sigma_j(\cm{\widehat{G}}_{(k)})=\sigma_j(\cm{\widehat{A}}_{(k)})$, and then
	\begin{equation}\begin{split}
			\|\cm{\widehat{G}}_{i_k>r_k}\|_{\text{F}}^2=\sum_{i=r_{k+1}}^{p_k}\sigma_i^2(\cm{\widehat{A}}_{(k)})&=\sum_{i=r_{k+1}}^{p_k}[\sigma_i(\cm{\widehat{A}}_{(k)})-\sigma_i(\cm{A}^*_{(k)})]^2\\
			&\leq\sum_{i=1}^{p_k}[\sigma_i(\cm{\widehat{A}}_{(k)})-\sigma_i(\cm{A}^*_{(k)})]^2\leq\|\cm{\widehat{A}}-\cm{A}^*\|_{\text{F}}^2,
	\end{split}\end{equation} 
	where the last inequality follows from \eqref{eq:mirsky}.
	
	Finally, on the event $E$, $\|\cm{\widetilde{A}}-\cm{A}^*\|_{\text{F}}\leq\|\cm{\widetilde{A}}-\cm{\widehat{A}}^*\|_{\text{F}}+\|\cm{\widehat{A}}-\cm{A}^*\|_{\text{F}}\leq(1+\sqrt{2d})\|\cm{\widehat{A}}-\cm{A}^*\|_{\text{F}}$, where $d$ is fixed. Note  that Theorem \ref{thm:SSN} implies the asymptotic rate $\|\cm{\widehat{A}}-\cm{A}^*\|_{\text{F}}=O_p(\sqrt{s_0p/T})$ and the first part of this proof shows that $\mathbb{P}(E)\to1$, as $T\to\infty$. The proof is complete.	
\end{proof}

\subsection{Proofs of Lemmas \ref{lemma:ApproxCone}--\ref{lemma:RSC} \label{subsec:lemma}}

\begin{proof}[Proof of Lemma \ref{lemma:ApproxCone}]
	
	In this part, we focus on $\cm{\widehat{A}}_{\text{SSN}}$ and simplify it to $\cm{\widehat{A}}$. The tuning parameter $\lambda_{\text{SSN}}$ is simplified to $\lambda$. The proof can be readily extended to $\cm{\widehat{A}}_{\text{SN}}$ and $\cm{\widehat{A}}_{\text{MN}}$.
	
	Note that the quadratic loss function can be rewritten as $\mathcal{L}_T(\cm{A})=T^{-1}\sum_{t=1}^{T}\|\cm{Y}_t-\langle\cm{A},\cm{Y}_{t-1}\rangle\|_\text{F}^2=T^{-1}\sum_{t=1}^T\|\bm{y}_t-\cm{A}_{[S_2]}\bm{y}_{t-1}\|_2^2$, where $\bm{y}_t=\text{vec}(\cm{Y}_t)$.	
	By the optimality of the SSN estimator,
	\begin{equation}\begin{split}
			&\frac{1}{T}\sum_{t=1}^T\|\bm{y}_t-\cm{\widehat{A}}_{[S_2]}\bm{y}_{t-1}\|_2^2+\lambda\|\cm{\widehat{A}}\|_{\text{SSN}}\leq\frac{1}{T}\sum_{t=1}^T\|\bm{y}_t-\cm{A}^*_{[S_2]}\bm{y}_{t-1}\|_2^2+\lambda\|\cm{A}^*\|_{\text{SSN}}\\
			\Rightarrow~&\frac{1}{T}\sum_{t=1}^T\|\bm{\Delta}_{[S_2]}\bm{y}_{t-1}\|_2^2\leq\frac{2}{T}\sum_{t=1}^T\langle\bm{e}_t,\bm{\Delta}_{[S_2]}\bm{y}_{t-1}\rangle+\lambda(\|\cm{A}^*\|_{\text{SSN}}-\|\cm{\widehat{A}}\|_{\text{SSN}})\\
			\Rightarrow~&\frac{1}{T}\sum_{t=1}^T\|\bm{\Delta}_{[S_2]}\bm{y}_{t-1}\|_2^2\leq2\left\langle T^{-1}\sum_{t=1}^T\cm{Y}_{t-1}\circ\cm{E}_t,\bm{\Delta}\right\rangle+\lambda(\|\cm{A}^*\|_{\text{SSN}}-\|\cm{\widehat{A}}\|_{\text{SSN}})\\
			\Rightarrow~&\frac{1}{T}\sum_{t=1}^T\|\bm{\Delta}_{[S_2]}\bm{y}_{t-1}\|_2^2\leq2\|\bm{\Delta}\|_{\text{SSN}}\left|\left|T^{-1}\sum_{t=1}^T\cm{Y}_{t-1}\circ\cm{E}_t\right|\right|_{\text{SSN}^*}+\lambda(\|\cm{A}^*\|_{\text{SSN}}-\|\cm{\widehat{A}}\|_{\text{SSN}}),
	\end{split}\end{equation}
	where $\|\cdot\|_{\text{SSN}^*}$ refers to the dual norm of the SSN norm.
	
	By triangle inequality and decomposability, we have
	\begin{equation}\begin{split}
			&\|\cm{\widehat{A}}\|_{\text{SSN}}-\|\cm{A}^*\|_{\text{SSN}}=\|\cm{A}+\bm{\Delta}\|_{\text{SSN}}-\|\cm{A}^*\|_{\text{SSN}}=\sum_{k=1}^{2^{d-1}}\|\cm{A}^*_{[I_k]}+\bm{\Delta}_{[I_k]}\|_\textup{nuc}-\sum_{k=1}^{2^{d-1}}\|\cm{A}^*_{[I_k]}\|_\textup{nuc}\\
			=&\sum_{k=1}^{2^{d-1}}\|\cm{A}_{\mathcal{M}}^{*(k)}+\cm{A}_{\mathcal{M}^\perp}^{*(k)}+\bm{\Delta}_{\overline{\mathcal{M}}}^{(k)}+\bm{\Delta}_{\overline{\mathcal{M}}^\perp}^{(k)}\|_\textup{nuc}-\sum_{k=1}^{2^{d-1}}\|\cm{A}^*_{[I_k]}\|_\textup{nuc}\\
			\geq&\sum_{k=1}^{2^{d-1}}\left[\|\cm{A}_{\mathcal{M}}^{*(k)}+\bm{\Delta}_{\overline{\mathcal{M}}^\perp}^{(k)}\|_\textup{nuc}-\|\cm{A}_{\mathcal{M}^\perp}^{*(k)}+\bm{\Delta}_{\overline{\mathcal{M}}}^{(k)}\|_\textup{nuc}-\|\cm{A}^{*(k)}_{\mathcal{M}^\perp}\|_\textup{nuc}-\|\cm{A}^{*(k)}_{\mathcal{M}}\|_\textup{nuc}\right]\\
			\geq&\sum_{k=1}^{2^{d-1}}\left[\|\bm{\Delta}_{\overline{\mathcal{M}}^\perp}^{(k)}\|_\textup{nuc}-2\|\cm{A}_{\mathcal{M}^\perp}^{*(k)}\|_\textup{nuc}-\|\bm{\Delta}_{\overline{\mathcal{M}}}^{(k)}\|_\textup{nuc}\right].
	\end{split}\end{equation}
	
	If $\lambda\geq 4\|T^{-1}\sum_{t=1}^T\cm{Y}_{t-1}\circ\cm{E}_t\|_{\text{SSN}^*}$, we have
	\begin{equation}\begin{split}
			0&\leq\frac{1}{T}\sum_{t=1}^T\|\bm{\Delta}_{[S_2]}\bm{y}_{t-1}\|_2^2
			\leq\frac{\lambda}{2}\|\bm{\Delta}\|_{\text{SSN}}-\lambda(\|\cm{\widehat{A}}\|_{\text{SSN}}-\|\cm{A}^*\|_{\text{SSN}})\\
			&\leq\frac{\lambda}{2}\sum_{k=1}^{2^{d-1}}\left[\|\bm{\Delta}^{(k)}_{\overline{\mathcal{M}}}\|_\textup{nuc}+\|\bm{\Delta}^{(k)}_{\overline{\mathcal{M}}^\perp}\|_\textup{nuc}-2\|\bm{\Delta}^{(k)}_{\overline{\mathcal{M}}^\perp}\|_\textup{nuc}+4\|\cm{A}^{*(k)}_{\mathcal{M}^\perp}\|_\textup{nuc}+2\|\bm{\Delta}_{\overline{\mathcal{M}}}^{(k)}\|_\textup{nuc}\right]\\
			&=\frac{\lambda}{2}\sum_{k=1}^{2^{d-1}}\left[3\|\bm{\Delta}_{\overline{\mathcal{M}}}^{(k)}\|_\textup{nuc}+4\|\cm{A}^{*(k)}_{\mathcal{M}^\perp}\|_\textup{nuc}-\|\bm{\Delta}^{(k)}_{\overline{\mathcal{M}}^\perp}\|_\textup{nuc}\right].
	\end{split}\end{equation}
	Hence, the error $\bm{\Delta}$ lies in the restricted error set $\mathbb{C}_{\text{SSN}}(\overline{\mathcal{M}})$.
\end{proof}

\begin{proof}[Proof of Lemma \ref{lemma:errorbound}]
	
	Similar to Lemma \ref{lemma:ApproxCone}, we focus on the SSN estimator, and the results for SN and MN estimators can be extended in a similar way.
	
	Note that $T^{-1}\sum_{t=1}^T\|\langle\bm{\Delta},\cm{Y}_{t-1}\rangle\|_\text{F}^2=T^{-1}\sum_{t=1}^T\|\bm{\Delta}_{[S_2]}\bm{y}_{t-1}\|_2^2$.
	Following the proof of Lemma \ref{lemma:ApproxCone}, $\bm{\Delta}\in\mathbb{C}_{\text{SSN}}(\overline{\mathcal{M}})$ and
	\begin{equation}\begin{split}
			\frac{1}{T}\sum_{t=1}^T\|\langle\bm{\Delta},\cm{Y}_{t-1}\rangle\|_\text{F}^2&\leq\frac{\lambda}{2}\|\bm{\Delta}\|_{\text{SSN}}+\lambda(\|\cm{A}^*\|_{\text{SSN}}-\|\cm{\widehat{A}}\|_{\text{SSN}})\leq\frac{3\lambda}{2}\|\bm{\Delta}\|_{\text{SSN}}\\
			&=\frac{3\lambda}{2}\sum_{k=1}^{2^{d-1}}\|\bm{\Delta}_{[I_k]}\|_\textup{nuc}\leq\frac{3\lambda}{2}\sum_{k=1}^{2^{d-1}}\left(\|\bm{\Delta}_{\overline{\mathcal{M}}}^{(k)}\|_\textup{nuc}+\|\bm{\Delta}_{\overline{\mathcal{M}}^\perp}^{(k)}\|_\textup{nuc}\right)\\
			&\leq6\lambda\sum_{k=1}^{2^{d-1}}\|\bm{\Delta}_{\overline{\mathcal{M}}}^{(k)}\|_\textup{nuc}+6\lambda\sum_{k=1}^{2^{d-1}}\|\cm{A}^{*(k)}_{\mathcal{M}^\perp}\|_\textup{nuc}\\
			&\leq6\lambda\sum_{k=1}^{2^{d-1}}\sqrt{2s_k}\|\bm{\Delta}_{\overline{\mathcal{M}}}^{(k)}\|_{\text{F}}+6\lambda\sum_{k=1}^{2^{d-1}}\|\cm{A}^{*(k)}_{\mathcal{M}^\perp}\|_\textup{nuc}\\
			&\lesssim\lambda\sum_{k=1}^{2^{d-1}}\sqrt{2s_k}\|\bm{\Delta}\|_{\text{F}}+6\lambda\sum_{k=1}^{2^{d-1}}\|\cm{A}^{*(k)}_{\mathcal{M}^\perp}\|_\textup{nuc}
	\end{split}\end{equation}
	where the last inequality stems from the fact that $\bm{\Delta}_{\overline{\mathcal{M}}}^{(k)}$ has a matrix rank at most $2s_k$, similar to Lemma 1 in \citet{negahban2011estimation}.
	
	As the RSC condition holds with the parameter $\alpha_{\text{RSC}}$ and restricted error set $\mathbb{C}_{\text{SSN}}(\overline{\mathcal{M}})$,
	\begin{equation}
		\alpha_{\text{RSC}}\|\bm{\Delta}\|_{\text{F}}^2\leq	\frac{1}{T}\sum_{t=1}^T\|\langle\bm{\Delta},\cm{Y}_{t-1}\rangle\|_\text{F}^2\lesssim \lambda\sum_{k=1}^{2^{d-1}}\sqrt{s_k}\|\bm{\Delta}\|_\text{F}+\lambda\sum_{k=1}^{2^{d-1}}\|\cm{A}^{*(k)}_{\mathcal{M}^\perp}\|_\text{nuc}.
	\end{equation}
	Thus, by the Cauchy--Schwarz inequality,
	\begin{equation}
		\begin{split}
			\|\bm{\Delta}\|_\text{F}^2\lesssim\frac{\lambda^2(\sum_{k=1}^{2^{d-1}}\sqrt{s_k})^2}{\alpha_{\text{RSC}}^2}+\frac{\lambda\sum_{k=1}^{2^{d-1}}\|\cm{A}^{*(k)}_{\mathcal{M}^\perp}\|_\textup{nuc}}{\alpha_{\text{RSC}}}
			\lesssim\frac{\lambda^22^{d-1}\sum_{k=1}^{2^{d-1}}s_k}{\alpha_{\text{RSC}}^2}+\frac{\lambda\sum_{k=1}^{2^{d-1}}\|\cm{A}^{*(k)}_{\mathcal{M}^\perp}\|_\textup{nuc}}{\alpha_{\text{RSC}}}.
		\end{split}
	\end{equation}	
	
	Consider any threhold $\tau_k\geq0$ and define the thresholded subspace $\mathcal{M}^{(k)}$ corresponding to the column and row spaces spanned by the first $r^{(k)}$ singular vectors of $\cm{A}_{[I_k]}$ where $\sigma_1(\cm{A}^*_{[I_k]})\geq\cdots\geq\sigma_{r^{(k)}}(\cm{A}^*_{[I_k]})> \tau_k\geq\sigma_{r^{(k)}+1}(\cm{A}^*_{[I_k]})$.
	By the definition of $\mathbb{B}_q(s_q^{(k)};p,p)$, we have $s_q^{(k)}\geq r^{(k)}\cdot\tau_k^q$ and thus $r^{(k)}\leq s_q^{(k)}\cdot\tau_k^{-q}$.
	
	Then, the approximation error can be bounded by
	\begin{equation}
		\|\cm{A}^{*(k)}_{\mathcal{M}^\perp}\|_\textup{nuc}=\sum_{r=r^{(k)}+1}^{p}\sigma_r(\cm{A}^*_{[I_k]})=\sum_{r=r^{(k)}+1}^{p}\sigma_r^q(\cm{A}^*_{[I_k]})\cdot\sigma_r^{1-q}(\cm{A}^*_{[I_k]})\leq s_q^{(k)}\cdot\tau_k^{1-q}.
	\end{equation}
	The estimation error can be bounded by
	\begin{equation}
		\|\bm{\Delta}\|_\text{F}^2\lesssim\frac{\lambda^22^{d-1}\sum_{k=1}^{2^{d-1}}s_q^{(k)}\cdot\tau_k^{-q}}{\alpha_{\text{RSC}}^2}+\frac{\lambda\sum_{k=1}^{2^{d-1}} s_q^{(k)}\cdot\tau_k^{1-q}}{\alpha_{\text{RSC}}}.
	\end{equation}
	Setting each $\tau_k\asymp\alpha_\text{RSC}^{-1}(q/(1-q))2^{d-1}\lambda$, the upper bound can be minimized to
	\begin{equation}
		\|\bm{\Delta}\|_\text{F}^2\lesssim 2^{1-d}\sum_{k=1}^{2^{d-1}}s_q^{(k)}\left(\frac{\lambda\cdot 2^{d-1}}{\alpha_{\text{RSC}}}\right)^{2-q}.
	\end{equation}
	The proof is complete.	
\end{proof}

\begin{proof}[Proof of Lemma \ref{lemma:dualnorm}]
	First, we derive an upper bound of the dual norm of the SSN norm. By definition, for any tensor $\cm{A}$ and collection of index sets $\mathbb{I}=\{I_1,\dots,I_{2^{d-1}}\}$, the SSN norm is
	\begin{equation}
		\|\cm{A}\|_{\text{SSN}}=\sum_{k=1}^{2^{d-1}}\|\cm{A}_{[I_k]}\|_\textup{nuc},
	\end{equation}
	and its dual norm is $\|\cm{A}\|_{\text{SSN}^*}:=\sup\langle\cm{W},\cm{A}\rangle$ such that $\|\cm{W}\|_{\text{SSN}}\leq1$. By a method similar to that in  \cite{tomioka2011statistical}, it can be shown that 
	\begin{equation}
		\|\cm{A}\|_{\text{SSN}^*}=\inf_{\sum_{k=1}^{2^{d-1}}\scalebox{0.7}{\cm{X}}_k=\scalebox{0.7}{\cm{A}}}\max_{k=1,\dots,2^{d-1}}\|(\cm{X}_k)_{[I_k]}\|_{\text{op}}.
	\end{equation}
	Then, we can take $\cm{X}_k=(\sum_{k=1}^{2^{d-1}}1/c_k)^{-1}(\cm{A}/c_k)$, where $c_k=\|\cm{A}_{[I_k]}\|_{\text{op}}$, and apply Jensen's inequality so that we have
	\begin{equation}
		\|\cm{A}\|_{\text{SSN}^*}\leq2^{-2(d-1)}\sum_{k=1}^{2^{d-1}}\|\cm{A}_{[I_k]}\|_{\text{op}}.
	\end{equation}
	Hence, we have
	\begin{equation}
		\left\|\frac{1}{T}\sum_{t=1}^T\cm{Y}_{t-1}\circ\cm{E}_t\right\|_{\text{SSN}^*}\leq\frac{1}{2^{2(d-1)}}\sum_{k=1}^{2^{d-1}}\left\|\frac{1}{T}\sum_{t=1}^T(\cm{Y}_{t-1}\circ\cm{E}_t)_{[I_k]}\right\|_{\text{op}}.
	\end{equation}
	In other words, the dual norm of the SSN norm can be upper bounded by the sum of the scaled matrix operator norms of different matricizations of the tensor $T^{-1}\sum_{t=1}^T\cm{Y}_{t-1}\circ\cm{E}_t$.
	
	All of the square matricizations based on $I_k$ lead to a square $p$-by-$p$ matrix. Therefore, by the deviation bound in Lemma \ref{lemma:deviation}, we can take a union bound such that
	\begin{equation}
		\left\|\frac{1}{T}\sum_{t=1}^T\cm{Y}_{t-1}\circ\cm{E}_t\right\|_{\text{SSN}^*}\leq \frac{C\kappa^2 M_1}{2^{d-1}}\sqrt{\frac{p}{T}}
	\end{equation}
	with probability at least $1-\exp[-C(p-d)]$.
	
	
	Next, for the SN estimator, we can obtain a similar upper bound of the dual norm of the SN norm. The SN norm is defined as
	\begin{equation}
		\|\cm{A}\|_{\text{SN}}=\sum_{i=1}^{2d}\|\cm{A}_{(i)}\|_\textup{nuc},
	\end{equation}
	and its dual norm has the equivalent form
	\begin{equation}
		\|\cm{A}\|_{\text{SN}^*}=\inf_{\sum_{i=1}^{2d}\scalebox{0.7}{\cm{Y}}_i=\scalebox{0.7}{\cm{A}}}\max_{i=1,\dots,2d}\|(\cm{Y}_i)_{(i)}\|_{\text{op}}.
	\end{equation}
	Then, we can obtain an upper bound,
	\begin{equation}
		\|\cm{A}\|_{\text{SN}^*}\leq\frac{1}{(2d)^2}\sum_{i=1}^{2d}\|\cm{A}_{(i)}\|_{\text{op}}.
	\end{equation}
	Then, for each one-mode matricization, we have the deviation bound. Then, we can take a union bound such that
	\begin{equation}
		\left\|\frac{1}{T}\sum_{t=1}^T\cm{Y}_{t-1}\circ\cm{E}_t\right\|_{\text{SN}^*}\leq \frac{C\kappa^2 M_1}{(2d)^2}\sum_{i=1}^{2d}\sqrt{\frac{p_{-i}p}{T}},
	\end{equation}
	with probability at least $1-2d\exp[-Cp]$.
	
	Finally, the MN estimator uses a special case of square matricization, and the upper bound for the MN estimator can be obtained by Lemma \ref{lemma:deviation}.
\end{proof}

\begin{proof}[Proof of Lemma \ref{lemma:RSC}]
	For any $\bm{M}\in\mathbb{R}^{m\times p}$, denote $R_T(\bm{M})=\sum_{t=0}^{T-1}\|\bm{M}\bm{y}_t\|_2^2$. Note that $R_T(\bm{\Delta}_{[S_2]})\geq\mathbb{E}R_T(\bm{\Delta}_{[S_2]})-\sup_{\bm{\Delta}}|R_T(\bm{\Delta}_{[S_2]})-\mathbb{E}R_T(\bm{\Delta}_{[S_2]})|$. 	Following the proof of Lemma \ref{lemma:deviation}, $\mathbb{E}R_T(\bm{\Delta}_{[S_2]})=\|(\bm{I}_T\otimes\bm{\Delta}_{[S_2]})\bm{PD}\|_{\text{F}}^2\geq T\|\bm{\Delta}\|_{\text{F}}^2\cdot\lambda_{\min}(\bm{\Sigma_e})\lambda_{\min}(\bm{PP}^\top)$.
	
	Similar to Lemma \ref{lemma:deviation2}, for any $\bm{v}\in\mathbb{S}^{p-1}$ and any $t>0$,
	\begin{equation}
		\begin{split}
			&\mathbb{P}[|R_T(\bm{v}^\top)-\mathbb{E}R_T(\bm{v}^\top)|\geq t]\\
			\leq& 2\exp\left(-\min\left(\frac{t^2}{\kappa^4T\lambda_{\max}^2(\bm{\Sigma_e})\lambda_{\max}^2(\bm{P}\bm{P}^\top)},\frac{t}{\kappa^2\lambda_{\max}(\bm{\Sigma_e})\lambda_{\max}(\bm{P}\bm{P}^\top)}\right)\right).
		\end{split}
	\end{equation}
	Considering an $\epsilon$-covering net of $\mathbb{S}^{p-1}$, by Lemma \ref{lemma:covering}, we can easily construct the union bound for $T\gtrsim p$,
	\begin{equation}\begin{split}
			&\mathbb{P}\left[\sup_{v\in\mathbb{S}^{p-1}}|R_T(\bm{v}^\top)-\mathbb{E}R_T(\bm{v}^\top)|\geq t\right]\\
			\leq& C\exp\left(p-\min\left(\frac{t^2}{\kappa^4T\lambda_{\max}^2(\bm{\Sigma_e})\lambda_{\max}^2(\bm{P}\bm{P}^\top)},\frac{t}{\kappa^2\lambda_{\max}(\bm{\Sigma_e})\lambda_{\max}(\bm{P}\bm{P}^\top)}\right)\right),
	\end{split}\end{equation}
	
	Letting $t=\lambda_{\min}(\bm{\Sigma_e})\lambda_{\min}(\bm{PP}^\top)/2$, for $T\gtrsim M_2^{-2}\max(\kappa^4,\kappa^2)p$, we have
	\begin{equation}
		\mathbb{P}[|R_T(\bm{v}^\top)-\mathbb{E}R_T(\bm{v}^\top)|\geq \lambda_{\min}(\bm{\Sigma_e})\lambda_{\min}(\bm{PP}^\top)/2]\leq 2\exp(-CM_2^2\min(\kappa^{-4},\kappa^{-2})T),
	\end{equation}
	where $M_2=[\lambda_{\min}(\bm{\Sigma_e})\lambda_{\min}(\bm{PP}^\top)]/[\lambda_{\max}(\bm{\Sigma_e})\lambda_{\max}(\bm{PP}^\top)]$.
	
	Therefore, with probability at least $1-2\exp(-CM_2^2\min(\kappa^{-4},\kappa^{-2})T)$,
	\begin{equation}
		R_T(\bm{\Delta}_{[S_2]})\geq\frac{1}{2}\lambda_{\min}(\bm{\Sigma_e})\lambda_{\min}(\bm{PP}^\top)\|\bm{\Delta}\|_{\text{F}}^2.
	\end{equation}
	
	Finally, since $\bm{P}$ is related to the VMA($\infty$) process, by the spectral measure of ARMA process discussed in \citet{basu2015regularized}, we may replace $\lambda_{\max}(\bm{P}\bm{P}^\top)$ and $\lambda_{\min}(\bm{P}\bm{P}^\top)$ with $1/\mu_{\min}(\mathcal{A})$ and $1/\mu_{\max}(\mathcal{A})$, respectively.
\end{proof}

\subsection{Three Auxiliary Lemmas \label{subsec:auxlemma}}
Three auxiliary lemmas used in the proofs of Lemmas  \ref{lemma:dualnorm} and \ref{lemma:RSC} are presented below.

\begin{lemma}[Deviation bound on different matricizations]
	\label{lemma:deviation}
	For any index set $I\subset\{1,2,\dots,2d\}$, denote $q=\prod_{i=1,i\in I}^{2d}p_i$ and $q'=\prod_{i=1,i\notin I}^{2d}p_i$. If $T\gtrsim(q+q')$, with probability at least $1-\exp[-C(q+q')]$,
	\begin{equation}
		\left\|\frac{1}{T}\sum_{t=1}^T(\cm{Y}_{t-1}\circ\cm{E}_t)_{[I]}\right\|_{\textup{op}}<C\kappa^2M_1\sqrt{(q+q')/T}.
	\end{equation}
	where $M_1=\lambda_{\max}(\bm{\Sigma_e})/\mu^{1/2}_{\min}(\mathcal{A})$.
\end{lemma}

\begin{proof}
	For any index set $I\subset\{1,2,\dots,2d\}$ and $2d$th-mode tensor $\cm{T}$, denote the inverse operation of the multi-mode matricization $\bm{T}=\cm{T}_{[I]}$ by $\bm{T}^{[I]}=\cm{T}$. Denote $\mathcal{W}(r;q,q')=\{\bm{W}\in\mathbb{R}^{q\times q'}:\textup{rank}(\bm{W})=r,~\|\bm{W}\|_\textup{F}=1\}$.
	
	By definition, $\|T^{-1}\sum_{t=1}^T(\cm{Y}_{t-1}\circ\cm{E}_t)_{[I]}\|_{\textup{op}}=\sup_{\bm{W}\in\mathcal{W}(1;q,q')}\langle T^{-1}\sum_{t=1}^T(\cm{Y}_{t-1}\circ\cm{E}_t)_{[I]},\bm{W}\rangle=\sup_{\bm{W}\in\mathcal{W}(1;q,q')}\langle T^{-1}\sum_{t=1}^T\textup{vec}(\cm{E}_{t})\textup{vec}(\cm{Y}_{t-1})^\top,(\bm{W}^{[I]})_{[S_1]}^\top\rangle$.
	
	For an arbitrary matrix $\bm{W}\in\mathbb{R}^{q\times q'}$ such that $\|\bm{W}\|_\textup{F}=1$, denote $\bm{M}=(\bm{W}^{[I]})_{[S_1]}^\top$. Then, one can easily check that $\langle(\cm{Y}_{t-1}\circ\cm{E}_t)_{[I]},\bm{W}\rangle=\langle \bm{e}_t,\bm{M}\bm{y}_{t-1}\rangle$.
	
	For a fixed $\bm{M}$, denote $S_t(\bm{M})=\sum_{s=1}^t\langle\bm{e}_s,\bm{M}\bm{y}_{s-1}\rangle$ and $R_t(\bm{M})=\sum_{s=0}^{t-1}\|\bm{M}\bm{y}_s\|_2^2$, for $1\leq t\leq T$. By the standard Chernoff argument, for any $\alpha>0$, $\beta>0$ and $c>0$,
	\begin{equation}
		\begin{split}
			&\mathbb{P}[\{S_T(\bm{M})\geq\alpha\}\cap\{R_T(\bm{M})\leq\beta\}]\\
			=&\inf_{m>0}\mathbb{P}[\{\exp(m S_T(\bm{M}))\geq\exp(m\alpha)\}\cap\{R_T(\bm{M})\leq\beta\}]\\
			=&\inf_{m>0}\mathbb{P}[\exp(m S_T(\bm{M}))\mathbb{I}(R_T(\bm{M})\leq\beta)\geq\exp(m\alpha)]\\
			\leq&\inf_{m>0}\exp(-m\alpha)\mathbb{E}[\exp(m S_T(\bm{M}))\mathbb{I}(R_T(\bm{M})\leq\beta)]\\
			=&\inf_{m>0}\exp(-m\alpha+cm^2\beta)\mathbb{E}[\exp(m S_T(\bm{M})-cm^2\beta)\mathbb{I}(R_T(\bm{M})\leq\beta)]\\
			\leq&\inf_{m>0}\exp(-m\alpha+cm^2\beta)\mathbb{E}[\exp(m S_T(\bm{M})-cm^2R_T(\bm{M}))].
		\end{split}
	\end{equation}
	By the tower rule, we have
	\begin{equation}
		\begin{split}
			&\mathbb{E}[\exp(mS_T(\bm{M})-cm^2R_T(\bm{M}))]\\
			=&\mathbb{E}[\mathbb{E}[\exp(m S_T(\bm{M})-cm^2R_T(\bm{M}))]|\mathcal{F}_{T-1}]\\
			=&\mathbb{E}[\exp(m S_{T-1}(\bm{M})-cm^2R_{T-1}(\bm{M}))\mathbb{E}[\exp(m\langle\bm{e}_T,\bm{M}\bm{y}_{T-1}\rangle-cm^2\|\bm{M}\bm{y}_T\|_2^2)|\mathcal{F}_{T-1}]].
		\end{split}
	\end{equation}
	
	Since $\langle\bm{e}_T,\bm{M}\bm{y}_{T-1}\rangle=\langle\bm{\xi}_T,\bm{\Sigma}_{\bm{e}}^{1/2}\bm{M}\bm{y}_{T-1}\rangle$, one can easily check that $\langle\bm{e}_T,\bm{M}\bm{y}_{T-1}\rangle$ is a $\kappa^2\lambda_{\max}(\bm{\Sigma}_{\bm{e}})\|\bm{M}\bm{y}_{T-1}\|_2^2$-sub-Gaussian random variable. In other words, $\mathbb{E}[\exp(m\langle\bm{e}_T,\bm{M}\bm{y}_{T-1}\rangle)]\leq\exp(m^2\kappa^2\lambda_{\max}(\bm{\Sigma}_{\bm{e}})\|\bm{M}\bm{y}_{T-1}\|_2^2/2)$.
	Thus, letting $c=\kappa\lambda_{\max}(\bm{\Sigma}_{\bm{e}})/2$, we have
	\begin{equation}
		\begin{split}
			&\mathbb{E}[\exp(mS_T(\bm{M})-m^2\kappa^2\lambda_{\max}(\bm{\Sigma}_{\bm{e}})R_T(\bm{M})/2)]\\
			\leq&\mathbb{E}[\exp(mS_{T-1}(\bm{M})-m^2\kappa^2\lambda_{\max}(\bm{\Sigma}_{\bm{e}})R_{T-1}(\bm{M})/2)]\\
			\leq&\cdots\leq\mathbb{E}[\exp(mS_1(\bm{M})-m^2\kappa^2\lambda_{\max}(\bm{\Sigma}_{\bm{e}})R_1(\bm{M})/2)]\leq1.
		\end{split}
	\end{equation}
	Hence, we have that, for any $\alpha>0$ and $\beta>0$,
	\begin{equation}
		\label{eq:martingale_deviation}
		\begin{split}
			&\mathbb{P}[\{S_T(\bm{M})\geq\alpha\}\cap\{R_T(\bm{M})\leq\beta\}]\\
			\leq&\inf_{m>0}\exp(-m\alpha+m^2\kappa^2\lambda_{\max}(\bm{\Sigma}_{\bm{e}})\beta/2)\\
			=&\exp\left(-\frac{\alpha^2}{2\kappa^2\lambda_{\max}(\bm{\Sigma}_{\bm{e}})\beta}\right).
		\end{split}
	\end{equation}
	
	By Lemma \ref{lemma:deviation2}, we have that for any $t>0$,
	\begin{equation}
		\begin{split}
			&\mathbb{P}[|R_T(\bm{M})-\mathbb{E}R_T(\bm{M})|\geq t]\\
			\leq&2\exp\left(-\min\left(\frac{t^2}{\kappa^4T\lambda_{\max}^2(\bm{\Sigma_e})\lambda_{\max}^2(\bm{PP}^\top)},\frac{t}{\kappa^2\lambda_{\max}^2(\bm{\Sigma_e})\lambda_{\max}^2(\bm{PP}^\top)}\right)\right).
		\end{split}
	\end{equation}
	In addition, $\mathbb{E}R_T(\bm{M})=\text{tr}(\bm{\Sigma_M})=\|(\bm{I}_T\otimes\bm{M})\bm{P}\bm{D}\|_{\text{F}}^2\leq T\cdot\lambda_{\max}(\bm{\Sigma_e})\lambda_{\max}(\bm{P}\bm{P}^\top)$. Letting $t=C\kappa^2T\lambda_{\max}(\bm{\Sigma_e})\lambda_{\max}(\bm{P}\bm{P}^\top)$, we have
	\begin{equation}
		\mathbb{P}[R_T(\bm{M})\geq C\kappa^2T\lambda_{\max}(\bm{\Sigma_e})\lambda_{\max}(\bm{P}\bm{P}^\top)]\leq2\exp(-CT).
	\end{equation}
	
	Next, consider a $\epsilon$-net $\overline{\mathcal{W}}(1;q,q')$ for $\mathcal{W}(1;q,q')$. For any matrix $\bm{W}\in\mathcal{W}(1;q,q')$, there exist a matrix $\overline{\bm{W}}\in\overline{\mathcal{W}}(1;q,q')$ such that $\|\bm{W}-\overline{\bm{W}}\|_\textup{F}\leq\epsilon$. Since the rank of $\overline{\bm{\Delta}}=\bm{W}-\overline{\bm{W}}$ is at most 2, we can split the SVD of $\overline{\bm{\Delta}}$ into 2 parts, such that $\overline{\bm{\Delta}}=\bm{\Delta}_1+\bm{\Delta}_2$, where $\textup{rank}(\bm{\Delta}_1)=\textup{rank}(\bm{\Delta}_2)=1$ and $\langle\bm{\Delta}_1,\bm{\Delta}_2\rangle=0$. Then, for any matrix $\bm{N}\in\mathbb{R}^{q\times q'}$, we have
	\begin{equation}
		\langle\bm{N},\bm{W}\rangle=\langle\bm{N},\overline{\bm{W}}\rangle+\langle\bm{N},\overline{\bm{\Delta}}\rangle=\langle\bm{N},\overline{\bm{W}}\rangle+\sum_{i=1}^2\langle\bm{N},\bm{\Delta}_i	/\|\bm{\Delta}_i\|_\textup{F}\rangle\|\bm{\Delta}_i\|_\textup{F},
	\end{equation}
	where $\bm{\Delta}_i/\|\bm{\Delta}_i\|_{\textup{F}}\in\mathcal{W}(1;q,q')$. Since $\|\overline{\bm{\Delta}}\|_\textup{F}^2=\|\bm{\Delta}_1\|_\textup{F}^2+\|\bm{\Delta}_2\|_\textup{F}^2$, by Cauchy inequality, $\|\bm{\Delta}_1\|_\textup{F}+\|\bm{\Delta}_2\|_\textup{F}\leq\sqrt{2}\|\overline{\bm{\Delta}}\|_\text{F}=\sqrt{2}\epsilon$. Hence, we have
	\begin{equation}
		\gamma:=\sup_{\bm{W}\in\mathcal{W}(1;q,q')}\langle\bm{N},\bm{W}\rangle\leq\max_{\overline{\bm{W}}\in\overline{\mathcal{W}}(1;q,q')}\langle\bm{N},\overline{\bm{W}}\rangle+\sqrt{2}\gamma\epsilon.
	\end{equation}
	In other words,
	\begin{equation}
		\sup_{\bm{W}\in\mathcal{W}(1;q,q')}\langle\bm{N},\bm{W}\rangle\leq(1-\sqrt{2}\epsilon)^{-1}\max_{\overline{\bm{W}}\in\overline{\mathcal{W}}(1;q,q')}\langle\bm{N},\overline{\bm{W}}\rangle.
	\end{equation}
	Therefore, we have that, for any $x>0$,
	\begin{equation}
		\label{eq:sup_covering}
		\begin{split}
			&\mathbb{P}\left[\sup_{\bm{W}\in\mathcal{W}(1;q,q')}\left\langle\frac{1}{T}\sum_{t=1}^T(\cm{Y}_{t-1}\circ\cm{E}_t)_{[I]},\bm{W}\right\rangle\geq x\right]\\
			\leq&\mathbb{P}\left[\max_{\bm{W}\in\overline{\mathcal{W}}(1;q,q')}\left\langle\frac{1}{T}\sum_{t=1}^T(\cm{Y}_{t-1}\circ\cm{E}_t)_{[I]},\bm{W}\right\rangle\geq (1-\sqrt{2}\epsilon)x\right]\\
			\leq&|\overline{\mathcal{W}}(1;q,q')|\cdot\mathbb{P}\left[\left\langle\frac{1}{T}\sum_{t=1}^T(\cm{Y}_{t-1}\circ\cm{E}_t)_{[I]},\bm{W}\right\rangle\geq (1-\sqrt{2}\epsilon)x\right].
		\end{split}
	\end{equation}
	Note that by \eqref{eq:martingale_deviation}, for any $x>0$,
	\begin{equation}
		\begin{split}
			&\mathbb{P}\left[\left\langle\frac{1}{T}\sum_{t=1}^T(\cm{Y}_{t-1}\circ\cm{E}_t)_{[I]},\bm{W}\right\rangle\geq (1-\sqrt{2}\epsilon)x\right]\\
			\leq&\mathbb{P}[\{S_T(\bm{M})\geq T(1-\sqrt{2}\epsilon)x\}\cap\{R_T(\bm{M})\leq C\kappa^2 T\lambda_{\max}(\bm{\Sigma_e})\lambda_{\max}(\bm{P}\bm{P}^\top)\}]\\
			+&\mathbb{P}[R_T(\bm{M})>C\kappa^2T\lambda_{\max}(\bm{\Sigma_e})\lambda_{\max}(\bm{P}\bm{P}^\top)]\\
			\leq&\exp\left[-\frac{CTx^2}{\kappa^4\lambda_{\max}^2(\bm{\Sigma}_{\bm{e}})\lambda_{\max}(\bm{P}\bm{P}^\top)}\right]+2\exp(-CT).
		\end{split}
	\end{equation}
	By Lemma 3.1 in \citet{candes2011tight}, for a $\epsilon$-net for $\mathcal{W}(1;q,q')$, the covering number $|\overline{\mathcal{W}}(1;q,q')|\leq(9/\epsilon)^{q+q'}$. Combining \eqref{eq:sup_covering}, we have that, when $T\gtrsim q+q'$, for any $x>0$,
	\begin{equation}
		\begin{split}
			&\mathbb{P}\left[\sup_{\bm{W}\in\mathcal{W}(1;q,q')}\left\langle\frac{1}{T}\sum_{t=1}^T(\cm{Y}_{t-1}\circ\cm{E}_t)_{[I]},\bm{W}\right\rangle\geq x\right]\\
			\leq&\exp\left[(q+q')\log(9/\epsilon)-\frac{CTx^2}{\kappa^4\lambda_{\max}^2(\bm{\Sigma}_{\bm{e}})\lambda_{\max}(\bm{P}\bm{P}^\top)}\right]+2\exp[(q+q')\log(9/\epsilon)-CT].
		\end{split}
	\end{equation}
	Taking $\epsilon=0.1$ and $x=C\kappa^2\lambda_{\max}(\bm{\Sigma}_{\bm{e}})\lambda_{\max}^{1/2}(\bm{P}\bm{P}^\top)\cdot\sqrt{(q+q')/T}$, we have
	\begin{equation}\begin{split}
			&\mathbb{P}\left[\sup_{\bm{W}\in\mathcal{W}(1;q,q')}\left\langle\frac{1}{T}\sum_{t=1}^T(\cm{Y}_{t-1}\circ\cm{E}_t)_{[I]},\bm{W}\right\rangle\geq C\kappa^2\lambda_{\max}(\bm{\Sigma_e})\lambda_{\max}^{1/2}(\bm{P}\bm{P}^\top)\sqrt{\frac{q+q'}{T}}\right]\\
			\leq&\exp[-C(q+q')].
	\end{split}\end{equation}	
	Finally, since $\bm{P}$ is related to the VMA($\infty$) process, by the spectral measure of ARMA process discussed in \citet{basu2015regularized}, we may replace $\lambda_{\max}(\bm{P}\bm{P}^\top)$ with $1/\mu_{\min}(\mathcal{A})$.
\end{proof}

\begin{lemma}
	\label{lemma:deviation2}
	Suppose we simplify the notation of $\bm{A}^*$ to $\bm{A}$. For any $\bm{M}\in\mathbb{R}^{p\times p}$ such that $\|\bm{M}\|_{\textup{F}}=1$, denote $R_T(\bm{M})=\sum_{t=0}^{T-1}\|\bm{M}\bm{y}_t\|_2^2$. Then, for any $t>0$,
	\begin{equation}
		\begin{split}
			&\mathbb{P}[|R_T(\bm{M})-\mathbb{E}R_T(\bm{M})|\geq t]\\
			\leq&2\exp\left(-\min\left(\frac{t^2}{\kappa^4T\lambda_{\max}^2(\bm{\Sigma_e})\lambda_{\max}^2(\bm{PP}^\top)},\frac{t}{\kappa^2\lambda_{\max}^2(\bm{\Sigma_e})\lambda_{\max}^2(\bm{PP}^\top)}\right)\right),
		\end{split}
	\end{equation}
	where $\bm{P}$ is defined as
	\begin{equation}
		\label{eq:P}
		\bm{P}=\begin{bmatrix}
			\bm{I}_p & \bm{A} & \bm{A}^2 & \bm{A}^3 & \dots & \bm{A}^{T-1} & \dots\\
			\bm{O} & \bm{I}_p & \bm{A} & \bm{A}^2 & \dots & \bm{A}^{T-2} & \dots\\
			\vdots & \vdots & \vdots & \vdots & \ddots & \vdots & \dots \\
			\bm{O} & \bm{O} & \bm{O} & \bm{O} & \dots & \bm{I}_p & \dots
		\end{bmatrix}.
	\end{equation}
	
\end{lemma}

\begin{proof}
	Let $\bm{y}=(\bm{y}_{T-1}^\top,\bm{y}_{T-2}^\top,\dots,\bm{y}_0^\top)^\top$, $\bm{e}=(\bm{e}_{T-1}^\top,\bm{e}_{T-2}^\top,\dots,\bm{e}_0^\top,\dots)^\top$, and $\bm{\xi}=(\bm{\xi}_{T-1}^\top,\bm{\xi}_{T-2}^\top,\dots,\bm{\xi}_0^\top,\dots)^\top$. Based on the moving average representation of VAR(1), we can rewrite $\bm{y}_t$ to a VMA($\infty$), $\bm{y}_t=\bm{e}_t+\bm{A}\bm{e}_{t-1}+\bm{A}^2\bm{e}_{t-2}+\bm{A}^3\bm{e}_{t-2}+\cdots$. Note that $R_T(\bm{M})=\bm{y}^\top(\bm{I}_T\otimes \bm{M}^\top\bm{M})\bm{y}=\bm{e}^\top\bm{P}^\top(\bm{I}_T\otimes\bm{M}^\top\bm{M})\bm{P}\bm{e}=\bm{\xi}^\top\bm{D}\bm{P}^\top(\bm{I}_T\otimes\bm{M}^\top\bm{M})\bm{P}\bm{D}\bm{\xi}:=\bm{\xi}^\top\bm{\Sigma_M}\bm{\xi}$, where $P$ is defined in \eqref{eq:P} and
	\begin{equation}
		\bm{D}=\begin{bmatrix}
			\bm{\Sigma}_{\bm{e}}^{1/2} & \bm{O} & \bm{O} & \dots\\
			\bm{O} & \bm{\Sigma}_{\bm{e}}^{1/2} & \bm{O} & \dots\\
			\bm{O} & \bm{O} & \bm{\Sigma}_{\bm{e}}^{1/2} & \dots\\
			\vdots & \vdots & \vdots & \ddots
		\end{bmatrix}.
	\end{equation}
	
	By Hanson-Wright inequality, for any $t>0$,
	\begin{equation}
		\mathbb{P}[|R_T(\bm{M})-\mathbb{E}R_T(\bm{M})|\geq t]\leq 2\exp\left(-\min\left(\frac{t^2}{\kappa^4\|\bm{\Sigma_M}\|_{\text{F}}^2},\frac{t}{\kappa^2\|\bm{\Sigma_M}\|_{\text{op}}}\right)\right).
	\end{equation}
	
	As $\|\bm{M}\|_{\text{F}}=1$, by the submultiplicative property of the Frobenius norm and operator norm, we have $\|\bm{\Sigma_M}\|_{\textup{F}}^2\leq T\cdot\lambda_{\max}^2(\bm{\Sigma_e})\lambda_{\max}^2(\bm{P}\bm{P}^\top)$ and $\|\bm{\Sigma_M}\|_{\text{op}}\leq\lambda_{\max}(\bm{\Sigma_e})\lambda_{\max}(\bm{P}\bm{P}^\top)$. These imply that, for any $t>0$,
	\begin{equation}
		\begin{split}
			&\mathbb{P}[|R_T(\bm{M})-\mathbb{E}R_T(\bm{M})|\geq t]\\
			\leq& 2\exp\left(-\min\left(\frac{t^2}{\kappa^4T\lambda_{\max}^2(\bm{\Sigma_e})\lambda_{\max}^2(\bm{P}\bm{P}^\top)},\frac{t}{\kappa^2\lambda_{\max}(\bm{\Sigma_e})\lambda_{\max}(\bm{P}\bm{P}^\top)}\right)\right).
		\end{split}
	\end{equation}
	The proof of this lemma is accomplished.	
\end{proof}

\begin{lemma} \label{lemma:covering}
	(Covering number of unit sphere) Let $\mathcal{N}$ be an $\varepsilon$-net of the unit sphere $\mathbb{S}^{p-1}$, where $\varepsilon\in(0,1]$. Then,
	\begin{equation}
		|\mathcal{N}|\leq\left(\frac{3}{\varepsilon}\right)^p.
	\end{equation}
\end{lemma}
\begin{proof}
	This lemma follows directly from Corollary 4.2.13 of \cite{Vershynin2018}.
\end{proof}

\section{Proofs for Non-Convex Estimation}\label{append:OR}

In this Appendix, we present the theoretical analysis for the NC estimation approach. In Appendix \ref{sec:B.1}, we present the proof of Theorem \ref{thm:comp_conv} by providing the local convergence guarantees for the algorithm. The optimization error of the gradient descent iteratives are shown to converge linearly to zero, and the dominating term is the statistical error rate after a sufficient number of iterations. In Appendix \ref{sec:B.2}, we present the statistical convergence analysis. To be specific, we prove the stochastic RSC, RSS and deviation conditions in the analysis of the local convergence analysis. In Appendix \ref{sec:B.3}, we present the theoretical justification of the ridge-type ratio estimator for rank selection. Auxiliary lemmas are presented at the end of each appendix.

\subsection{Proof of Local Convergence Analysis}\label{sec:B.1}

\begin{proof}[Proof of Theorem \ref{thm:comp_conv}]
	
	The proof generally follows that of Theorem 3.1 of \cite{han2020optimal}. We devide the proof into four steps. For brevity, we focus on the case with $d=3$, and the results can be readily extended to the general case of $d>3$. In Step 1, we introduce the notations and conditions. In Steps 2 and 3, we present some intermediate results. Finally, in Step 4, we present the local convergence results and verify the conditions imposed in Step 1.\\
	
	\noindent\textit{Step 1 (Notations and Conditions)}
	
	We first introduce some notations used in the proof. Let $\cm{A}^*=[\![\cm{G}^*;\bm{U}_1^*,\dots,\bm{U}_{6}^*]\!]$ such that $\bm{U}^{*\top}\bm{U}_i^*=\bm{I}_{r_i}$. For each step $s=0,1,\dots,I$, we define
	\begin{equation}
		\begin{split}
			E^{(s)} & = \min_{\bm{R}_i\in\mathbb{O}_{r_i}}\left\{\sum_{i=1}^{6}\left\|\bm{U}_i^{(s)}-\bm{U}_i^*\bm{R}_i\right\|_\text{F}^2+\left\|\cm{G}^{(s)}-[\![\cm{G}^*;\bm{R}_1^\top,\dots,\bm{R}_{6}^\top]\!]\right\|_\text{F}^2\right\},\\
			(\bm{R}_1^{(s)},\dots,\bm{R}_{6}^{(s)}) & = {\arg\min}_{\bm{R}_i\in\mathbb{O}_{r_i}}\left\{\sum_{i=1}^{6}\left\|\bm{U}_i^{(s)}-\bm{U}_i^*\bm{R}_i\right\|_\text{F}^2+\left\|\cm{G}^{(s)}-[\![\cm{G}^*;\bm{R}_1^\top,\dots,\bm{R}_{6}^\top]\!]\right\|_\text{F}^2\right\},\\
			\bm{V}_i^{(s)} & = \left(\otimes_{j\neq i}\bm{U}_j^{(s)}\right)\cm{G}^{(s)\top}_{(i)}~\text{and}~\cm{A}^{(s)} = [\![\cm{G}^{(s)};\bm{U}_1^{(s)},\dots,\bm{U}_{6}^{(s)}]\!].
		\end{split}
	\end{equation}
	For any $s=0,1,\dots,I$, we have
	\begin{equation}
		\|\bm{U}_i^{(s)}\|_\text{op}\leq 1.01b,~\|\cm{G}^{(s)}_{(i)}\|_\text{op}\leq 1.01b,~\text{for}~i=1,\dots,2d,~\text{and}~E^{(s)}\leq c_0\alpha\beta^{-1}\rho^{-2}.
	\end{equation}
	
	By definition, for the given sample size $T$, $\overline{\mathcal{L}}$ is restricted strongly convex (RSC) with parameter $\alpha$ and restricted strongly smooth (RSS) with parameter $\beta$, such that for any tensors $\cm{A}_1,\cm{A}_2\in\mathbb{R}^{p_1\times\dots\times p_d\times p_1\times\cdots\times p_d}$ with Tucker ranks $(r_1,\dots,r_{2d})$
	\begin{equation}
		\frac{\alpha}{2}\|\cm{A}_1-\cm{A}_2\|_\text{F}^2\leq\overline{\mathcal{L}}(\cm{A}_1)-\overline{\mathcal{L}}(\cm{A}_2)-\langle\nabla\overline{\mathcal{L}}(\cm{A}_2),\cm{A}_1-\cm{A}_2\rangle\leq\frac{\beta}{2}\|\cm{A}_1-\cm{A}_2\|_\text{F}^2.
	\end{equation}
	The $\alpha$-RSC condition implies that
	\begin{equation}
		\overline{\mathcal{L}}(\cm{A}_1)\geq\overline{\mathcal{L}}(\cm{A}_2)+\langle\nabla\overline{\mathcal{L}}(\cm{A}_1),\cm{A}_1-\cm{A}_2\rangle+\frac{\alpha}{2}\|\cm{A}_1-\cm{A}_2\|_\text{F}^2,
	\end{equation}
	and as in \cite{nesterov2003introductory}, the convexity and $\beta$-RSS condition jointly imply that
	\begin{equation}
		\overline{\mathcal{L}}(\cm{A}_2)-\overline{\mathcal{L}}(\cm{A}_1)\geq\langle\nabla\overline{\mathcal{L}}(\cm{A}_1),\cm{A}_2-\cm{A}_1\rangle+\frac{1}{2\beta}\|\nabla\mathcal{L}(\cm{A}_2)-\nabla\mathcal{L}(\cm{A}_1)\|_\text{F}^2.
	\end{equation}
	Combining these two inequalities, we have that
	\begin{equation}\label{eq:RSC}
		\langle\nabla\overline{\mathcal{L}}(\cm{A}_1)-\nabla\overline{\mathcal{L}}(\cm{A}_2),\cm{A}_1-\cm{A}_2\rangle\geq\frac{\alpha}{2}\|\cm{A}_1-\cm{A}_2\|_\text{F}^2+\frac{1}{2\beta}\|\nabla\overline{\mathcal{L}}(\cm{A}_1)-\nabla\overline{\mathcal{L}}(\cm{A}_2)\|_\text{F}^2,
	\end{equation}
	which is also known as the restricted correlated gradient condition in \cite{han2020optimal}. Additonally, by definition, we immediately have $\alpha\leq\beta$.\\
	
	\noindent\textit{Step 2. (Upper bound of $E^{(s+1)}-E^{(s)}$)}
	
	By definition,
	\begin{equation}
		\begin{split}
			E^{(s+1)} & = \sum_{i=1}^{6}\left\|\bm{U}_i^{(s+1)}-\bm{U}_i^*\bm{R}_i^{(s+1)}\right\|_\text{F}^2+\left\|\cm{G}^{(s)}-[\![\cm{G}^*;\bm{R}_1^{(s+1)\top},\dots,\bm{R}_{6}^{(s+1)\top}]\!]\right\|_\text{F}^2\\
			& \leq \sum_{i=1}^{6}\left\|\bm{U}_i^{(s+1)}-\bm{U}_i^*\bm{R}_i^{(s)}\right\|_\text{F}^2 + \left\|\cm{G}^{(s)}-[\![\cm{G}^*;\bm{R}_1^{(s)\top},\dots,\bm{R}_{6}^{(s)\top}]\!]\right\|_\text{F}^2.
		\end{split}
	\end{equation}
	Note that $\bm{U}_i^{(s+1)}=\bm{U}_i^{(s)}-\eta\nabla_{\bm{U}_i}\mathcal{L}^{(s)}-\eta a\bm{U}_i^{(s)}(\bm{U}_i^{(s)\top}\bm{U}_i^{(s)}-b^2\bm{I}_{r_i})$, for $i=1,\dots,2d$. Then,
	\begin{equation}
		\begin{split}
			& \|\bm{U}_i^{(s+1)}-\bm{U}_i^*\bm{R}_i^{(s)}\|_\text{F}^2\\
			= & \|\bm{U}_i^{(s)}-\bm{U}_i^*\bm{R}_i^{(s)}\|_\text{F}^2 + \eta^2\underbrace{\|\nabla\overline{\mathcal{L}}(\cm{A}^{(s)})_{(i)}\bm{V}_i^{(s)}+a\bm{U}_i^{(s)}(\bm{U}_i^{(s)\top}\bm{U}_i^{(s)}-b^2\bm{I}_{r_i})\|_\text{F}^2}_{T_{1i}}\\
			& -2\eta\underbrace{\left\langle\bm{U}_i^{(s)}-\bm{U}_i^*\bm{R}_i^{(s)},\nabla\overline{\mathcal{L}}(\cm{A}^{(s)})_{(i)}\bm{V}_i^{(s)}\right\rangle}_{T_{2i}}-2\eta a\underbrace{\left\langle\bm{U}_i^{(s)}-\bm{U}_i^*\bm{R}_i^{(s)},\bm{U}_i^{(s)}(\bm{U}_i^{(s)\top}\bm{U}_i^{(s)}-b^2\bm{I}_{r_i})\right\rangle}_{T_{3i}}.
		\end{split}
	\end{equation}
	
	For $T_{1i}$, $i=1,\dots,6$,
	\begin{equation}
		\begin{split}
			T_{1i} & \leq 2\|\nabla\overline{\mathcal{L}}(\cm{A}^{(s)})_{(i)}\bm{V}_i^{(s)}\|_\text{F}^2 +2a^2\|\bm{U}_i^{(s)}(\bm{U}_i^{(s)\top}\bm{U}_i^{(s)}-b^2\bm{I}_{r_i})\|_\text{F}^2\\
			& \leq 4\|\nabla\overline{\mathcal{L}}(\cm{A}^*)_{(i)}\bm{V}_i^{(s)}\|_\text{F}^2+4\|[\nabla\overline{\mathcal{L}}(\cm{A}^{(s)})-\nabla\overline{\mathcal{L}}(\cm{A}^*)]_{(i)}\bm{V}_i^{(s)}\|_\text{F}^2\\
			& + 2a^2\|\bm{U}_i^{(s)}\|_\text{op}^2\cdot\|\bm{U}_i^{(s)\top}\bm{U}_i^{(s)}-b^2\bm{I}_{r_i}\|_\text{F}^2\\
			& \leq 4\|\cm{G}^{(s)}_{(i)}\|_\text{op}^2\cdot\|\otimes_{j=1,n\neq i}^{2d}\bm{U}_j^{(s)}\|_\text{op}^2\cdot\left(\xi^2 + \|\nabla\overline{\mathcal{L}}(\cm{A}^{(s)})-\nabla\overline{\mathcal{L}}(\cm{A}^*)\|_\text{F}^2\right)\\
			& + 2a^2\|\bm{U}_i^{(s)}\|_\text{op}^2\cdot\|\bm{U}_i^{(s)\top}\bm{U}_i^{(s)}-b^2\bm{I}_{r_i}\|_\text{F}^2\\
			& \leq 5\bar{\sigma}^2b^{-2}\left(\xi^2 + \|\nabla\overline{\mathcal{L}}(\cm{A}^{(s)})-\nabla\overline{\mathcal{L}}(\cm{A}^*)\|_\text{F}^2\right)+3a^2b^2\|\bm{U}_i^{(s)\top}\bm{U}_i^{(s)}-b^2\bm{I}_{r_i}\|_\text{F}^2:=Q_{i,2}.
		\end{split}
	\end{equation}
	
	For $T_{2i}$ and $T_{3i}$, $i=1,\dots,6$,
	\begin{equation}
		\begin{split}
			T_{2i} & = \left\langle\bm{U}_i^{(s)}\bm{V}_i^{(s)\top}-\bm{U}_i^*\bm{R}_i^{(s)}\bm{V}_i^{(s)\top},\nabla\overline{\mathcal{L}}(\cm{A}^{(s)})_{(i)}\right\rangle\\
			& = \left\langle \cm{A}^{(s)} - [\![\cm{G}^{(s)};\bm{U}_1^{(s)},\dots,\bm{U}_i^*\bm{R}_i^{(s)}\dots,\bm{U}_{2d}^{(s)},]\!],\nabla\overline{\mathcal{L}}(\cm{A}^{(s)})\right\rangle = \left\langle \cm{A}^{(s)} - \cm{A}_i^{(s)},\nabla\overline{\mathcal{L}}(\cm{A}^{(s)})\right\rangle
		\end{split}
	\end{equation}
	and
	\begin{equation}
		\begin{split}
			& T_{2i} + a\cdot T_{3i} = \left\langle \cm{A}^{(s)} - \cm{A}_i^{(s)},\nabla\overline{\mathcal{L}}(\cm{A}^{(s)})\right\rangle+\left\langle\bm{U}_i^{(s)}-\bm{U}_i^*\bm{R}_i^{(s)},a\bm{U}_i^{(s)}(\bm{U}_i^{(s)\top}\bm{U}_i^{(s)}-b^2\bm{I}_{r_i})\right\rangle\\
			= & \left\langle \cm{A}^{(s)} - \cm{A}_i^{(s)},\nabla\overline{\mathcal{L}}(\cm{A}^{(s)})\right\rangle+a\left\langle\bm{U}_i^{(s)\top}\bm{U}_i^{(s)}-\bm{U}_i^{(s)\top}\bm{U}_i^*\bm{R}_i^{(s)},\bm{U}_i^{(s)\top}\bm{U}_i^{(s)}-b^2\bm{I}_{r_i}\right\rangle\\
			\geq & \left\langle \cm{A}^{(s)} - \cm{A}_i^{(s)},\nabla\overline{\mathcal{L}}(\cm{A}^{(s)})\right\rangle + \frac{a}{4}\|\bm{U}_i^{(s)\top}\bm{U}_i^{(s)}-b^2\bm{I}_{r_i}\|_\text{F}^2-\frac{a}{4}E^{(s)}\|\bm{U}_i^{(s)}-\bm{U}_i^*\bm{R}_i^{(s)}\|_\text{F}^2:=Q_{i,1}.
		\end{split}
	\end{equation}
	Therefore, we have
	\begin{equation}
		\|\bm{U}_i^{(s+1)}-\bm{U}_i^*\bm{R}_i^{(s)}\|_\text{F}^2\\
		\leq \|\bm{U}_i^{(s)}-\bm{U}_i^*\bm{R}_i^{(s)}\|_\text{F}^2 -2\eta Q_{i,1} + \eta^2Q_{i,2}.
	\end{equation}
	
	Similarly, we can show that
	\begin{equation}
		\left\|\cm{G}^{(s+1)}-[\![\cm{G}^*;\bm{R}_1^{(s)\top},\dots,\bm{R}_{2d}^{(s)\top}]\!]\right\|_\text{F}^2 = \left\|\cm{G}^{(s)}-[\![\cm{G}^*;\bm{R}_1^{(s)\top},\dots,\bm{R}_{2d}^{(s)\top}]\!]\right\|_\text{F}^2 - 2\eta Q_{S,1} + \eta^2Q_{S,2}
	\end{equation}
	where
	$Q_{\scalebox{0.7}{\cm{G}},1}:=\left\langle\cm{A}^{(s)}-\cm{A}_{\scalebox{0.7}{\cm{G}}}^{(s)},\nabla\overline{\mathcal{L}}(\cm{A}^{(s)})\right\rangle$,
	$\cm{A}_{\scalebox{0.7}{\cm{G}}}^{(s)}=[\![\cm{G}^*;\bm{U}_1^{(s)}\bm{R}_1^{(s)\top},\dots,\bm{U}_6^{(s)}\bm{R}_{6}^{(s)\top}]\!]$, and
	$Q_{\scalebox{0.7}{\cm{G}},2}:=4b^{12}(\xi^2+\|\nabla\overline{\mathcal{L}}(\cm{A}^{(s)})-\nabla\overline{\mathcal{L}}(\cm{A}^*)\|_\text{F}^2)$.
	
	Therefore, we have
	\begin{equation}
		E^{(s+1)} \leq E^{(s)} - 2\eta\left(Q_{\scalebox{0.7}{\cm{G}},1}+\sum_{i=1}^{6}Q_{i,1}\right) + \eta^2\left(Q_{\scalebox{0.7}{\cm{G}},2}+\sum_{i=1}^{6}Q_{i,2}\right).
	\end{equation}~
	
	\noindent\textit{Step 3. Lower bound of $Q_{\scalebox{0.7}{\cm{G}},1}+\sum_{i=1}^{6}Q_{i,1}$}
	
	By definition, we have
	\begin{equation}\label{eq:Q1}
		\begin{split}
			Q_{\scalebox{0.7}{\cm{G}},1} + \sum_{i=1}^{6}Q_{i,1}
			=&\left\langle7\cm{A}^{(s)}-\cm{A}_{\scalebox{0.7}{\cm{G}}}^{(s)}-\sum_{i=1}^{6}\cm{A}_i^{(s)},\nabla\overline{\mathcal{L}}(\cm{A}^{(s)})\right\rangle\\
			&+a\sum_{i=1}^{6}\left(\frac{1}{4}\|\bm{U}_i^{(s)\top}\bm{U}_i^{(s)}-\bm{I}_{r_i}\|_\text{F}^2-\frac{1}{4}E^{(s)}\|\bm{U}_i^{(s)}-\bm{U}_i^*\bm{R}_i^{(s)}\|_\text{F}^2\right)\\
			=&\left\langle\cm{A}^{(s)}-\cm{A}^*+\cm{H}^{(s)},\nabla\overline{\mathcal{L}}(\cm{A}^{(s)})\right\rangle\\
			&+a\sum_{i=1}^{6}\left(\frac{1}{4}\|\bm{U}_i^{(s)\top}\bm{U}_i^{(s)}-\bm{I}_{r_i}\|_\text{F}^2-\frac{1}{4}E^{(s)}\|\bm{U}_i^{(s)}-\bm{U}_i^*\bm{R}_i^{(s)}\|_\text{F}^2\right),
		\end{split}
	\end{equation}
	where $\cm{H}^{(s)}=\cm{A}^*-\cm{A}_{\scalebox{0.7}{\cm{G}}}^{(s)}-\sum_{i=1}^{6}(\cm{A}_i^{(s)}-\cm{A}^{(s)})$. By Lemma \ref{lemma:tensor_perturbation},
	\begin{equation}
		\|\cm{H}^{(s)}\|_\text{F}\leq B_2B_3^3+6B_1B_2B_3^{5/2}+15B_1^2B_2B_3^2+20B_1^3B_2B_3^{3/2}+15B_1^4B_2B_3+6B_1^5B_3
	\end{equation}
	where
	\begin{equation}
		\begin{split}	
			B_1&:=\max_{1\leq i\leq 6}\{\|\bm{U}_i^{(s)}\|_\textup{op},\|\bm{U}_i^*\|_\textup{op}\},~B_2:=\max_{1\leq i\leq 6}\{\|\cm{G}^{(s)}_{(i)}\|_\textup{op},\|\cm{G}_{(i)}^*\|_\textup{op}\},\\
			B_3&:=\max_{1\leq i\leq 6}\{\|\cm{H}^{(s)}_{\scalebox{0.7}{\cm{G}}}\|_\textup{F}^2,\|\bm{H}^{(s)}_i\|_\textup{F}^2\}.
		\end{split}
	\end{equation}
	Since $B_1\leq 1.01$, $B_2\leq 1.01$ and $B_3\leq E^{(s)}\leq c_0\alpha\beta^{-1}\rho^{-2}$, we have $\|\cm{H}^{(s)}\|_\text{F}\leq \alpha\beta^{-1}E^{(s)}/4$.
	
	By \eqref{eq:RSC}, the first term on the right hand side of \eqref{eq:Q1} can be further bounded as
	\begin{equation}
		\begin{split}
			& \left\langle\cm{A}^{(s)}-\cm{A}^*+\cm{H}^{(s)},\nabla\overline{\mathcal{L}}(\cm{A}^{(s)})\right\rangle = \langle\cm{A}^{(s)}-\cm{A}^*,\nabla\overline{\mathcal{L}}(\cm{A}^{(s)})-\nabla\overline{\mathcal{L}}(\cm{A}^*)\rangle\\
			&+\left\langle\cm{A}^{(s)}-\cm{A}^*+\cm{H}^{(s)},\nabla\overline{\mathcal{L}}(\cm{A}^{*})\right\rangle + \left\langle\cm{H}^{(s)},\nabla\overline{\mathcal{L}}(\cm{A}^{(s)})-\nabla\overline{\mathcal{L}}(\cm{A}^{*})\right\rangle\\
			&\geq\frac{\alpha}{2}\|\cm{A}^{(i)}-\cm{A}^*\|_\text{F}^2+\frac{1}{2\beta}\|\nabla\overline{\mathcal{L}}(\cm{A}^{(i)})-\nabla\overline{\mathcal{L}}(\cm{A}^*)\|_\text{F}^2-\|\cm{H}^{(s)}\|_\text{F}\cdot\|\nabla\overline{\mathcal{L}}(\cm{A}^{(s)})-\nabla\overline{\mathcal{L}}(\cm{A}^*)\|_\text{F}\\
			&-\left|\left\langle\cm{A}^{(s)}-\cm{A}^*+\cm{H}^{(s)},\nabla\overline{\mathcal{L}}(\cm{A}^{*})\right\rangle\right|.
		\end{split}
	\end{equation}
	In addition, we have that for any $c_1>0$,
	\begin{equation}
		\begin{split}
			& \|\cm{H}^{(s)}\|_\text{F}\cdot\|\nabla\overline{\mathcal{L}}(\cm{A}^{(s)})-\nabla\overline{\mathcal{L}}(\cm{A}^*)\|_\text{F}\\
			\leq & \frac{1}{4\beta}\|\nabla\overline{\mathcal{L}}(\cm{A}^{(i)})-\nabla\overline{\mathcal{L}}(\cm{A}^*)\|_\text{F}^2+\beta\|\cm{H}^{(s)}\|_\text{F}^2\\
			\leq & \frac{1}{4\beta}\|\nabla\mathcal{L}(\cm{A}^{(i)})-\nabla\mathcal{L}(\cm{A}^*)\|_\text{F}^2+\frac{\alpha E^{(s)}}{4} 
		\end{split}
	\end{equation}
	and
	\begin{equation}
		\begin{split}
			& \left|\left\langle\cm{A}^{(s)}-\cm{A}^*+\cm{H}^{(s)},\nabla\overline{\mathcal{L}}(\cm{A}^{*})\right\rangle\right|\\
			=&\left|\left\langle7\cm{A}^{(s)}-\cm{A}_{\scalebox{0.7}{\cm{G}}}^{(s)}-\sum_{i=1}^6\cm{A}_i^{(s)},\nabla\overline{\mathcal{L}}(\cm{A}^*)\right\rangle\right|\\
			\leq&\left|\left\langle\cm{A}^{(s)}-\cm{A}_{\scalebox{0.7}{\cm{G}}}^{(s)},\nabla\overline{\mathcal{L}}(\cm{A}^*)\right\rangle\right|+\sum_{i=1}^6\left|\left\langle\cm{A}^{(s)}-\cm{A}_i^{(s)},\nabla\overline{\mathcal{L}}(\cm{A}^*)\right\rangle\right|\\
			\leq&\xi\cdot\|\cm{G}^{(s)}-[\![\cm{G}^*;\bm{R}_1^{(s)\top},\dots,\bm{R}_6^{(s)\top}]\!]\|_\text{F}\cdot\|\bm{U}_1^{(s)}\|_\text{op}\cdots\|\bm{U}_6^{(s)}\|_\text{op}\\
			+&\xi\cdot\sum_{i=1}^6\|\cm{G}^{(s)}_{(i)}\|_\text{op}\cdot\|\bm{U}_i^{(s)}-\bm{U}_i^*\bm{R}_i^{(s)}\|_\text{F}\cdot\|\otimes_{j\neq i}\bm{U}_j^{(s)}\|_\text{op}\\
			\leq&\xi\sqrt{E^{(s)}}\left(1.01^6\times7\right)\leq14c_1E^{(s)}+\frac{1}{4c_1}\xi^2,
		\end{split}
	\end{equation}
	for any $c_1>0$. Thus, we have
	\begin{equation}
		\begin{split}
			& \left\langle\cm{A}^{(s)}-\cm{A}^*+\cm{H}^{(s)},\nabla\overline{\mathcal{L}}(\cm{A}^{(s)})\right\rangle\\
			\geq & \frac{\alpha}{2}\|\cm{A}^{(i)}-\cm{A}^*\|_\text{F}^2+\frac{1}{4\beta}\|\nabla\overline{\mathcal{L}}(\cm{A}^{(i)})-\nabla\overline{\mathcal{L}}(\cm{A}^*)\|_\text{F}^2-(\alpha/4+14c_1)E^{(s)}-\frac{1}{4c_1}\xi^2.
		\end{split}
	\end{equation}
	Now, applying Lemma \ref{lemma:tensor_perturbation2}, we have
	\begin{equation}
		E^{(s)}\leq b^{-12}(64+24\underline{\sigma}^{-2}C_1)\|\cm{A}^{(s)}-\cm{A}^*\|_\text{F}^2+2C_1b^{-2}\sum_{i=1}^6\|\bm{U}_i^{(s)\top}\bm{U}_i^{(s)}-b^2\bm{I}_{r_i}\|_\text{F}^2.
	\end{equation}
	
	Combining these inequalities, by setting $c_1=\alpha b^{12}/112$, we have
	\begin{equation}
		\begin{split}
			& Q_{\scalebox{0.7}{\cm{G}},1} + \sum_{i=1}^{6}Q_{i,1}\\
			\geq & \left(\frac{\alpha b^{12}}{4}-14c_1\right)E^{(s)}+\frac{1}{4\beta}\|\nabla\overline{\mathcal{L}}(\cm{A}^{(s)})-\nabla\overline{\mathcal{L}}(\cm{A}^*)\|_\text{F}^2-\frac{1}{4c_1}\xi^2\\
			&+\frac{a}{8}\sum_{i=1}^6\|\bm{U}_i^{(s)\top}\bm{U}_i^{(s)}-b^2\bm{I}_{r_i}\|_\text{F}^2\\
			= & \frac{\alpha b^{12}}{8}E^{(s)}+\frac{1}{4\beta}\|\nabla\overline{\mathcal{L}}(\cm{A}^{(s)})-\nabla\overline{\mathcal{L}}(\cm{A}^*)\|_\text{F}^2-28\alpha^{-1}b^{12}\xi^2+\frac{a}{8}\sum_{i=1}^6\|\bm{U}_i^{(s)\top}\bm{U}_i^{(s)}-b^2\bm{I}_{r_i}\|_\text{F}^2.
		\end{split}
	\end{equation}~
	
	\noindent\textit{Step 4. Convergence analysis of $E^{(s)}$ and verification of conditions}
	
	In the following, we combine all the results in previous steps to establish the error bound for $E^{(s)}$ and $\|\cm{A}^{(s)}-\cm{A}^*\|_\text{F}$. Plugging in $b=\overline{\sigma}^{1/7}$, $a=C\alpha\rho^{-2}\overline{\sigma}$ to $Q_{\scalebox{0.7}{\cm{G}},2}$ and $Q_{i,2}$, we have
	\begin{equation}
		\begin{split}
			&Q_{\scalebox{0.7}{\cm{G}},2}+\sum_{i=1}^{6}Q_{i,2}\\
			\leq & 5\overline{\sigma}^{12/7}[\xi^2+\|\nabla\overline{\mathcal{L}}(\cm{A}^{(s)})-\nabla\overline{\mathcal{L}}(\cm{A}^*)\|_\text{F}^2]+\frac{C\alpha^2\overline{\sigma}^{16/7}}{\rho^4}\sum_{i=1}^6\|\bm{U}_i^{(s)\top}\bm{U}_i^{(s)}-b^2\bm{I}_{r_i}\|_\text{F}^2.
		\end{split}
	\end{equation}
	
	Combining the results above, as $\eta=\eta_0\beta^{-1}\rho^{-2}\overline{\sigma}^{-12/7}$ for any $\eta_0>1/25088$, we have
	\begin{equation}
		\begin{split}
			E^{(s+1)} \leq & \left(1-\frac{\alpha b^{-12}\eta}{8}\right)E^{(s)} \\
			+ & \left(184\overline{\sigma}^{12/7}-\frac{\eta}{2\beta}\right)\|\nabla\overline{\mathcal{L}}(\cm{A}^{(s)})-\nabla\overline{\mathcal{L}}(\cm{A}^*)\|_\text{F}^2
			+ \left(\frac{\eta}{2c_1}+C\overline{\sigma}_1^{12/7}\eta^2\right)\xi^2\\
			+ & \left(\frac{C\eta^2\alpha^2\overline{\sigma}}{\rho^4}-\frac{C\eta\alpha\overline{\sigma}^{6/7}}{\rho^2}\right)\sum_{i=1}^6\|\bm{U}_i^{(s)\top}\bm{U}_i^{(s)}-b^2\bm{I}_{r_i}\|_\text{F}^2\\
			\leq & (1-C\eta_0\alpha\beta^{-1}\rho^{-2})E^{(s)}+\rho^2\alpha^{-2}\xi^2
		\end{split}
	\end{equation}
	
	By induction, we have that for any $i=1,2,\dots$,
	\begin{equation}
		E^{(s)}\leq (1-C\eta_0\alpha\beta^{-1}\rho^{-2})^sE^{(0)}+C\rho^2\alpha^{-2}\overline{\sigma}^{12/7}\xi^2.
	\end{equation}
	For the error bound of $\|\cm{A}^{(s)}-\cm{A}^{(0)}\|_\text{F}$, by Lemma \ref{lemma:tensor_perturbation2},
	\begin{equation}\label{eq:recursive_A}
		\begin{split}
			&\|\cm{A}^{(s)}-\cm{A}^*\|_\text{F}^2\leq C\overline{\sigma}^{12/7}E^{(i)}\\
			\leq & C\overline{\sigma}^{12/7}(1-C\eta_0\alpha\beta^{-1}\rho^{-2})^iE^{(0)}+C\rho^2\alpha^{-2}\xi^2\\
			\leq & C\rho^2(1-C\eta_0\alpha\beta^{-1}\rho^{-2})^i\|\cm{A}^{(0)}-\cm{A}^*\|_\text{F}^2+C\rho^2\alpha^{-2}\xi^2.
		\end{split}
	\end{equation}
	
	Finally, we show that conditions $E^{(s)}\leq$ hold. Since $\bm{U}_i^{(0)\top}\bm{U}_i^{(0)}=b^2\bm{I}_{r_i}$ for $i=1,\dots,2d$, by Lemma \ref{lemma:tensor_perturbation2} and initialization bound $\|\cm{A}^{(0)}-\cm{A}^*\|_\text{F}\leq C\underline{\sigma}\alpha^{1/2}\beta^{-1/2}$, we have
	\begin{equation}
		E^{(0)}\leq (C\overline{\sigma}^{12/7}+C\underline{\sigma}^{-2})\|\cm{A}^{(0)}-\cm{A}^*\|_\text{F}^2\leq C\overline{\sigma}^{-5/7}\rho^2\|\cm{A}^{(0)}-\cm{A}^*\|_\text{F}^2\leq C\alpha\beta^{-1}\rho^{-2}.
	\end{equation}
	Based on the recursive relationship in \eqref{eq:recursive_A}, by induction it is easy to check that $E^{(s)}\leq C\overline{\sigma}^{2/7}\alpha\beta^{-1}\rho^{-2}$ for all $s\geq1$. In other words, as $\alpha\beta^{-1}\leq 1$ and $\rho^2\geq1$, we have $E^{(s)}\leq Cb^2$, which further implies that
	\begin{equation}
		\begin{split}
			\|\bm{U}_i^{(s)}\|_\text{op}&\leq\|\bm{U}_i^*\bm{O}_i^{(s)}\|_\text{op}+\|\bm{U}_i^{(s)}-\bm{U}_i^*\bm{O}_i^{(s)}\|_\text{op}\\
			&\leq b+\|\bm{U}_i^{(s)}-\bm{U}_i^*\bm{O}_i^{(s)}\|_\text{F}\leq(1+c_b)b
		\end{split}
	\end{equation}
	and for $i=1,\dots,2d$,
	\begin{equation}
		\begin{split}
			\|\cm{G}^{(s)}_{(i)}\|_\text{op}&\leq\|\bm{O}_1^{(s)\top}\cm{G}^*_{(i)}\otimes_{j\neq i}\bm{O}_j^{(s)}\|_\text{op} + \|\cm{G}^{(s)}_{(i)}-\bm{O}_1^{(s)\top}\cm{G}^*_{(i)}\otimes_{j\neq i}\bm{O}_j^{(s)}\|_\text{op}\\
			&\leq \overline{\sigma}b^{-6}+\|\cm{G}^{(s)}_{(i)}-\bm{O}_1^{(s)\top}\cm{G}^*_{(i)}\otimes_{j\neq i}\bm{O}_j^{(s)}\|_\text{F}\leq (1+c_b)b.
		\end{split}
	\end{equation}
\end{proof}

Next, we present some auxiliary lemmas for the proof of local convergence.

\begin{lemma}\label{lemma:tensor_perturbation}
	Consider the case of $d=3$. Suppose that $\cm{A}^*=[\![\cm{G}^*;\bm{U}_1^*,\dots,\bm{U}_6^*]\!]$ and $\cm{A}=[\![\cm{G};\bm{U}_1,\dots,\bm{U}_6]\!]$. For $\bm{R}_i\in\mathbb{O}_{r_i}$ and $i=1,\dots,6$, let
	\begin{equation}
		\begin{split}
			\cm{A}_i=&[\![\cm{G};\bm{U}_1,\dots,\bm{U}_{i-1},\bm{U}_i^*\bm{R}_i,\bm{U}_{i+1},\dots,\bm{U}_{6}]\!],~~\cm{A}_{\scalebox{0.7}{\cm{G}}}=[\![([\![\cm{G}^*;\bm{R}_1,\dots,\bm{R}_{6}]\!]);\bm{U}_1,\dots,\bm{U}_{6}]\!]\\
			\bm{H}_i=&\bm{U}_i^*-\bm{U}_i\bm{R}_i^\top,~~\cm{H}_{\scalebox{0.7}{\cm{G}}}=\cm{G}^*-[\![\cm{G};\bm{R}_1,\dots,\bm{R}_{6}]\!],~~\text{and}~~\cm{H}=\cm{A}^*-\cm{A}_{\scalebox{0.7}{\cm{G}}}-\sum_{i=1}^{6}(\cm{A}_i-\cm{A}).
		\end{split}
	\end{equation}
	Then, we have that
	\begin{equation}
		\|\cm{H}\|_\textup{F}\leq B_2B_3^3+6B_1B_2B_3^{5/2}+15B_1^2B_2B_3^2+20B_1^3B_2B_3^{3/2}+15B_1^4B_2B_3+6B_1^5B_3.
	\end{equation}
	where
	\begin{equation}
		\begin{split}	
			B_1&:=\max_{1\leq i\leq 6}\{\|\bm{U}_i\|_\textup{op},\|\bm{U}_i^*\|_\textup{op}\},~B_2:=\max_{1\leq i\leq 6}\{\|\cm{G}_{(i)}\|_\textup{op},\|\cm{G}_{(i)}^*\|_\textup{op}\},\\
			B_3&:=\max_{1\leq i\leq 6}\{\|\cm{H}_{\scalebox{0.7}{\cm{G}}}\|_\textup{F}^2,\|\bm{H}_i\|_\textup{F}^2\}.
		\end{split}
	\end{equation}
	
\end{lemma}

\begin{proof}
	Since $\cm{G}^*=[\![\cm{S};\bm{R}_1,\dots,\bm{R}_{6}]\!]+\cm{H}_{\scalebox{0.7}{\cm{G}}}$, we have
	\begin{equation}
		\cm{A}^*=[\![\cm{G};\bm{U}_1^*\bm{R}_1,\dots,\bm{U}_{6}^*\bm{R}_{6}]\!] + [\![\cm{H}_{\scalebox{0.7}{\cm{G}}};\bm{U}_1^*,\dots,\bm{U}_{6}^*]\!].
	\end{equation}
	For the first term on the right hand of the above equation, we have
	\begin{equation}
		\begin{split}
			& [\![\cm{G};\bm{U}_1+\bm{H}_1\bm{R}_1,\dots,\bm{U}_{6}+\bm{H}_{6}\bm{R}_{6}]\!]\\
			=&\cm{A}+\sum_{i=1}^{6}\cm{G}\times_{j=1,j\neq i}^n\bm{U}_j\times_i\bm{H}_i\bm{R}_i+\cm{H}_1\\
			=&\cm{A}+\sum_{i=1}^{6}\cm{G}\times_{j=1,j\neq i}^n\bm{U}_j\times_i(\bm{U}^*_i\bm{R}_i-\bm{U}_i)+\cm{H}_1\\
			=&\sum_{i=1}^{6}\cm{A}_i-5\cm{A}+\cm{H}_1
		\end{split}
	\end{equation}
	where
	\begin{equation}
		\cm{H}_1=\sum_{i\neq j}\cm{G}\times_{k=1,k\neq i,k\neq j}^{6}\bm{U}_k\times_i\bm{H}_i\bm{R}_i\times_j\bm{H}_j\bm{R}_j+\cdots+\cm{G}\times_{i=1}^6\bm{H}_i\bm{R}_i.
	\end{equation}
	
	For the second term, we have
	\begin{equation}
		[\![\cm{H}_{\scalebox{0.7}{\cm{G}}};\bm{U}_1^*,\dots,\bm{U}_{6}^*]\!]= (\cm{G}^*-\cm{G}\times_{i=1}^{6}\bm{R}_i)\times_{i=1}^{6}(\bm{H}_i+\bm{U}_i\bm{R}_i^\top)=\cm{A}_{\scalebox{0.7}{\cm{G}}}-\cm{A}+\cm{H}_2
	\end{equation}
	where
	\begin{equation}
		\cm{H}_2=\cm{H}_{\scalebox{0.7}{\cm{G}}}\times_{i=1}^6\bm{H}_i+\sum_{i=1}^6\cm{H}_{\scalebox{0.7}{\cm{G}}}\times_{j\neq i}\bm{H}_j\times_i\bm{U}_i\bm{R}_i^\top+\cdots+\sum_{i=1}^6\cm{H}_{\scalebox{0.7}{\cm{G}}}\times_{j\neq i}\bm{U}_j\bm{R}_j^\top\times_i\bm{H}_i.
	\end{equation}
	Then, it follows that
	\begin{equation}
		\cm{A}^*=\cm{A}_{\scalebox{0.7}{\cm{G}}}+\sum_{i=1}^6\cm{A}_i-6\cm{A}+(\cm{H}_1+\cm{H}_2)
	\end{equation}
	and
	\begin{equation}
		\begin{split}
			\cm{H}&=\cm{H}_1+\cm{H}_2\\
			&=\cm{G}^*\times_{i=1}^6\bm{H}_i+\sum_{i=1}^6\cm{G}^*\times_{j\neq i}\bm{H}_j\times_i\bm{U}_i+\cdots+\sum_{i=1}^6\cm{H}_{\scalebox{0.7}{\cm{G}}}\times_{j\neq i}\bm{U}_j\bm{R}_j^\top\times_i\bm{H}_i.
		\end{split}
	\end{equation}
	Hence, we have
	\begin{equation}
		\|\cm{H}\|_\text{F}\leq B_2B_3^3+6B_1B_2B_3^{5/2}+15B_1^2B_2B_3^2+20B_1^3B_2B_3^{3/2}+15B_1^4B_2B_3+6B_1^5B_3.
	\end{equation}
\end{proof}

\begin{lemma}\label{lemma:tensor_perturbation2}
	
	Consider the case of $d=3$. Suppose $\cm{A}^*=[\![\cm{G}^*;\bm{U}_1^*,\dots,\bm{U}_6^*]\!]$, $\bm{U}^{*\top}_i\bm{U}_i=\bm{I}_{r_i}$, $i=1,2,3$, $\bar{\sigma}=\max_{1\leq i\leq 6}\|\cm{G}^*_{(i)}\|_\textup{op}$, and $\underline{\sigma}=\min_{1\leq i\leq 6}\sigma_{r_i}(\cm{G}^*_{(i)})$. Let $\cm{A}=[\![\cm{G};\bm{U}_1,\dots,\bm{U}_6]\!]$ be another Tucker low-rank tensor with $\|\bm{U}_i\|_\textup{op}\leq (1+c_0)$, and $\max_{1\leq i\leq 6}\|\cm{G}_{(i)}\|_\textup{op}\leq (1+c_0)\bar{\sigma}$ for some constant $c_0>0$. Define
	\begin{equation}
		E:=\min_{\bm{R}_i\in\mathbb{O}_{r_i}}\left(\sum_{i=1}^6\|\bm{U}_i-\bm{U}_i^*\bm{R}_i\|_\textup{F}^2+\|\cm{G}-[\![\cm{G}^*;\bm{R}_1,\dots,\bm{R}_6]\!]\|_\textup{F}^2\right).
	\end{equation}
	Then, we have
	\begin{equation}
		\begin{split}
			&E\leq (64+24\underline{\sigma}^{-2}C_{1})b^{-12}\|\cm{A}-\cm{A}^*\|_\textup{F}^2+2C_{1}b^{-2}\sum_{i=1}^6\min_{\bm{R}_i\in\mathbb{O}_{r_i}}\|\bm{U}_i^\top\bm{U}_i-b^2\bm{I}_{r_i}\|_\textup{F}^2.\\
			&\|\cm{A}-\cm{A}^*\|_\textup{F}^2\leq 7b^{12}(1+C_2^2\overline{\sigma}^{2}b^{-14})E.
		\end{split}
	\end{equation}
	where $C_1$ and $C_2$ are some universal constants.
	
\end{lemma}

\begin{proof}
	
	First, note that
	\begin{equation}
		\|[\![\cm{G};\bm{R}_1,\dots,\bm{R}_6]\!]-\cm{G}\|_\text{F}=b^{-6}\|[\![\cm{G};\bm{U}_1^*\bm{R}_1,\dots,\bm{U}_6^*\bm{R}_6]\!]-\cm{A}^*\|_\text{F}.
	\end{equation}
	Then, by the inequality of means, we can
	\begin{equation}
		\begin{split}
			&\|[\![\cm{G};\bm{U}_1+\bm{U}_1^*\bm{R}_1-\bm{U}_1,\dots,\bm{U}_6+\bm{U}_6^*\bm{R}_6-\bm{U}_6]\!]-\cm{A}^*\|_\text{F}^2\\
			=&\left\|(\cm{A}-\cm{A}^*)+\sum_{i=1}^6\cm{G}\times_i(\bm{U}_i^*\bm{R}_i-\bm{U}_i)\times_{j\neq i}\bm{U}_j+\cdots+\cm{G}\times_{i=1}^6(\bm{U}^*_i\bm{R}_i-\bm{U}_i)\right\|_\text{F}^2\\
			\leq&64b^{-12}\|\cm{A}-\cm{A}^*\|_\text{F}^2+64b^{-12}\sum_{i=1}^6\|\cm{G}\times_i(\bm{U}_i^*\bm{R}_i-\bm{U}_i)\times_{j\neq i}\bm{U}_j\|_\text{F}^2+\cdots+ 64b^{-12}\|\cm{G}\times_{i=1}^6(\bm{U}^*_i\bm{R}_i-\bm{U}_i)\|_\text{F}^2\\
			\leq & 64b^{-12}\|\cm{A}-\cm{A}^*\|_\text{F}^2+64b^{-12}\sum_{i=1}^6\|\cm{G}_{(i)}\|_\text{op}^2\cdot\|\otimes_{j\neq i}\bm{U}_j\|_\text{op}^2\cdot\|\bm{U}_i^*\bm{R}_i-\bm{U}_i\|_\text{F}^2\\
			& + 64b^{-12}\sum_{i\neq j}\|\cm{G}_{(i)}\|_\text{op}^2\cdot\|\otimes_{k\neq i,k\neq j}\bm{U}_k\|_\text{op}^2\cdot\|\bm{U}_i^*\bm{R}_i-\bm{U}_i\|_\text{op}^2\cdot\|\bm{U}_j^*\bm{R}_j-\bm{U}_j\|_\text{F}^2\\
			& + \cdots\\
			& + (64/6)b^{-12}\sum_{i=1}^6\|\cm{G}_{(i)}\|_\text{op}^2\cdot\|\otimes_{j\neq i}(\bm{U}_j^*\bm{R}_j-\bm{U}_j)\|_\text{op}^2\cdot\|\bm{U}_i^*\bm{R}_i-\bm{U}_i\|_\text{F}^2\\
			\leq & 64b^{-12}\|\cm{A}-\cm{A}^*\|_\text{F}^2\\
			&+b^{-12}\Bigg[64\bar{\sigma}^2(1+c_0)^{12}+160\bar{\sigma}^2(1+c_0)^{10}(2+c_0)^2+\frac{640}{3}\bar{\sigma}^2(1+c_0)^8(2+c_0)^4+160\bar{\sigma}^2(1+c_0)^6\\
			&(2+c_0)^6+64\bar{\sigma}^2(1+c_0)^{4}(2+c_0)^8+\frac{32}{3}\bar{\sigma}^2(1+c_0)^2(2+c_0)^{10}\Bigg]\sum_{i=1}^6\|\bm{U}_i^*\bm{R}_i-\bm{U}_i\|_\text{F}^2.
		\end{split}
	\end{equation}
	Since the above inequality holds for any orthogonal matrices $\bm{R}_1,\dots,\bm{R}_6$, it follows that
	\begin{equation}
		E\leq 64b^{-12}\|\cm{A}-\cm{A}^*\|_\text{F}^2+\overline{\sigma}^2b^{-14}C_{1}\sum_{i=1}^6\min_{\bm{R}_i\in\mathbb{O}_{r_i}}\|\bm{U}_i^*\bm{R}_i-\bm{U}_i\|_\text{F}^2.
	\end{equation}
	By Lemma E.2 in \cite{han2020optimal},
	\begin{equation}
		\min_{\bm{R}_i\in\mathbb{O}_{r_i}}\|\bm{U}_i-\bm{U}_i^*\bm{R}_i\|_\text{F}^2\leq 2\|\bm{U}_i^\top\bm{U}_i-\bm{I}_{r_i}\|_\text{F}^2+4\underline{\sigma}^{-2}\|\cm{A}-\cm{A}^*\|_\text{F}^2,
	\end{equation}
	which implies that
	\begin{equation}
		E\leq (64+24\underline{\sigma}^{-2}C_{1})\|\cm{A}-\cm{A}^*\|_\text{F}^2+2C_{1}\sum_{i=1}^6\min_{\bm{R}_i\in\mathbb{O}_{r_i}}\|\bm{U}_i^\top\bm{U}_i-\bm{I}_{r_i}\|_\text{F}^2.
	\end{equation}
	
	For the second inequality, denote the optimal rotation matrices by
	\begin{equation}
		(\bm{R}_1,\dots,\bm{R}_6)=\argmin_{\bm{R}_i\in\mathbb{O}_{r_i}}\left\{\sum_{i=1}^6\|\bm{U}_i-\bm{U}_i^*\bm{R}_i\|_\text{F}^2+\|\cm{G}-[\![\cm{G}^*;\bm{R}_1^\top,\cdots,\bm{R}_6^\top]\!]\|_\text{F}^2\right\}.
	\end{equation}
	Let $\cm{H}_{\scalebox{0.7}{\cm{G}}}=\cm{G}^*-[\![\cm{G};\bm{R}_1,\dots,\bm{R}_6]\!]$ and $\bm{H}_i=\bm{U}_i^*-\bm{U}_i\bm{R}_i^\top$. Then, we have
	\begin{equation}
		\cm{A}^*=(\cm{H}_{\scalebox{0.7}{\cm{G}}}+[\![\cm{G};\bm{R}_1,\dots,\bm{R}_6]\!])\times_{i=1}^6(\bm{H}_i+\bm{U}_i\bm{R}_i^\top)
	\end{equation}
	and it follows that
	\begin{equation}
		\begin{split}
			&\|\cm{A}-\cm{A}^*\|_\text{F}\\
			\leq & \|\cm{H}_{\scalebox{0.7}{\cm{G}}}\times_{i=1}^{6}\bm{U}_i^*\|_\text{F}+\sum_{i=1}^6\|\cm{G}\times_{j\neq i}\bm{U}_j\times_i\bm{H}_i\bm{R}_i\|_\text{F}+\cdots + \|\cm{G}\times_{i=1}^6\bm{H}_i\bm{R}_i\|_\text{F}\\
			\leq & b^6\|\cm{H}_{\scalebox{0.7}{\cm{G}}}\|_\text{F} + \Bigg[\bar{\sigma}(1+c_0)^6+\frac{5}{2}\bar{\sigma}(1+c_0)^5(2+c_0)+\frac{10}{3}\bar{\sigma}(1+c_0)^4(2+c_0)^2\\
			&+\frac{5}{2}\bar{\sigma}(1+c_0)^3(2+c_0)^3+\bar{\sigma}(1+c_0)^2(2+c_0)^4+\frac{1}{6}\bar{\sigma}(1+c_0)(2+c_0)^5\Bigg]\sum_{i=1}^6\|\bm{H}_i\|_\text{F}\\
			=&\|\cm{H}_{\scalebox{0.7}{\cm{G}}}\|_\text{F}+C_{2}\sum_{i=1}^6\|\bm{H}_i\|_\text{F}.
		\end{split}
	\end{equation}
	Thus, we have
	\begin{equation}
		\|\cm{A}-\cm{A}^*\|_\text{F}^2\leq 7b^{12}\|\cm{H}_{\scalebox{0.7}{\cm{G}}}\|_\text{F}^2+7C_{2}^2b^{-2}\sum_{i=1}^6\|\bm{H}_i\|_\text{F}^2.
	\end{equation}
\end{proof}

\subsection{Statistical Convergence Analysis}\label{sec:B.2}

In this appendix, we present the stochastic properties of the time series data. The main technique is the martingale-based concentration inequalities introduced in Appendix \ref{append:high-dim}.

\begin{proof}[Proof of Theorem \ref{thm:stat_conv}]
	
	The proof of Theorem \ref{thm:stat_conv} follows the deterministic computational convergence analysis in Theorem \ref{thm:comp_conv}. It suffices to show that the RSC, RSS and deviation bound conditions hold with high probability.
	
	By Lemmas \ref{lem:RSC_RSS} and $\ref{lem:deviation}$, with probability at least $1-2\exp[-CM_2^2\min(\tau^{-4},\tau^{-2})T]-C\exp(-C\max_{1\leq i\leq d}p_i)$, the empirical loss function $\overline{\mathcal{L}}$ satisfies the RSC-$\alpha_\text{RSC}$ and RSS-$\beta_\text{RSS}$ conditions, and
	\begin{equation}
		\xi(r_1,\dots,r_{2d})\lesssim\kappa^2M_1\sqrt{\frac{\prod_{i=1}^{2d}r_i+\sum_{i=1}^{2d}p_ir_i}{T}}.
	\end{equation}
	
	By Theorem \ref{thm:comp_conv}, we have that, for all $i=1,2,\dots$,
	\begin{equation}
		\begin{split}
			&\|\cm{A}^{(i)}-\cm{A}^*\|_\text{F}^2\\
			\lesssim&~\rho^2(1-C\eta_0\alpha_\text{RSC}\beta_\text{RSS}^{-1}\rho^{-2})^i\|\cm{A}^{(0)}-\cm{A}^*\|_\text{F}^2+\rho^2\alpha_\text{RSC}^{-2}\xi(r_1,\dots,r_{2d}).
		\end{split}
	\end{equation}
	Hence, when
	\begin{equation}
		I\gtrsim\frac{\log(\alpha_\text{RSC}\beta^{-1}_\text{RSS}\xi^2(r_1,\dots,r_{2d}))-\log(\|\cm{A}^{(0)}-\cm{A}^*\|_\text{F}^2)}{\log(1-C\eta_0\alpha_\text{RSC}\beta^{-1}_\text{RSS}\rho^{-2})},
	\end{equation}
	the optimization error is absorbed by the statistical error, so
	\begin{equation}
		\|\cm{A}^{(I)}-\cm{A}^*\|_\text{F}^2\lesssim\rho^2\alpha_\text{RSC}^{-2}\kappa^2M_1\sqrt{\frac{\sum_{i=1}^{2d}p_ir_i+\prod_{i=1}^{2d}r_i}{T}}.
	\end{equation}
\end{proof}

In the following, we prove the restricted strong convexity (RSC) and restricted strong smoothness (RSS) conditions. For the least squares loss function $\overline{\mathcal{L}}(\cm{A})=(2T)^{-1}\sum_{t=1}^T\|\cm{Y}_t-\langle\cm{A},\cm{Y}_{t-1}\rangle\|_\textup{F}^2$, it is easy to check that for any $\cm{A}_1,\cm{A}_2\in\mathbb{R}^{p_1\times\cdots\times p_d\times p_1\times\cdots\times p_d}$,
\begin{equation}
	\begin{split}
		&\overline{\mathcal{L}}(\cm{A}_1)-\overline{\mathcal{L}}(\cm{A}_2)-\langle\nabla\cm{A}_1-\cm{A}_2,\overline{\mathcal{L}}(\cm{A}_2)\rangle\\
		=&\frac{1}{2T}\sum_{t=1}^T\|\langle\cm{A}_1-\cm{A}_2,\cm{Y}_{t-1}\rangle\|_\textup{F}^2
		=\frac{1}{2T}\sum_{t=0}^{T-1}\|(\cm{A}_1-\cm{A}_2)_{[S_2]}\bm{y}_t\|_2^2.
	\end{split}
\end{equation}

\begin{lemma}\label{lem:RSC_RSS}
	
	Assume the conditions in Theorem \ref{thm:stat_conv} hold and $T\gtrsim M_2^{-2}\max(\kappa^2,\kappa^4)\max_{1\leq i\leq d}p_i$. For any tensor $\bm{\Delta}\in\mathbb{R}^{p_1\times\cdots\times p_d\times p_1\times\cdots\times p_d}$ of Tucker ranks $(2r_1,2r_2,\dots,2r_{2d})$, with probability at least $1-2\exp[-CM_2^2\min(\kappa^{-2},\kappa^{-4})T]$,
	\begin{equation}
		\alpha_\textup{RSC}\|\bm{\Delta}\|_\textup{F}^2\leq\frac{1}{T}\sum_{t=1}^T\|\langle\bm{\Delta},\cm{Y}_{t-1}\rangle\|_\textup{F}^2\leq\beta_\textup{RSS}\|\bm{\Delta}\|_\textup{F}^2,
	\end{equation}
	where $\alpha_\textup{RSC}=\lambda_{\min}(\bm{\Sigma}_{\bm{e}})/(2\mu_{\max}(\mathcal{A}))$ and $\beta_\textup{RSS}=(3\lambda_{\max}(\bm{\Sigma}_{\bm{e}}))/(2\mu_{\min}(\mathcal{A}))$.
	
\end{lemma}

\begin{proof}
	
	Denote $\mathcal{T}(p_1,\dots,p_{2d};r_1,\dots,r_{2d})=\{\cm{T}\in\mathbb{R}^{p_1\times\cdots\times p_{2d}}:\|\cm{T}\|_\text{F}=1,\text{rank}(\cm{T}_{(i)})=r_i,~\text{for}~i=1,\dots,2d\}$ as the set of low-rank tensors of given dimensions and Tucker ranks.
	It suffices to prove the result for $\bm{\Delta}\in\mathcal{T}(p_1,\dots,p_{2d};2r_1,\dots,2r_{2d})$.
	
	For any $\bm{M}\in\mathbb{R}^{m\times p}$, denote $R_T(\bm{M})=\sum_{t=0}^{T-1}\|\bm{M}\bm{y}_t\|_2^2$. Note that $R_T(\bm{\Delta}_{[S_2]})\geq\mathbb{E}R_T(\bm{\Delta}_{[S_2]})-\sup_{\bm{\Delta}}|R_T(\bm{\Delta}_{[S_2]})-\mathbb{E}R_T(\bm{\Delta}_{[S_2]})|$. Similarly to the proof of Lemma \ref{lemma:RSC}, we have that $\mathbb{E}R_T(\bm{\Delta}_{[S_2]})\geq T\lambda_{\min}(\bm{\Sigma}_{\bm{e}})\mu_{\max}^{-1}(\mathcal{A})$.
	
	For any $\bm{M}\in\mathbb{R}^{p\times p}$ such that $\|\bm{M}\|_\text{F}=1$ and any $t>0$, similarly to Lemma \ref{lemma:RSC}, by the VMA($\infty$) representation of VAR(1) model, we have
	\begin{equation}
		\begin{split}
			&\mathbb{P}[|R_T(\bm{M})-\mathbb{E}R_T(\bm{M})|\geq t]\\
			\leq & 2\exp\left(-\min\left(\frac{t^2}{\kappa^4T\lambda_{\max}^2(\bm{\Sigma}_{\bm{e}})\mu_{\min}^{-2}(\mathcal{A})},\frac{t}{\kappa^2\lambda_{\max}^2(\bm{\Sigma}_{\bm{e}})\mu_{\min}^{-2}(\mathcal{A})}\right)\right).
		\end{split}
	\end{equation}
	
	Considering an $\epsilon$-covering net $\overline{\mathcal{T}}$ for $\mathcal{T}(p_1,\dots,p_{2d};2r_1,\dots,2r_{2d})$; in other words, for any $\cm{T}\in\mathcal{T}(p_1,\dots,p_{2d};r_1,\dots,r_{2d})$. there exists a $\overline{\cm{T}}\in\overline{\mathcal{T}}$ such that $\|\cm{T}-\overline{\cm{T}}\|_\text{F}\leq\epsilon$.
	By Lemma \ref{lemma:tensor_covering}, we have that $|\overline{\mathcal{T}}|\leq((6d+3)/\epsilon)^{\prod_{i=1}^{2d}2r_i+\sum_{i=1}^{2d}2p_ir_i}$. Then, for some small $\epsilon$, the deviation bound between $R_T(\cm{T}_{[S_2]})$ and its expection can be bounded uniformly over $\mathcal{T}(p_1,\dots,p_{2d};2r_1,\dots,2r_{2d})$ by
	\begin{equation}
		\begin{split}
			&\mathbb{P}\left[\sup_{\scalebox{0.75}{\cm{T}}\in\mathcal{T}(p_1,\dots,p_{2d};r_1,\dots,r_{2d})}|R_T(\bm{\Delta}_{[S_2]})-\mathbb{E}R_T(\bm{\Delta}_{[S_2]})|\geq t\right]\\
			\leq & 2\exp\left[C\left(\prod_{i=1}^{2d}2r_i+\sum_{i=1}^{2d}2p_ir_i\right)-\min\left(\frac{t^2}{\kappa^4T\lambda_{\max}^2(\bm{\Sigma}_{\bm{e}})\mu_{\min}^{-2}(\mathcal{A})},\frac{t}{\kappa^2\lambda_{\max}^2(\bm{\Sigma}_{\bm{e}})\mu_{\min}^{-2}(\mathcal{A})}\right)\right].
		\end{split}
	\end{equation}
	Letting $t=T\lambda_{\min}(\bm{\Sigma}_{\bm{e}})\mu_{\max}^{-1}(\mathcal{A})/2$, for $T\gtrsim M_2^{-2}\max(\kappa^{-4},\kappa^{-2})(\prod_{i=1}^{2d}r_i+\sum_{i=1}^{2d}p_ir_i)$, we have
	\begin{equation}
		\mathbb{P}\left[\sup_{\bm{\Delta}\in\mathcal{T}(p_1,\dots,p_{2d};r_1,\dots,r_{2d})}|R_T(\bm{\Delta}_{[S_2]})-\mathbb{E}R_T(\bm{\Delta}_{[S_2]})|\geq\frac{T\lambda_{\min}(\bm{\Sigma_e})}{2\mu_{\max}(\mathcal{A})}\right]\leq C\exp\left[-c\left(\prod_{i=1}^{2d}r_i+\sum_{i=1}^{2d}p_ir_i\right)\right]
	\end{equation}
	and thus
	\begin{equation}
		\mathbb{P}\left[\frac{1}{2T}\sum_{t=1}^T\|\langle\bm{\Delta},\cm{Y}_{t-1}\rangle\|_\textup{F}^2\leq\frac{\lambda_{\min}(\bm{\Sigma_e})}{2\mu_{\max}(\mathcal{A})}\right]\leq C\exp\left[-c\left(\prod_{i=1}^{2d}r_i+\sum_{i=1}^{2d}p_ir_i\right)\right].
	\end{equation}
	
	Similarly, $R_T(\bm{\Delta}_{[S_2]})\leq\mathbb{E}R_T(\bm{\Delta}_{[S_2]})+\sup_{\bm{\Delta}}|R_T(\bm{\Delta}_{[S_2]})-\mathbb{E}R_T(\bm{\Delta}_{[S_2]})|$ and $\mathbb{E}R_T(\bm{\Delta}_{[S_2]})\leq T\lambda_{\max}(\bm{\Sigma_e})\mu_{\min}^{-1}(\mathcal{A})$. Therefore, the deviation above implies that with high probability, $T^{-1}\sum_{t=1}^T\|\langle\bm{\Delta},\cm{Y}_{t-1}\rangle\|_\textup{F}^2\leq(3\lambda_{\max}(\bm{\Sigma_e}))/(2\mu_{\min}(\mathcal{A}))$.
\end{proof}

Next, we prove the deviation bound for $\xi(r_1,\dots,r_{2d})$. For the least squares loss function $\overline{\mathcal{L}}(\cm{A})=(2T)^{-1}\sum_{t=1}^T\|\cm{Y}_t-\langle\cm{A},\cm{Y}_{t-1}\rangle\|_\text{F}^2$, it is clear that $\nabla\overline{\mathcal{L}}(\cm{A}^*)=T^{-1}\sum_{t=1}^T\cm{Y}_{t-1}\circ\cm{E}_t$.

\begin{lemma}\label{lem:deviation}
	
	Assume conditions in Theorem \ref{thm:stat_conv} hold and $T\gtrsim M_2^{-2}\max(\kappa^{-4},\kappa^{-2})\max_{1\leq i\leq d}p_i$. With probability at least $1-\exp[-C(\prod_{i=1}^{2d}r_i+\sum_{i=1}^{2d}p_ir_i)]$,
	\begin{equation}
		\begin{split}
			\xi(r_1,\dots,r_{2d}):=&\sup_{\|\scalebox{0.75}{\cm{G}}\|_\textup{F}=1,\bm{U}_i^\top\bm{U}_i=\bm{I}_{r_i}}\left\langle\nabla\overline{\mathcal{L}}(\cm{A}^*),[\![\cm{G};\bm{U}_1,\dots,\bm{U}_{2d}]\!]\right\rangle\\
			\lesssim&~\kappa^2M_1\sqrt{\frac{\prod_{i=1}^{2d}r_i+\sum_{i=1}^{2d}p_ir_i}{T}},
		\end{split}
	\end{equation}
	where $M_1=\lambda_{\max}(\bm{\Sigma_e})/\mu_{\min}^{1/2}(\mathcal{A})$.
	
\end{lemma}

\begin{proof}
	
	For simplicity, we consider the case of $d=3$, and the result can be easily extended to the general case of a fixed $d>3$.
	
	Denote $\mathcal{T}(p_1,\dots,p_6;r_1,\dots,r_{6})=\{\cm{T}\in\mathbb{R}^{p_1\times\cdots\times p_6}:\|\cm{T}\|_\text{F}=1,\text{rank}(\cm{T}_{(i)})=r_i,~\text{for}~i=1,\dots,6\}$ as the set of low-rank tensors of given dimensions and Tucker ranks. By definition,
	\begin{equation}
		\xi(r_1,\dots,r_6)=\sup_{\scalebox{0.7}{\cm{T}}\in\mathcal{T}(p_1,\dots,p_6;r_1,\dots,r_6)}\left\langle\frac{1}{T}\sum_{t=1}^T\cm{Y}_{t-1}\circ\cm{E}_t,\cm{T}\right\rangle.
	\end{equation}
	
	First, we consider an $\epsilon$-net $\overline{\mathcal{T}}(p_1,\dots,p_6;r_1,\dots,r_6)$ for $\mathcal{T}(p_1,\dots,p_6;r_1,\dots,r_6)$. For any tensor $\cm{T}\in\mathcal{T}(p_1,\dots,p_6;r_1,\dots,r_6)$, there exists a tensor $\overline{\cm{T}}\in\mathcal{T}(p_1,\dots,p_6;r_1,\dots,r_6)$ such that $\|\cm{T}-\overline{\cm{T}}\|_\text{F}\leq\epsilon$. Obviously, $\bm{\Delta}=\cm{T}-\overline{\cm{T}}$ is a tensor of Tucker ranks $(2r_1,\dots,2r_6)$. Based on the HOSVD of $\bm{\Delta}$, we can split $\bm{\Delta}$ into $64$ orthogonal components by splitting each of the factor matrices into two equal-size groups and splitting the core tensor correspondingly. In other words, we can write $\bm{\Delta}=\sum_{i=1}^{64}\bm{\Delta}_i$, where each $\bm{\Delta}_i$ is a tensor of Tucker ranks $(r_1,\dots,r_6)$ and $\langle\bm{\Delta}_i,\bm{\Delta}_j\rangle=0$ for all $i\neq j$.
	
	By Cauchy's inequality, as $\|\bm{\Delta}\|_\text{F}^2=\sum_{i=1}^{64}\|\bm{\Delta}_i\|_\textup{F}^2$, we have $\sum_{i=1}^{64}\|\bm{\Delta}_i\|_\text{F}\leq8\|\bm{\Delta}\|_\text{F}\leq8\epsilon$. Moreover, since $\bm{\Delta}_i/\|\bm{\Delta}_i\|_\text{F}\in\mathcal{T}(p_1,\dots,p_6;r_1,\dots,r_6)$,
	\begin{equation}
		\begin{split}
			\xi(r_1,\dots,r_6)&\leq\max_{\scalebox{0.7}{\cm{T}}\in\overline{\mathcal{T}}(p_1,\dots,p_6;r_1,\dots,r_6)}\left\langle\frac{1}{T}\sum_{t=1}^T\cm{Y}_{t-1}\circ\cm{E}_t,\cm{T}\right\rangle+\sum_{i=1}^{64}\left\langle\frac{1}{T}\sum_{t=1}^T\cm{Y}_{t-1}\circ\cm{E}_t,\frac{\bm{\Delta}_i}{\|\bm{\Delta}_i\|_\text{F}}\right\rangle\|\bm{\Delta}_i\|_\text{F}\\
			&\leq\max_{\scalebox{0.7}{\cm{T}}\in\overline{\mathcal{T}}(p_1,\dots,p_6;r_1,\dots,r_6)}\left\langle\frac{1}{T}\sum_{t=1}^T\cm{Y}_{t-1}\circ\cm{E}_t,\cm{T}\right\rangle+8\epsilon\cdot\xi(r_1,\dots,r_6).
		\end{split}
	\end{equation}
	It implies that
	\begin{equation}
		\xi(r_1,\dots,r_6)\leq(1-8\epsilon)^{-1}\max_{\scalebox{0.7}{\cm{T}}\in\overline{\mathcal{T}}(p_1,\dots,p_6;r_1,\dots,r_6)}\left\langle\frac{1}{T}\sum_{t=1}^T\cm{Y}_{t-1}\circ\cm{E}_t,\cm{T}\right\rangle.
	\end{equation}
	
	Next, for any fixed $\cm{T}\in\mathbb{R}^{p_1\times\cdots\times p_6}$ such that $\|\cm{T}\|_\text{F}=1$, $\langle\cm{Y}_{t-1}\circ\cm{E}_t,\cm{T}\rangle=\langle\bm{e}_t,\cm{T}_{[S_2]}\bm{y}_{t-1}\rangle$ and we denote $S_t(\cm{T})=\sum_{s=1}^t\langle\bm{e}_s,\cm{T}_{[S_2]}\bm{y}_{s-1}\rangle$ and $R_t(\cm{T})=\sum_{s=0}^{t-1}\|\cm{T}_{[S_2]}\bm{y}_s\|_2^2$, for $1\leq t\leq T$.
	Similar to Lemma \ref{lemma:deviation}, by the standard Chernoff bound, for any $z_1>0$ and $z_2>0$,
	\begin{equation}
		\mathbb{P}[\{S_T(\cm{T})\geq z_1\}\cap\{R_T(\cm{T})\leq z_2\}]\leq\exp\left(-\frac{z_1^2}{2\kappa^2\lambda_{\max}(\bm{\Sigma_e})z_2}\right).
	\end{equation}
	Similar to the proof of Lemma \ref{lemma:deviation}, by Lemma \ref{lemma:deviation2}, we have
	\begin{equation}
		\mathbb{P}[R_T(\cm{T})\geq C\kappa^2T\lambda_{\max}(\bm{\Sigma_e})\mu_{\min}^{-1}(\mathcal{A})]\leq2\exp(-CT).
	\end{equation}
	Therefore, for any $x>0$,
	\begin{equation}
		\begin{split}
			& \mathbb{P}\left[\sup_{\scalebox{0.7}{\cm{T}}\in\mathcal{T}(p_1,\dots,p_6;r_1,\dots,r_6)}\left\langle\frac{1}{T}\sum_{t=1}^T\cm{Y}_{t-1}\circ\cm{E}_t,\cm{T}\right\rangle\geq x\right]\\
			\leq & \mathbb{P}\left[\max_{\scalebox{0.7}{\cm{T}}\in\overline{\mathcal{T}}(p_1,\dots,p_6;r_1,\dots,r_6)}\left\langle\frac{1}{T}\sum_{t=1}^T\cm{Y}_{t-1}\circ\cm{E}_t,\cm{T}\right\rangle\geq (1-8\epsilon)x\right]\\
			\leq & |\overline{\mathcal{T}}(p_1,\dots,p_6;r_1,\dots,r_6)|\cdot\mathbb{P}\left[\left\langle\frac{1}{T}\sum_{t=1}^T\cm{Y}_{t-1}\circ\cm{E}_t,\cm{T}\right\rangle\geq(1-8\epsilon)x\right]\\
			\leq & |\overline{\mathcal{T}}(p_1,\dots,p_6;r_1,\dots,r_6)|\cdot\Big\{\mathbb{P}[\{S_T(\cm{T})\geq T(1-8\epsilon)x\}\cap\{R_T(\cm{T})\leq C\kappa^2T\lambda_{\max}(\bm{\Sigma_e})\mu_{\min}^{-1}(\mathcal{A})\}]\\
			+&\mathbb{P}[R_T(\cm{T})>C\kappa^2\lambda_{\max}(\bm{\Sigma_e})\mu_{\min}^{-1}(\mathcal{A})]\Big\}\\
			\leq & |\overline{\mathcal{T}}(p_1,\dots,p_6;r_1,\dots,r_6)|\cdot\left\{\exp\left[-\frac{CTx^2}{\kappa^4\lambda_{\max}^2(\bm{\Sigma_e})\mu^{-1}_{\min}(\mathcal{A})}\right]+2\exp[-CM_2^{-2}\min(\kappa^{-2},\kappa^{-4})T]\right\}.
		\end{split}
	\end{equation}
	
	By Lemma \ref{lemma:tensor_covering}, $|\overline{\mathcal{T}}(p_1,\dots,p_6;r_1,\dots,r_6)|\leq(21/\epsilon)^{\prod_{i=1}^6r_i+\sum_{i=1}^6p_ir_i}$. Thus, if we take $\epsilon=0.1$ and $x=C\kappa^2\lambda_{\max}(\bm{\Sigma_e})\mu_{\min}^{-1}(\mathcal{A})\sqrt{(\prod_{i=1}^6r_i+\sum_{i=1}^6p_ir_i)/T}$, when $T\gtrsim M_2^{-2}\max(\kappa^4,\kappa^2)\max_{1\leq i\leq 3}p_i$, we have
	\begin{equation}
		\begin{split}
			&\mathbb{P}\left[\sup_{\scalebox{0.7}{\cm{T}}\in\mathcal{T}(p_1,\dots,p_6;r_1,\dots,r_6)}\left\langle\frac{1}{T}\sum_{t=1}^T\cm{Y}_{t-1}\circ\cm{E}_t,\cm{T}\right\rangle\geq \kappa^2\lambda_{\max}(\bm{\Sigma_e})\mu^{-1/2}_{\min}(\mathcal{A})\sqrt{\frac{\prod_{i=1}^6r_i+\sum_{i=1}^6p_ir_i}{T}}\right]\\
			& \leq \exp\left[-C\left(\prod_{i=1}^6r_i+\sum_{i=1}^6p_ir_i\right)\right].
		\end{split}
	\end{equation}
\end{proof}

Next, we present the covering number of the set of low-rank tensors.
\begin{lemma}\label{lemma:tensor_covering}
	
	The $\epsilon$-covering number of the set $\mathcal{T}(p_1,\dots,p_d;r_1,\dots,r_d):=\{\cm{T}\in\mathbb{R}^{p_1\times\cdots\times p_d}:\|\cm{T}\|_\text{F}=1,~\textup{rank}(\cm{T}_{(i)})\leq r_i,~i=1,\dots,d\}$ is
	\begin{equation}
		|\overline{\mathcal{T}}(p_1,\dots,p_d;r_1,\dots,r_d)|\leq[(3d+3)/\epsilon]^{\prod_{i=1}^dr_i+\sum_{i=1}^dp_ir_i}.
	\end{equation}
	
\end{lemma}

\begin{proof}
	
	The proof hinges on the covering number for the low-rank matrix developed by \citet{candes2011tight}. Recall the HOSVD $\cm{T}=[\![\cm{G};\bm{U}_1,\dots,\bm{U}_d]\!]$ where $\|\cm{G}\|_\text{F}=1$ and each $\bm{U}_i$ is an orthonormal matrix. We construct an $\epsilon$-net for $\mathcal{T}$ by covering the set of $\cm{G}$ and all $\bm{U}_i$'s. We take $\overline{G}$ to be an $\epsilon/(d+1)$ net for $\mathcal{G}$ with $\|\overline{G}\|\leq[(3d+3)/\epsilon]^{\prod_{i=1}^dr_i}$. Next, let $O_{p,r}=\{\bm{U}\in\mathbb{R}^{p\times r}:\bm{U}^\top\bm{U}=\bm{I}_r\}$. To cover $O_{p,r}$, it is beneficial to use the $\|\cdot\|_{2,\infty}$ norm, defined as
	\begin{equation}
		\|\bm{M}\|_{2,\infty}=\max_{i}\|\bm{M}_i\|_2,
	\end{equation}
	where $\bm{M}_i$ denotes the $i$th column of $\bm{M}$. Let $Q_{p,r}=\{\bm{M}\in\mathbb{R}^{p\times r}:\|\bm{M}\|_{2,\infty}\leq1\}$. It is obvious that $O_{p,r}\subset Q_{p,r}$, and thus an $\epsilon/(d+1)$-net $\overline{O}_{p,r}$ for $O_{p,r}$ obeying $|\overline{O}_{p,r}|\leq[(3d+3)/\epsilon]^{pr}$.
	
	Denote $\overline{T}=\{[\![\overline{\cm{G}};\overline{\bm{U}}_1,\dots,\overline{\bm{U}}_d]\!]:\overline{\cm{G}}\in\overline{G},\overline{\bm{U}}_i\in\overline{O}_{p_i,r_i},~i=1,\dots,d\}$ and we have $|\overline{T}|\leq|\overline{G}|\times|\overline{O}_{p_1,r_1}|\times\cdots\times|\overline{O}_{p_d,r_d}|=[(3d+3)/\epsilon]^{\prod_{i=1}^dr_i+\sum_{i=1}^dp_ir_i}$. It suffices to show that for any $\cm{T}\in\mathcal{T}(p_1,\dots,p_d;r_1,\dots,r_d)$, there exists a $\overline{\cm{T}}\in\overline{T}$ such that $\|\cm{T}-\overline{\cm{T}}\|_\text{F}\leq\epsilon$.
	
	For any fixed $\cm{T}\in\mathcal{T}(p_1,\dots,p_d;r_1,\dots,r_d)$, decompose it by HOSVD as $\cm{T}=[\![\cm{G};\bm{U}_1,\cdots,\bm{U}_d]\!]$. Then, there exists $\overline{\cm{T}}=[\![\overline{\cm{G}};\overline{\bm{U}}_1,\dots,\overline{\bm{U}}_d]\!]$ with $\overline{\cm{G}}\in\overline{G}$, $\overline{\bm{U}}_i\in\overline{O}_{p_i,r_i}$ satisfying that $\|\bm{U}_i-\overline{\bm{U}}_i\|_{2,\infty}\leq\epsilon/(d+1)$ and $\|\cm{G}-\overline{\cm{G}}\|_\text{F}\leq\epsilon/(d+1)$. This implies that
	\begin{equation}
		\begin{split}
			&\|\cm{T}-\overline{\cm{T}}\|_\text{F}\\
			\leq&\|[\![\cm{G}-\overline{\cm{G}};\bm{U}_1,\dots,\bm{U}_d]\!]\|_\text{F}+\|[\![\overline{\cm{G}};\bm{U}_1-\overline{\bm{U}}_1,\dots,\bm{U}_d]\!]\|_\text{F}+\cdots+\|[\![\overline{\cm{G}};\overline{\bm{U}}_1,\dots,\bm{U}_d-\overline{\bm{U}}_d]\!]\|_\text{F}.
		\end{split}
	\end{equation}
	
	Since each $\bm{U}_i$ is an orthonormal matrix, the first term is $\|\cm{G}-\overline{\cm{G}}\|_\text{F}\leq\epsilon/(d+1)$. For the second term, by the all-orthogonal property of $\overline{\cm{G}}$ and orthonormal property of $\bm{U}_2,\dots,\bm{U}_d$,
	\begin{equation}
		\|[\![\overline{\cm{G}};\bm{U}_1-\overline{\bm{U}}_1,\dots,\bm{U}_d]\!]\|_\text{F}=\|\overline{\cm{G}}\times_1(\bm{U}_1-\overline{\bm{U}}_1)\|_\text{F}\leq\|\overline{\cm{G}}\|_\text{F}\|\bm{U}_1-\overline{\bm{U}}_1\|_{2,\infty}\leq\epsilon/(d+1).
	\end{equation}
	Similarly, we can obtain the same upper bound for the other terms, and thus show that $\|\cm{T}-\overline{\cm{T}}\|_\text{F}\leq\epsilon$.
\end{proof}

\subsection{Rank Selection Consistency}\label{sec:B.3}

This subsection presents the theoretical justification of the proposed ridge-type ratio estimator for rank selection.

\begin{proof}[Proof of Theorem \ref{thm:rankconsistency2}]
	
	The proof of Theorem \ref{thm:rankconsistency2} consists of two steps. First, the results in Theorem \ref{thm:stat_conv} can be readily extended to the rank upper bounds $(\bar{r}_1,\bar{r}_2,\dots,\bar{r}_{2d})$, as the non-asymptotic analysis can be adapted to any ranks not smaller than the true ranks. Under mild conditions on the signal strengths and the number of iterations, the statistical error bound holds 
	\begin{equation}
		\|\widetilde{\cm{A}}-\cm{A}^*\|_\text{F}\lesssim\alpha_\text{RSC}^{-1}\kappa^2M_1\sqrt{\frac{\sum_{i=1}^{2d}p_i\bar{r}_i+\prod_{i=1}^{2d}\bar{r}_i}{T}}\lesssim \alpha_\text{RSC}^{-1}\kappa^2M_1\sqrt{\frac{p_{\max}\bar{r}_{\max}}{T}}=B.
	\end{equation}
	
	Second, by definition, for any tensor $\cm{T}\in\mathbb{R}^{p_1\times\cdots\times p_{2d}}$,
	\begin{equation}
		\|\cm{T}\|_\text{F}^2=\|\cm{T}_{(i)}\|_\text{F}^2=\sum_{j=1}^{p_j}\sigma_j^2(\cm{T}_{(i)}),~~1\leq i\leq 2d.
	\end{equation}
	That is, the Frobenius norm of the error tensor is equivalent to the $\ell_2$ norm of the singular values of any matricization. By Mirsky's singular value inequality,
	\begin{equation}
		\sum_{j=1}^{p_j}[\sigma_j(\widetilde{\cm{A}}_{(i)})-\sigma_j(\cm{A}^*_{(i)})]^2 \leq \sum_{j=1}^{p_j}\sigma_j^2(\widetilde{\cm{A}}_{(i)}-\cm{A}^*_{(i)})=\|\widetilde{\cm{A}}-\cm{A}^*\|_\text{F}^2.
	\end{equation}
	In addition, the $\ell_\infty$ error bound is smaller than the $\ell_2$ error bound, and it directly follows the same upper bound
	\begin{equation}
		\max_{1\leq j\leq \bar{r}_i}|\sigma_j(\widetilde{\cm{A}}_{(i)})-\sigma_j(\cm{A}^*_{(i)})| \leq \left\{\sum_{j=1}^{\bar{r}_i}[\sigma_j(\widetilde{\cm{A}}_{(i)})-\sigma_j(\cm{A}^*_{(i)})]^2\right\}^{1/2}\leq \|\widetilde{\cm{A}}-\cm{A}^*\|_\text{F}\lesssim B.
	\end{equation}
	
	Note that $\sigma_j(\widetilde{\cm{A}}_{(i)})+s(p_{\max},T)=\sigma_j(\cm{A}_{(i)}^*)+[\sigma_j(\widetilde{\cm{A}}_{(i)})-\sigma_j(\cm{A}^*_{(i)})]+s(p_{\max},T)$. For $j>r_i$, since $\sigma_j(\cm{A}_{(i)}^*)=0$ and $\sigma_j(\widetilde{\cm{A}}_{(i)})-\sigma_j(\cm{A}_{(i)}^*)=o_p(s(p_{\max},T))$, $s(p_{\max},T)$ is the dominating term in $\sigma_j(\widetilde{\cm{A}}_{(i)})+s(p_{\max},T)$. For $j\geq r_i$, since $\sigma_j(\widetilde{\cm{A}}_{(i)})-\sigma_j(\cm{A}_{(i)}^*)=o_p(s(p_{\max},T))$ and $s(p_{\max},T)=o(\sigma_j(\cm{A}^*_{(i)}))$, $\sigma_j(\cm{A}^*_{(i)})$ is the dominating term.
	
	Hence, for $j>r_i$, as $T\to\infty$,
	\begin{equation}
		\frac{\sigma_{j+1}(\widetilde{\cm{A}}_{(i)})+c}{\sigma_{j}(\widetilde{\cm{A}}_{(i)})+c}\to\frac{c}{c}=1.
	\end{equation}
	For $j<r_i$,
	\begin{equation}
		\frac{\sigma_{j+1}(\widetilde{\cm{A}}_{(i)})+c}{\sigma_{j}(\widetilde{\cm{A}}_{(i)})+c}\to\frac{\sigma_{j+1}(\cm{A}^*_{(i)})}{\sigma_{j}(\cm{A}^*_{(i)})}.
	\end{equation}
	For $j=r_i$,
	\begin{equation}
		\frac{\sigma_{j+1}(\widetilde{\cm{A}}_{(i)})+c}{\sigma_{j}(\widetilde{\cm{A}}_{(i)})+c}\to\frac{s(p_{\max},T)}{\sigma_{r_i}(\cm{A}^*_{(i)})}\leq\frac{s(p_{\max},T)}{\underline{\sigma}}=o\left(\min_{1\leq i\leq r_i-1}\frac{\sigma_{j+1}(\cm{A}^*_{(i)})}{\sigma_{j}(\cm{A}^*_{(i)})}\right).
	\end{equation}
	Combining these two steps, we can conclude the rank selection consistency in this theorem.
\end{proof}

\section{ADMM Algorithm for (T)SSN Estimator}
\label{append:ADMM}

This subsection  presents the algorithm for the proposed (T)SSN regularized estimator. The algorithm for $\cm{\widehat{A}}_{\textup{SN}}$ can be developed analogously, while $\cm{\widehat{A}}_{\textup{MN}}$ can be obtained easily as in \cite{negahban2011estimation}.

The objective function for the estimator $\cm{\widehat{A}}_{\textup{SSN}}$ in \eqref{eq:SSN_est} can be rewritten as
\begin{equation}\label{eq:objective}
	\mathcal{L}_T(\cm{A})+\lambda_{\text{SSN}}\|\cm{A}\|_{\text{SSN}}=\mathcal{L}_T(\cm{A})+\lambda_{\text{SSN}}\sum_{k=1}^{2^{d-1}}\|\cm{A}_{[I_k]}\|_\textup{nuc},
\end{equation}
where $\mathcal{L}_T(\cm{A})=T^{-1}\sum_{t=1}^{T}\|\cm{Y}_t-\langle\cm{A},\cm{Y}_{t-1}\rangle\|_\text{F}^2$ is the quadratic loss function. In \eqref{eq:objective}, the regularizer $\|\cm{A}\|_{\text{SSN}}$ involves $2^{d-1}$ nuclear norms $\|\cm{A}_{[I_k]}\|_\textup{nuc}$, which are challenging to handle at the same time. A similar difficulty also occurs in low-rank tensor completion, for which \citet{gandy2011tensor}  applied the alternating direction method of multipliers (ADMM) algorithm \citep{boyd2011distributed} to efficiently separate the different nuclear norms. Borrowing the idea of \citet{gandy2011tensor}, we develop an ADMM algorithm for the miminization of \eqref{eq:objective}.

\begin{algorithm}[!htp]
	\caption{ADMM algorithm for (T)SSN estimator \label{alg:SSN}}
	Initialize: $\cm{C}_k^{(0)}$, $\cm{W}_k^{(0)}=\cm{A}^{(0)}=\cm{\widehat{A}}_{\text{MN}}$, for $k=1,\dots, 2^{d-1}$, threshold parameter $\gamma$\\
	\textbf{for} $j\in\{0,1,\dots, J-1\}$ \textbf{do}\\
	\hspace*{1cm} $\cm{A}^{(j+1)}\leftarrow\arg\min\Big\{\mathcal{L}_T(\cm{A})+\sum_{k=1}^{2^{d-1}}\rho\|\cm{A}-\cm{W}_k^{(j)}+\cm{C}_k^{(j)}\|_{\text{F}}^2\Big\}$\\
	\hspace*{1cm} \textbf{for} $k\in\{1,2,\dots, 2^{d-1}\}$ \textbf{do}\\
	\hspace*{2cm} $\cm{W}_k^{(j+1)}\leftarrow\arg\min\Big\{\rho\|\cm{A}^{(j+1)}-\cm{W}_k+\cm{C}_k^{(j)}\|_{\text{F}}^2+\lambda_{\text{SSN}}\|(\cm{W}_k)_{[I_k]}\|_\textup{nuc}\Big\}$\\
	\hspace*{2cm} $\cm{C}^{(j+1)}_k\leftarrow\cm{C}_k^{(j)}+\cm{A}^{(j+1)}-\cm{W}^{(j+1)}_k$\\
	\hspace*{1cm} \textbf{end for}\\
	\textbf{end for}\\
	$\cm{\widehat{A}}_{\text{SSN}}\leftarrow \cm{A}^{(J)}$\\
	\textbf{for} $i\in\{1,2,\dots, 2d\}$ \textbf{do}\\
	\hspace*{1cm}  $\bm{\widetilde{U}}_i \leftarrow \text{Truncated\_SVD}((\cm{\widehat{A}}_{\text{SSN}})_{(i)}, \gamma)$\\
	\textbf{end for}\\
	$\cm{\widetilde{G}}\leftarrow\cm{\widehat{A}}_{\text{SSN}}\times_{i=1}^{2d}\bm{\widetilde{U}}_i^\top$\\ $\cm{\widehat{A}}_{\text{TSSN}}\leftarrow\cm{\widetilde{G}}\times_{i=1}^{2d}\bm{\widetilde{U}}_i$ 
\end{algorithm}

To separate the $2^{d-1}$ nuclear norms in $\|\cm{A}\|_{\text{SSN}}$, for each $\cm{A}_{[I_k]}$, we introduce a different dummy variable $\cm{W}_k$ as a surrogate for $\cm{A}$, where $k=1,\dots,2^{d-1}$. Then the augmented Lagrangian is
\begin{equation*}
	\mathcal{L}(\cm{A},\cm{W},\cm{C})=\mathcal{L}_T(\cm{A})+\sum_{k=1}^{2^{d-1}}\Big[\lambda_{\text{SSN}}\|(\cm{W}_k)_{[I_k]}\|_\textup{nuc}+2\rho\langle\cm{C}_k,\cm{A}-\cm{W}_k\rangle+\rho\|\cm{A}-\cm{W}_k\|_{\text{F}}^2\Big], 
\end{equation*}
where $\cm{C}_k$ are the Lagrangian multipliers, for $k=1,\dots,2^{d-1}$, and $\rho$ is the regularization parameter. Then we can iteratively update $\cm{A}, \cm{W}_k$ and $\cm{C}_k$ by the ADMM, as shown in Algorithm \ref{alg:SSN}. 

In Algorithm \ref{alg:SSN},  the $\cm{A}$-update step is an $\ell_2$-regularized least squares problem. Similarly to \citet{gandy2011tensor}, the $\cm{W}_k$-update step can be solved by applying the explicit  soft-thresholding operator to the singular values of $(\cm{A}+\cm{C}_k)_{[I_k]}$. Both subproblems have close-form solutions. Thus,  the miminization of \eqref{eq:objective} can be solved efficiently. 

\section{Interesting Special Cases of the LRTAR Model}
\label{append:discuss}

We discuss two special cases of the proposed LRTAR model and their connections with the  matrix autoregressive  model in \cite{chen2018autoregressive} and the tensor factor model in \cite{chen2019factor}.

\begin{example}\label{ex1}
	For simplicity, we first consider the case with $d=2$, so $\cm{Y}_t\equiv \bm{Y}_t, \cm{E}_t\equiv \bm{E}_t\in\mathbb{R}^{p_1\times p_2}$ are matrices. Then the VAR representation in \eqref{eq:VAR_rep} becomes 
	\begin{equation}\label{eq:VAR_mat}
		\textup{vec}(\bm{Y}_t)=(\bm{U}_4\otimes\bm{U}_3)\cm{G}_{[\{3,4\}]}(\bm{U}_2^\top\otimes\bm{U}_1^\top)\textup{vec}(\bm{Y}_{t-1})+\textup{vec}(\bm{E}_t),
	\end{equation}
	and the low-dimensional representation in \eqref{eq:matfactor} becomes
	\begin{equation*}
		\bm{U}_{3}^\top\bm{Y}_t\bm{U}_{4} =\left \langle\cm{G},\bm{U}_{1}^\top\bm{Y}_{t-1}\bm{U}_{2}\right \rangle+\bm{U}_{3}^\top\bm{E}_t\bm{U}_{4},
	\end{equation*}
	where $\cm{G}\in\mathbb{R}^{r_1\times \cdots \times r_4}$. It is interesting to compare this model with the matrix autoregressive (MAR) model in \cite{chen2018autoregressive} and \cite{hoff15}, which is defined by 
	\begin{equation}
		\label{eq:bilinear}
		\bm{Y}_t=\bm{B}_1\bm{Y}_{t-1}\bm{B}_2^\top+\bm{E}_t,
	\end{equation}
	where $\bm{B}_1\in\mathbb{R}^{p_1\times p_1}$ and $\bm{B}_2\in\mathbb{R}^{p_2\times p_2}$, whose vector form is 
	\begin{equation}\label{eq:bilinear_vec}
		\textup{vec}(\bm{Y}_t)=(\bm{B}_2\otimes \bm{B}_1)\textup{vec}(\bm{Y}_{t-1})+\textup{vec}(\bm{E}_t).
	\end{equation}
	It can be easily seen that if $r_1=r_3=p_1$, $r_2=r_4=p_2$, $\bm{U}_{3}=\bm{I}_{p_1}$, $\bm{U}_{4}=\bm{I}_{p_2}$, and $\cm{G}_{[\{3,4\}]}=(\bm{B}_2\otimes \bm{B}_1)(\bm{U}_2\otimes \bm{U}_1)$, then \eqref{eq:VAR_mat} becomes exactly  \eqref{eq:bilinear_vec}. Thus, the MAR model in \eqref{eq:bilinear} can be viewed as a special case of the proposed model without reducing dimensions $p_i$'s to $r_i$'s and without transforming $\bm{Y}_t$; see Figure \ref{fg:Fig_TAR} for an illustration. 
	The above comparison also applies to the general case  with $d\geq 3$. The tensor version of the MAR model is considered in \cite{hoff15} and is defined as
	\begin{equation}\label{eq:multilinear}
		\cm{Y}_t=\cm{Y}_{t-1}\times_{i=1}^d\bm{B}_i+\cm{E}_t,
	\end{equation}
	where $\bm{B}_i\in\mathbb{R}^{p_i\times p_i}$ for $i=1,\dots,d$. We call \eqref{eq:multilinear} the multilinear tensor autoregressive (MTAR) model.
	Note that its vector form is
	\begin{equation}\label{eq:multilinear_vec}
		\textup{vec}(\cm{Y}_t)=(\bm{B}_d\otimes \cdots \otimes\bm{B}_1)\textup{vec}(\cm{Y}_{t-1})+\textup{vec}(\cm{E}_t).
	\end{equation}
	Similarly, \eqref{eq:multilinear_vec} is a special case of \eqref{eq:VAR_rep} with $r_i=r_{d+i}=p_i$,  $\bm{U}_{d+i}=\bm{I}_{p_i}$, for $i=1,\dots, d$, and
	$\cm{G}_{[S_2]}=(\otimes_{i\in S_1}\bm{B}_i)(\otimes_{i\in S_1}\bm{U}_i)$. Obviously, the number of unknown parameters in the MTAR model, $\sum_{i=1}^{d}p_i^2$, is much larger than that of the proposed model as shown in \eqref{eq:dim}. Also note that  \cite{chen2018autoregressive} focuses on the low-dimensional estimation and its asymptotic theory, while \cite{hoff15} considers a Bayesian estimation method. 
\end{example}

\begin{figure}
	\centering
	\includegraphics[width=0.95\textwidth]{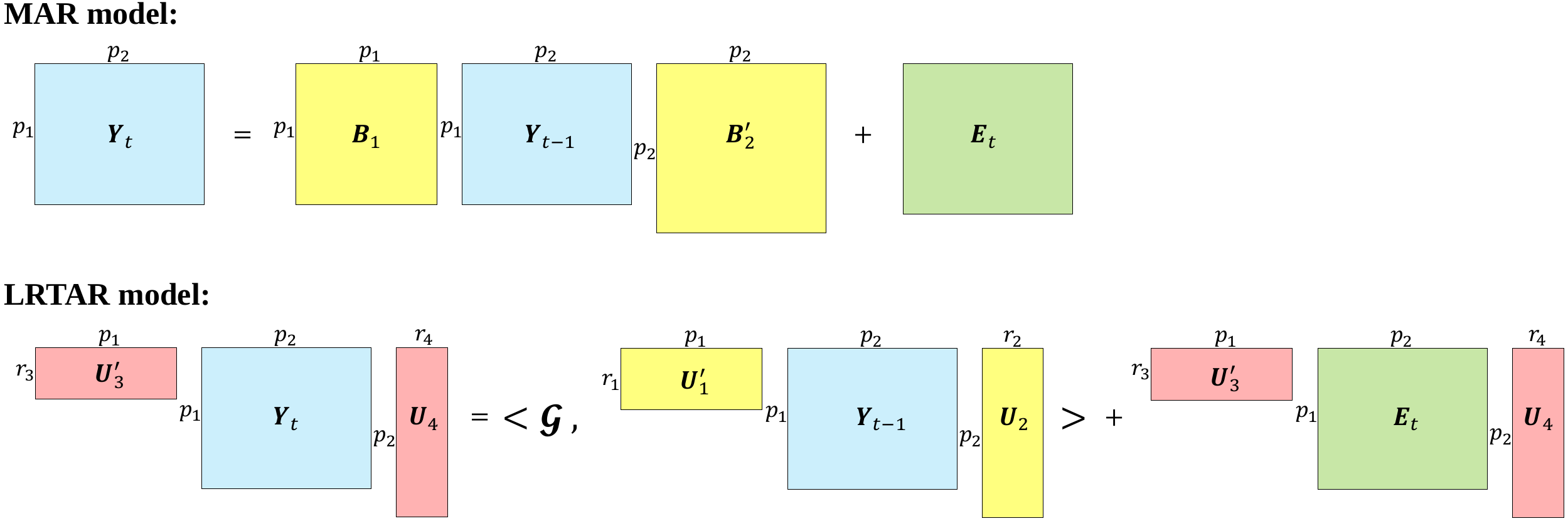}
	\caption{Illustration of the MAR model and the proposed LRTAR model in the case of $d=2$.	\label{fg:Fig_TAR}}
\end{figure}

\begin{example}\label{ex2}
	In the special case where $\bm{U}_{d+i}=\bm{U}_{i}$ and $r_{d+i}=r_i$ for $i=1,\dots, d$, the proposed model may be understood from the perspective of dynamic factor modeling \citep{SW11, BW16} for tensor-valued time series. Specifically, consider the following model:
	\begin{equation}\label{eq:dfm_tensor}
		\cm{Y}_t=\cm{F}_t\times_{i=1}^d\bm{U}_i, \quad \cm{F}_t=\langle\cm{G},\cm{F}_{t-1}\rangle+\cm{H}_t,
	\end{equation}
	where $\cm{Y}_t\in\mathbb{R}^{p_1\times\cdots\times p_d}$ is the observed tensor-valued time series,  $\cm{F}_t\in\mathbb{R}^{r_1\times\cdots\times r_d}$ represents $\prod_{i=1}^d r_i$ factors, and $\bm{U}_i\in\mathbb{R}^{p_i\times r_i}$ are orthonormal matrices for $i=1,\dots, d$. Here  $\cm{F}_t$ follows the tensor autoregression (TAR) with transition tensor $\cm{G}\in \mathbb{R}^{r_1\times\cdots\times r_d\times r_1\times\cdots\times r_d}$ and random error $\cm{H}_t$.
	Note that \eqref{eq:dfm_tensor} can be rewritten as
	\[
	\cm{Y}_t =\left \langle\cm{G}\times_{i=1}^d\bm{U}_{i} \times_{i=d+1}^{2d}\bm{U}_{i},\cm{Y}_{t-1}\right \rangle+\cm{H}_t\times_{i=1}^d\bm{U}_{i}.
	\]
	Thus,  model \eqref{eq:dfm_tensor} is a special case of the proposed model with $\bm{U}_{d+i}=\bm{U}_{i}$ and $r_{d+i}=r_i$ for $i=1,\dots, d$, and $\cm{E}_t=\cm{H}_t\times_{i=1}^d\bm{U}_{i}$.  \citet{chen2019factor} introduces the tensor factor model in the form of $\cm{Y}_t=\cm{F}_t\times_{i=1}^d\bm{U}_i+\cm{E}_t$  without an explicit modeling of the latent factors $\cm{F}_t$. Hence, model \eqref{eq:dfm_tensor} may be regarded as a special tensor factor model with autoregressive dynamic factors, but without any random error in the model equation of $\cm{Y}_t$.
\end{example}

\newpage
\bibliography{mybib}
\bibliographystyle{apalike}
\end{document}